\documentclass[12pt, a4paper]{article}  
\usepackage[margin=1.25in]{geometry}
\usepackage{amsmath,amssymb,amsthm,mathtools}
\usepackage{natbib}
\usepackage{graphicx}
\usepackage[nodisplayskipstretch]{setspace} %
\usepackage[inline]{enumitem}
\usepackage{fullpage,rotating,afterpage}
\usepackage{microtype} %
\usepackage{longtable}
\usepackage{subfigure}
\usepackage{ragged2e}

\usepackage{array}
\usepackage{hyperref}
\hypersetup{
	colorlinks   = true, 
	urlcolor     = blue,
	linkcolor    = black, 
	citecolor   =  black 
}

\usepackage[T1]{fontenc} 
\usepackage{tikz}
\usepackage{pdflscape}
\usepackage{booktabs}
\usepackage[justification=justified]{caption}
\usepackage[flushleft]{threeparttable}
\usepackage{multirow}
\usepackage{bm}
\usepackage{xcolor}

\usepackage{xr}
\newcommand{\E}{\mathbb{E}}
\newcommand{\V}{\mathbb{V}}
\newcommand{\IR}{\mathbb{R}}
\newcommand{\IN}{\mathbb{N}}
\newcommand{\IP}{\mathbb{P}}

\newcommand{\1}[1]{\mathbf{1}\{#1\}}

\newcommand{\Cov}{\operatorname{Cov}}

\newcommand{\sgn}{\operatorname{sign}}

\newcommand{\sumI}{\sum_{i \in I_s}}
\newcommand{\Sum}{\sum_{i=1}^n}
\newcommand{\wh}{\widehat}
\newcommand{\wt}{\widetilde}
\newcommand{\mc}{\mathcal}

\newcommand{\sss}{\scriptscriptstyle}
\newcommand{\X}{\mathcal{X}}
\newcommand{\Xn}{\mathbb{X}_n}
\newcommand{\Xh}{\mathcal{X}_h}
\newcommand{\Zh}{\mathcal{Z}_h}

\newcommand{\Tn}{\mathcal{T}_n}

\newcommand{\hs}{\widehat{\text{se}}}
\newcommand{\hmu}{\wh\ad}
\newcommand{\ad}{\eta}
\newcommand{\Class}{\mc E}
\newcommand{\ClassVar}{\mc V}
\newcommand{\htau}{\wh{\tau}}
\newcommand{\htheta}{\wh{\theta}}
\newcommand{\hse}{\wh{\text{se}}}
\newcommand{\Sumr}{\frac{1}{R} \sum_{j\in \mc R_i}}
\newcommand{\Sumn}{\sum_{i \in [n]}}

\DeclareMathOperator*{\argmin}{argmin}

\newtheorem{proposition}{Proposition} 

\newtheorem{theorem}{Theorem}

\newtheorem{lemma}{Lemma}
\newtheorem{assumption}{Assumption}
\theoremstyle{definition}

\numberwithin{equation}{section}

\usepackage{verbatim,color}

\def\boxit#1{\vbox{\hrule\hbox{\vrule\kern6pt
			\vbox{\kern6pt#1\kern6pt}\kern6pt\vrule}\hrule}}

\usepackage{dcolumn} 
\newcolumntype{d}[1]{D{.}{.}{#1}}

\newcommand{\captionfonts}{\small}

\makeatletter  
\long\def\@makecaption#1#2{%
	\vskip\abovecaptionskip
	\sbox\@tempboxa{{\captionfonts #1: #2}}%
	\ifdim \wd\@tempboxa >\hsize
	{\captionfonts #1: #2\par}
	\else
	\hbox to\hsize{\hfil\box\@tempboxa\hfil}%
	\fi
	\vskip\belowcaptionskip}
\makeatother   

\usepackage{titlesec}
\titleformat{\section}[block]{\centering\normalfont}{\thesection.}{0.5em}{\uppercase }
\titleformat{\subsection}[runin]{\normalfont}{\thesubsection.}{0.4em plus .1em minus .2em}{\bfseries}[.]
\titleformat{\subsubsection}[runin]{\normalfont}{\thesubsubsection.}{0.4em plus .1em minus .2em}{\it}[.]
\titlespacing*\section{0pt}{18pt plus 4pt minus 2pt}{5pt plus 2pt minus 2pt}
\titlespacing*\subsection{0pt}{10pt plus 2pt minus 1pt}{5pt plus 2pt minus 2pt }
\titlespacing*\subsubsection{0pt}{4pt plus 1pt minus 1pt}{5pt plus 2pt minus 2pt}

\usepackage{chngcntr}
\usepackage{apptools}
\AtAppendix{\counterwithin{lemma}{section}}
\AtAppendix{\counterwithin{assumption}{section}}
\AtAppendix{\counterwithin{corollary}{section}}
\AtAppendix{\counterwithin{theorem}{section}}
\AtAppendix{\counterwithin{proposition}{section}}
\AtAppendix{\counterwithin{table}{section}}
\AtAppendix{\counterwithin{figure}{section}}
\AtAppendix{\counterwithin{remark}{section}}

\makeatletter
\def\mythanks#1{%
	\protected@xdef \@thanks {\@thanks \protect \footnotetext [\the \c@footnote ]{#1}}%
}
\makeatother

\title{Flexible Covariate Adjustments in  Regression Discontinuity Designs\mythanks{First version: July 16, 2021. 
		This version: \today.
		Author contact information: 
		Claudia Noack, Department of Economics, University of Bonn. E-Mail: claudia.noack@uni-bonn.de. Website: https://claudianoack.github.io.   
		Tomasz Olma, Department of Statistics, Ludwig Maximilian University of Munich. E-Mail: t.olma@lmu.de. Website: https://tomaszolma.github.io.   
		Christoph Rothe, Department of Economics, University of Mannheim.
		E-Mail: rothe@vwl.uni-mannheim.de. Website: http://www.christophrothe.net. }}

\author{Claudia Noack \and Tomasz Olma \and Christoph Rothe}
\date{}

\begin{document}
	\pagestyle{plain}
	\onehalfspacing
	
	\maketitle 
	
	\begin{abstract}
		Empirical regression discontinuity (RD) studies often include covariates in their specifications to increase the precision of their estimates. In this paper, we propose a novel class of estimators that use such covariate information more efficiently than existing methods and can accommodate many covariates. Our estimators are simple to implement and involve running a standard RD analysis after subtracting a function of the covariates from the original outcome variable. We characterize the function of the covariates that minimizes the asymptotic variance of these estimators. We also show that the conventional RD framework gives rise to a special robustness property which implies that the optimal adjustment function can be estimated flexibly  via modern machine learning techniques without affecting the first-order properties of the final RD estimator. We demonstrate our methods' scope for efficiency improvements by reanalyzing data from a large number of recently published empirical studies.
	\end{abstract}
	
	\newpage

	\doublespacing
	\section{Introduction}

Regression discontinuity (RD) designs are widely used for estimating causal effects from observational data in economics and other social sciences. The sharp RD design exploits the fact that in many contexts a unit's treatment status is determined by whether its realization of a running variable exceeds some known cutoff value. For example, students might qualify for a scholarship if their GPA is above some threshold.
Under continuity conditions on the distribution of potential outcomes, the average treatment effect at the cutoff is  identified by the jump in the conditional expectation of the outcome given the running variable at the cutoff. Estimation and inference methods based on local linear regression are widely used and their properties are by now well understood \citep[e.g.,][]{hahn2001identification, imbens2012optimal, calonico2014robust, armstrong2020simple}.
	
An RD analysis generally does not require data beyond the outcome and the running variable, but additional covariate information can be used to reduce the variance of empirical estimates. A  common approach is to include the covariates linearly and  without separate localization in a local linear RD regression \citep{calonico2019regression}. 
This conventional linear adjustment estimator is consistent without functional form assumptions on the underlying conditional expectations if the covariates are unaffected by the treatment in some appropriate sense. It does not  exploit the available covariate information efficiently though, and it is also not well-suited for settings with many covariates.

In this paper, we propose a novel class of flexible covariate-adjusted RD estimators. Our approach involves running a standard local linear RD regression after subtracting an (estimated) function of the covariates from the original outcome variable. We characterize the function that leads to the RD estimator with the smallest asymptotic variance in this class and show how this function can be estimated with modern machine learning techniques. We also show that the usual RD framework gives rise to a special robustness property which implies that our final estimator is ``very insensitive'' to estimation uncertainty about the estimated adjustment function. One can therefore directly use existing methods for bandwidth choice and inference with our adjusted outcome variable.

To motivate our proposed procedure,  let $Y_i$ and $Z_i$ denote the outcome and covariates, respectively, of observational unit $i$. Then the conventional linear adjustment RD estimator is  asymptotically equivalent to a local linear RD estimator with the modified outcome variable $Y_i -Z_i^\top \gamma_0$, where $\gamma_0$ is a vector of projection coefficients.  We consider generalizations of such estimators which replace the linearly adjusted outcome with a flexibly adjusted outcome of the form $Y_i - \ad(Z_i)$, for some generic function $\ad$. Such estimators are easily seen to be consistent for \emph{any} fixed $\ad$ if the distribution of the covariates varies smoothly around the cutoff in some appropriate sense. 
This condition is stronger than the notion needed for linear adjustments to be consistent but still in line with the notion of the covariates being unaffected by the treatment in a causal sense.\footnote{Note that the adjustment function must be the same for observations on either side of the cutoff. If the adjustment function was   different on either side of the cutoff, this would generally yield inconsistent RD estimates \citep{calonico2019regression}.}
We  show that the asymptotic variance in this class of estimators is minimized over $\ad$ by the average of the two conditional expectations of the outcome variable given the running variable and the covariates just above and below the cutoff. This optimal adjustment function $\ad_0$ is generally nonlinear and unknown in practice but can be estimated from the data.

	Our proposed estimators of the RD parameter hence take the form of a local linear RD regression with  generated outcome $Y_i -\wh\ad(Z_i)$, where $\wh\ad$ is some estimate of $\ad_0$ obtained in a preliminary stage. We implement such estimators with cross-fitting, a common form of sample splitting, which allows us to theoretically accommodate a wide range of estimators of the optimal adjustment function. In particular, we can allow for the use of  modern machine learning methods like lasso regression, random forests,  deep neural networks, or ensemble combinations thereof, to estimate the optimal adjustment function. In low-dimensional settings, researchers can also use classical nonparametric approaches like  local polynomials or series regression, or estimators based on parametric specifications.\footnote{Note that our method is not a straightforward application of standard debiased/double machine learning (DML) methods \citep[e.g.,][]{chernozhukov2018double}. Major differences between conventional DML settings and ours include that one first needs to derive the optimal adjustment in our setting, that there is a nonparametric target parameter instead of a parametric one, that we have a single nuisance function instead of two, and that a particular global robustness property holds in our setting which is related to but stronger  than the usual local Neyman orthogonality condition.}

	Our theory  does not require that $\ad_0$ is consistently estimated to conduct valid inference on the RD parameter. We only require that in large samples the first-stage estimates concentrate in a mean-square sense around some deterministic function $\bar\ad$, which could in principle be different from $\ad_0$. The rate of this convergence can be arbitrarily slow. Our setup allows  this kind of potential misspecification because our proposed RD estimators are highly insensitive to estimation errors in the preliminary stage. Specifically, they are constructed as sample analogues of a moment function that
	contains $\ad_0$ as a nuisance function, but does not vary with it:
	as discussed above, our parameter of interest is equal to the jump in the conditional expectation of $Y_i - \ad(Z_i)$ given the running variable at the cutoff for \emph{any} fixed function $\ad$.
	This insensitivity property is related to  Neyman orthogonality, which features prominently in many modern two-stage estimation methods \citep[e.g.,][]{chernozhukov2018double}, but it is a global (with respect to the nuisance function) rather than a local  property and therefore has  substantially stronger implications.\footnote{A moment function is Neyman orthogonal if its first functional derivative with respect to the nuisance function is zero. In contrast,  the (conditional) moment function on which our estimates are based is fully invariant with respect to the nuisance function. \citet{chernozhukov2018double} give examples of setups with (unconditional) moment functions in which such a global insensitivity property occurs. These include optimal instrument problems, certain partial linear models, and treatment effect estimation under unconfoundedness with a known propensity score. This property also occurs  generally if one of the two nuisance functions in a doubly robust  moment condition \citep[cf.][]{robins01rotnitsky}  is known.}

Our theoretical analysis shows that, under the conditions outlined above, our proposed RD  estimator is first-order asymptotically equivalent to a local linear ``no covariates'' RD estimator with $Y_i -\bar\ad(Z_i)$ as the dependent variable. This result is then used to study our estimator's asymptotic bias and variance, and to derive an asymptotic normality result. 
The asymptotic variance of our estimator depends on the function $\bar\ad$ and achieves its minimum value if $\bar\ad = \ad_0$ (that is, if $\ad_0$ is consistently estimated in the first stage), but the asymptotic variance can be estimated consistently irrespective of whether or not that is the case. Because  our RD estimator does not require a particular rate of convergence for the first-step estimate of $\ad_0$, it can be seen as shielded from the ``curse of dimensionality'' to some degree. It can  hence be  expected to perform well in settings with many covariates.

We also show that widely used methods for bandwidth choice and construction of confidence intervals in settings without covariates remain valid in our framework: empirical researchers can apply these methods to a data set in which the outcome $Y_i$ is replaced with the generated outcome $Y_i -\wh\ad(Z_i)$ and simply ignore the fact  that $\wh\ad$ has been estimated. 
We show that these results also extend to fuzzy RD designs in a straightforward manner. 
Our approach can therefore easily be integrated into existing software packages.\footnote{The DoubleML package \citep{DoubleML2022} contains a dedicated implementation in Python.}
	
Practitioners should note that our approach requires additional assumptions about the smoothness of the covariate distribution relative to RD estimators that do not adjust for covariates or only adjust for covariates linearly. However, these assumptions follow naturally from the notion that the covariates are unaffected by the treatment in a causal sense and should therefore be plausible in many empirical settings (we do not recommend using covariate adjustments in RD settings in which there is a plausible causal relationship between the treatment and the covariates). 
The asymptotic bias of our proposed RD estimator depends in principle on the  adjustment function, and it could   be either larger or smaller than that of the unadjusted or the linearly adjusted RD estimator. Under our assumptions on the covariate distribution, however, the three procedures turn out to have the same leading bias in the usual pointwise asymptotic framework. In finite samples, we therefore expect the differences between the biases of the three procedures to be relatively small. This observation motivates our approach to choose the  estimated adjustment function $\wh\ad$ such that it targets the variance-minimizing adjustment function $\eta_0$.

	Our theoretical results are qualitatively similar to those that have been obtained for efficient influence function (EIF) estimators of the population average treatment effect (ATE) in simple randomized experiments with known and constant propensity scores \citep[e.g.,][]{wager2016high}. Such parallels arise because EIF estimators are also based on a moment function that is globally invariant with respect to a nuisance function. In fact, we argue that our RD estimator is in many ways a direct analogue of the EIF estimator, and that its asymptotic variance  under the optimal adjustment function is  similar in structure to the semiparametric efficiency bound for the ATE in simple randomized experiments.

 Our proposed flexible covariate adjustments can lead to substantial efficiency gains in practice. To illustrate this, we collect data from empirical papers recently published in leading  journals that use RD estimation with covariates.    In total, we reanalyze 56  specifications from 16 papers, and study how different types of covariate adjustments affect  confidence interval lengths. Including covariates in the RD regression does not meaningfully reduce the length of the confidence intervals in about half of the specifications we consider irrespective of the specific method used for adjustment, but  our proposed flexible adjustments also achieve a reduction of more than 30\% in one setting. To put this into perspective, obtaining this reduction would require roughly increasing the sample size by a factor of 2.4 if the covariates were not used. We also observe that  linear adjustments alone are often unable to exhaust all the available covariate information: the largest reduction in the confidence interval length from using our flexible adjustment relative to linear adjustments exceeds 20\%. 

We also conduct simulations based on the data set from one of the  papers from our empirical literature survey. In order to cover all types of settings from our empirical literature survey,  we consider simulation setups of large sample sizes and a moderate number of covariates as well as   small sample sizes and a varying number of covariates. Our proposed RD estimators perform very well in all these settings, in the sense that their standard errors are close to their standard deviations and the associated confidence intervals have simulated coverage rates close to the nominal one.

	\subsection*{Related Literature} 
	Our paper contributes to an extensive literature on estimation and inference in RD designs; see, e.g., \citet{imbens2008regression} and \citet{lee2010regression} for surveys, and \citet{cattaneo2019practical} for a textbook treatment. Different ad-hoc methods for incorporating covariates into an RD analysis have long been used in empirical economics  \citep[see, e.g.,][Section 3.2.3]{lee2010regression}. Following \citet{calonico2019regression}, it has become common practice to include covariates without localization into the usual local linear regression estimator. Our approach nests this estimator as a special case, but is generally more efficient. Other closely related papers are \citet{kreiss2021regression} and \citet{arai2024regression}, who extend the approach in \citet{calonico2019regression} to settings with high-dimensional covariates under sparsity conditions using the lasso. These estimators can be understood as ``one-step'' versions of our ``two-step'' procedure with a specific method for estimating the adjustment function. The approach in this paper can exploit other machine learning methods, can incorporate different variants of (post-) lasso adjustments that can be more stable in finite samples (see ``global adjustments'' in Section~\ref{sec:est_adj_function}), and, through the use of cross-fitting, yields more precise standard errors.	\citet{froelich2019including} propose to incorporate covariates into an RD analysis in a fully nonparametric fashion, but their approach is generally affected  by the curse of dimensionality, and is thus unlikely to perform well in practice.

	Our results are also related in a more general sense to the vast literature on two-step estimation problems with infinite-dimensional nuisance parameters \citep[e.g.,][]{andrews1994asymptotics,newey1994variance}, especially the recent strand that exploits Neyman orthogonal (or debiased) moment functions and  cross-fitting 
	(e.g., Section~25.8 of \citealp{van2000asymptotic}; \citealp{belloni2017program}; \citealp{chernozhukov2018double}).
	The latter literature focuses mostly on regular (root-$n$ estimable) parameters, while our RD treatment effect is a non-regular (nonparametric) quantity. Some general results on non-regular estimation based on orthogonal moments are derived in~\citet{chernozhukov2019double}, and specific results for estimating  conditional average treatment effects in models with unconfoundedness are given, for example, in \citet{kennedy2017nonparametric}, \citet{kennedy2020optimal}  and \citet{fan2020estimation}. Our results are qualitatively different
	because, as explained above, our estimator is based on a moment function that satisfies a property that is stronger than Neyman orthogonality.
	
Finally, our work is linked to the literature on inference in randomized experiments  with covariates \citep[e.g.,][]{ Freedman2008regressionadjustmentsseveraltreatments, freedman2008regression, lin2013agnostic, wager2016high, lei2021regression, chiang2023regression, chang2024exact}.

	\subsection*{Plan of the Paper} The remainder of this paper is organized as follows. 
	In Section~\ref{sec:Setup}, we introduce the setup and review existing procedures. In Section~\ref{sec:Covariate_adj}, we describe our proposed covariate-adjusted RD estimator, and we present our main theoretical results in Section~\ref{sec:Results}. Further extensions are discussed in Section~\ref{sec:Implementation}. 
We present our empirical results in Section~\ref{sec:practical_performance} and simulation results in Section~\ref{sec:Simulations}. Section~\ref{sec:Conclusions} concludes. The proofs of our main results are given in Appendix~\ref{A:main}. Appendix~\ref{sec:Inference} formally studies the proposed inference procedures and Appendix~\ref{sec:appendix_empirical_results} gives details on our literature survey.
	The Online Supplement contains additional empirical and simulation results.

	\section{Setup and Preliminaries}\label{sec:Setup}
		
	\subsection{Model and Parameter of Interest} We begin by considering sharp RD designs. The data $\{W_i\}_{i\in [n]}= \{(Y_i,X_i,Z_i)\}_{i\in [n]}$, where  $[n]= \{1, \dots, n\}$,
	are an i.i.d.\ sample of size $n$ from the distribution of $W=(Y,X,Z)$.
	Here, $Y_i\in\IR$ is the outcome variable, $X_i\in\IR$ is the running variable, and $Z_i \in \mathbb{R}^d$ is a (possibly high-dimensional) vector of  covariates.\footnote{Throughout the paper, we assume that the distribution of the running variable $X_i$ is fixed, but we allow the conditional distribution of $(Y_i,Z_i)$ given $X_i$ to change with the sample size in our asymptotic analysis.
		In particular, we allow the dimension of $Z_i$ to grow with $n$ in order to accommodate high-dimensional settings, but we generally leave such dependence on $n$ implicit in our notation.
	} Units receive the treatment if and only if the running variable exceeds a known threshold, which we normalize to zero without loss of generality. We denote the treatment indicator by $T_i$, so that $T_i=\1 {X_i \geq 0}$.
	The parameter of interest is the height of the jump in the conditional expectation function of the observed outcome variable given the running variable at zero:
	\begin{equation}\label{eq:estimand_limits}
		\tau = \E[Y_i|X_i=0^+] - \E[Y_i|X_i=0^-],
	\end{equation}
	where we use the notation that $f(0^+) =\lim_{x\downarrow 0}f(x)$ and $f(0^-) =\lim_{x\uparrow 0}f(x)$ are the right and left limit, respectively, of a generic function $f(x)$ at zero. In a potential outcomes framework, the parameter $\tau$ coincides with the average treatment effect of units at the cutoff
	under certain continuity conditions \citep{hahn2001identification}.
	
	\subsection{Standard RD Estimator}\label{sec:standardrdestimator}
	Without the use of covariates, the parameter $\tau$ is typically estimated by running separate local linear regressions \citep{fan1996local} on each side of the cutoff. That is, the baseline no covariates RD estimator takes the form
	\begin{equation}\label{eq:standard_RD}
		\wh\tau_{\sss base}(h) =  e_1^\top\argmin_{\beta\in\IR^4} \sum_{i=1}^n K_h(X_i) (Y_i -  S_i^\top\beta )^2,
	\end{equation}
	where $S_i =(T_i, X_i, T_i X_i,1)^\top$, $K_h(v)=K(v/h)/h$ with $K(\cdot)$ a  kernel function  and $h>0$ a bandwidth, and $e_1 = (1,0,0,0)^\top$ is the first unit vector.  This estimator is a linear smoother that can also be written as a weighted sum of the realizations of the outcome variable,
	$$ \wh\tau_{\sss base}(h) =  \Sum w_i(h) Y_i,$$
	where the $w_i(h)$ are  local linear regression weights that depend on the data through the realizations of the running variable only; see Appendix~\ref{A:sec:weights} for an explicit expression.
	
	Under standard conditions  \citep[e.g.][]{hahn2001identification}, which include a continuously distributed running variable and that the bandwidth  $h$ tends to zero at an appropriate rate,  the estimator $\wh\tau_{\sss base}(h)$ is approximately normally distributed in large samples under conventional pointwise asymptotics, with  bias of order $h^2$ and   variance of order $(nh)^{-1}$:
	\begin{align}\label{eq:dist_baseline}
		\wh{\tau}_{\sss base}(h) \stackrel{a}{\sim}N\left(\tau + h^2  B_{\sss base},(nh)^{-1}V_{\sss base}\right).
	\end{align}
	Here ``$\stackrel{a}{\sim}$'' indicates a finite-sample distributional approximation justified by an asymptotic normality result, and the bias and variance terms are given, respectively, by
	\begin{align*}
		B_{\sss base}&=\frac{\bar{\nu}}{2} \left(\partial_x^2 \E[Y_i|X_i=x]|_{x=0^+}-\partial_x^2 \E[ Y_i|X_i=x]|_{x=0^-}\right) \textnormal{ and}\\
		V_{\sss base} &=\frac{\bar{\kappa}}{f_X(0)}(\V[Y_i|X_i=0^+]+\V[Y_i|X_i=0^-]).
	\end{align*}
	The terms $\bar{\nu}$ and $\bar{\kappa}$ are kernel constants, defined as $\bar{\nu}= (\bar{\nu}_2^2 - \bar{\nu}_1 \bar{\nu}_{3})/(\bar{\nu}_2\bar{\nu}_0-\bar{\nu}_1^2)$ for $\bar{\nu}_{j}= \int_{0}^{\infty}v^j K(v)dv$ and $\bar{\kappa}= \int_0^\infty(K(v)(\bar{\nu}_1 v - \bar{\nu}_2))^2dv/ (\bar{\nu}_2\bar{\nu}_0-\bar{\nu}_1^2)^2$, and $f_X$ denotes the density of $X_i$. 
	Practical methods for inference based on approximations like~\eqref{eq:dist_baseline} are discussed, for instance, by \citet{calonico2014robust} and \citet{armstrong2020simple}. 
		
	\subsection{Conventional Linear Adjustment Estimator} If covariates are available, they can be used to improve the accuracy of empirical RD estimates.\footnote{Throughout the paper, we focus on settings in which covariates are included to improve estimation efficiency and not to restore identification of the RD parameter by making the design more plausible. } Arguably the   most  popular strategy \citep{calonico2019regression} is to include them linearly and without kernel localization in the local linear regression~\eqref{eq:standard_RD}:
	\begin{equation}\label{eq:lin}
		\wh\tau_{\sss lin}(h) = e_1^{\top}\argmin_{\beta,\gamma} \Sum K_h(X_i) (Y_i - S_i^\top\beta- Z_i^{\top}\gamma )^2.
	\end{equation}
By simple least squares algebra, this ``linear adjustment'' estimator can   be written as a no covariates estimator with covariate-adjusted outcome $Y_i-Z_i^{\top}\widehat\gamma_h$, where $\widehat\gamma_h$ is the minimizer with respect to $\gamma$ in~\eqref{eq:lin}:
	\begin{equation*}
		\wh\tau_{\sss lin}(h) = \Sum w_i(h)(Y_i-Z_i^{\top}\widehat\gamma_h).
	\end{equation*}
The linear adjustment estimator is consistent for the RD parameter without functional form assumptions on the underlying conditional expectations if the covariates are predetermined, in the sense that they are not causally affected by the treatment, and thus their conditional expectation given the running variable varies smoothly around the cutoff. Moreover, if $\E[Z_i|X_i=x]$ is twice continuously differentiable around the cutoff, then
	\begin{align*}
		\wh\tau_{\sss lin}(h) \stackrel{a}{\sim}N\left(\tau + h^2  B_{\sss base},(nh)^{-1}V_{\sss lin}\right)
	\end{align*}
	under pointwise asymptotics and regularity conditions analogous to those for the no covariates estimator. Here the bias term $B_{\sss base}$ is the same as that of the no covariates estimator and the new variance term is
	\begin{align*}
		V_{\sss lin}=\frac{\bar{\kappa}}{f_X(0)}(\V[Y_i -Z_i^{\top}\gamma_0 |X_i=0^+]+\V[Y_i-Z_i^{\top}\gamma_0|X_i=0^-]),
	\end{align*}
	where $\gamma_0$, a non-random vector of  projection coefficients, is the probability limit of $\widehat\gamma_h$. The first-order asymptotic properties of $\wh\tau_{\sss lin}(h)$ are thus the same as that of its infeasible counterpart $\widetilde\tau_{\sss lin}(h)=\Sum w_i(h)(Y_i-Z_i^{\top}\gamma_0)$ that uses the population projection coefficients $\gamma_0$ instead of their estimates $\widehat\gamma_h$ to create the adjusted  outcome variable. 
		As $V_{\sss lin}\leq V_{\sss base}$ under standard conditions \citep[Remark~3.5]{kreiss2021regression}, including a fixed number of covariates generally increases the precision of the estimator in large samples. 	
To construct standard errors and confidence intervals, one can then use methods developed for the no covariates case, replacing the original outcome $Y_i$ with
the adjusted outcome $Y_i-Z_i^{\top}\widehat\gamma_h$ in the respective formulas \citep{calonico2019regression,armstrong2018optimal}. For instance, a nearest-neighbor standard error $\widehat{\textnormal{se}}_{\sss lin}(h)$ of $\wh\tau_{\sss lin}(h)$ is
\begin{align}\label{eq:lin_se}
\widehat{\textnormal{se}}_{\sss lin}^2(h) = \sum_{i=1}^n w_i(h)^2 \widehat\sigma_{i, \sss lin}^2, \quad \widehat\sigma_{i, \sss lin}^2 = \frac{R}{R+1}\left( (Y_i-Z_i^{\top}\widehat\gamma_h)-\frac{1}{R}\sum_{j\in\mathcal{R}_i} (Y_j-Z_j^{\top}\widehat\gamma_h) \right)^2.
\end{align}
Here $\widehat\sigma_{i, \sss lin}^2$ is an estimate of
$\sigma_{i, \sss lin}^2=\V(Y_i-Z_i^{\top}\gamma_0|X_i)$,
$R\geq 1$ is a (small) fixed integer, and $\mathcal{R}_i$ is the set that contains the indices of the $R$ nearest neighbors of unit $i$ in terms of their realization of the running variable among units on the same side of the cutoff.

	\section{Flexible Covariate Adjustments}\label{sec:Covariate_adj}

\subsection{Motivation}\label{sec:motivation}

While linear adjustments are easy to implement, they might not exploit the available covariate information efficiently. Inference might also  not be reliable with linear adjustments if the number of covariates is large relative to the effective sample size.\footnote{If there are many covariates relative to the number of observations that receive positive kernel weights in~\eqref{eq:lin}, the standard error in~\eqref{eq:lin_se} is generally downward biased. This bias occurs because the local empirical variances $\widehat\sigma^2_{i,\sss lin}$ are typically smaller than their population counterparts $\sigma^2_{i,\sss lin}$   in such cases due to overfitting. If the number of covariates exceeds the number of observations with positive kernel weights, the estimator in equation~\eqref{eq:lin} is of course not even well-defined in the first place.}
In this paper, we propose a new method to address these issues. It allows for general nonlinear covariate adjustments and can accommodate regularization methods in the estimation of the adjustment terms.

To motivate our flexible covariate adjustments, recall that the linear adjustment estimator is asymptotically equivalent to a no covariates RD estimator of the form in~\eqref{eq:standard_RD} that uses the covariate-adjusted outcome $Y_i-Z_i^{\top}\gamma_0$ instead of the original outcome $Y_i$.	
	 We generalize this by considering a class of estimators with covariate-adjusted outcomes based on potentially nonlinear adjustment functions $\ad$:
	\begin{equation}\label{eq:whtau}
		\wh\tau(h;\ad)= \Sum w_i(h) M_i(\ad), \quad M_i(\ad) = Y_i - \ad(Z_i).
	\end{equation}
We note that the adjustment function must be the same for observations on either side of the cutoff, as using different adjustments on either side of the cutoff would generally yield inconsistent RD estimates \citep{calonico2019regression}.

	If the covariates are predetermined, their conditional distribution given the running variable should vary smoothly around the cutoff. We formalize this notion by assuming  that for every adjustment function $\ad$, the function $\E[\ad(Z_i)|X_i=x]$ is twice continuously differentiable around the cutoff.\footnote{Our analysis only rules out adjustment functions that do not satisfy certain technical regularity conditions, such as functions for which the respective conditional expectation does not  exist in the first place. Assuming smoothness of $\E[\ad(Z_i)|X_i=x]$ for (essentially) all $\ad$ is stronger than only assuming smoothness of $\E[Z_i|X_i=x]$, as in \citet{calonico2019regression}. Our stronger assumption, however, is still very much in line with the notion of covariates being predetermined.} 
This assumption implies that the parameter $\tau$, defined in~\eqref{eq:estimand_limits}, can  also be expressed as:
	\begin{equation}\label{eq:EMu}
		\tau = \E[M_i(\ad)|X_i=0^+] - \E[M_i(\ad)|X_i=0^-] \text{ for all }\ad.
	\end{equation}
The estimator $\wh\tau(h;\ad)$ can thus be seen as a sample analog estimator based on the moment condition~\eqref{eq:EMu}, which   identifies $\tau$ and  is globally invariant with respect to the adjustment function $\eta$. Because of this global invariance, we expect that 
	\begin{equation}\label{eq:dist_fixedmu}
		\wh\tau(h;\ad) \stackrel{a}{\sim}N\Big(\tau + h^2B_{\sss base}, (nh)^{-1} V(\ad)\Big) \text{ for all }\ad
	\end{equation}
under standard regularity conditions.	Under these pointwise asymptotics, the bias term  in~\eqref{eq:dist_fixedmu} is again that of the baseline no covariates estimator as it does not depend on the adjustment function due to the assumed 
smoothness of $\E[\ad(Z_i)|X_i=x]$. On the other hand, the variance term in~\eqref{eq:dist_fixedmu} does depend on $\eta$, and  is given by
	\begin{align*}
		V(\ad)=\frac{\bar{\kappa}}{f_X(0)}(\V[M_i(\ad) |X_i=0^+]+\V[M_i(\ad)|X_i=0^-]).
	\end{align*}
To maximize the precision of the estimator $\wh\tau(h;\ad)$ for any particular bandwidth $h$, we want to choose $\ad$  such that $V(\ad)$ is as small as possible. Our analysis below shows that using
the equally-weighted average of the left and right limits of the ``long'' conditional expectation function $\E[Y_i|X_i=x,Z_i=z]$ at the cutoff  achieves this goal. That is, we show that $$V(\ad)\geq V(\ad_0) \textnormal{ for all } \ad,$$ where
	\begin{align}\label{eq:mun}
		\ad_0(z) = \frac{1}{2}\left( \mu_0^+(z) + \mu_0^-(z) \right), \quad \mu_0^\star(z)=\E[Y_i|X_i=0^\star,Z_i=z] \textnormal{ for } \star\in\{+,-\}.
	\end{align}
As the optimal adjustment function $\ad_0$ is generally unknown in practice, we propose to estimate the RD parameter $\tau$ by a feasible version of $\wh \tau(h;\ad_0)$.

\subsection{Proposed Estimator and its General Properties}\label{sec:Estimator} 

Our proposed estimator requires a first-stage estimate of the optimal adjustment function, which does not have to be of a particular type: practitioners can use classical nonparametric or modern machine learning methods to reduce the risk of model misspecification, or choose suitable parametric methods (conventional linear adjustments can be seen as  a special case of the latter type).
Our proposed estimator also uses cross-fitting, which is an efficient form of sample splitting that prevents overfitting of the estimated adjustment function and avoids unrealistic empirical process conditions in the theoretical analysis \citep{chernozhukov2018double}.	Specifically, our estimator is computed in two steps:
	\begin{enumerate}
		\item  Randomly split the data $\{W_i \}_{i \in [n]}$ into $S$ folds of equal size, collecting the corresponding indices in the sets $I_s$, for $s \in [S]$.  In practice, $S=5$ or $S=10$ are common choices for the number of cross-fitting folds. 
		Let $\widehat\ad(z)=\widehat\ad(z;\{W_i\}_{i \in[n]})$ be the researcher's preferred estimator of $\ad_0$, calculated on the full sample; and  let  $\wh\ad_{s}(z)=\wh\ad(z;\{W_i\}_{i \in I^c_{s}})$, for $s \in [S]$, be a version of this estimator that only uses data outside the $s$th fold.
		\item Estimate $\tau$ by computing a local linear no covariates RD estimator that uses the adjusted outcome $ M_i(\wh\ad_{s(i)}) =Y_i -\wh\ad_{s(i)}(Z_i)$ as the dependent variable, where $s(i)$ denotes the fold that contains observation $i$:
		\begin{equation*}
			\htau(h; \wh\ad) = \Sum w_i(h) M_i(\wh\ad_{s(i)}).
		\end{equation*}
	\end{enumerate}
	
	Our theoretical analysis below shows that under weak conditions the estimator~$\htau(h; \wh\ad)$ is asymptotically equivalent to the infeasible estimator $\wh{\tau}(h; \bar\ad) = \Sum w_i(h) M_i(\bar\ad)$, where $\bar\ad$ is a deterministic approximation of $\wh\ad$ whose error vanishes in large samples in some appropriate sense. Importantly, our approach does not require the first-stage estimator of~$\ad_0$ to be consistent, in the sense that we allow for the possibility that $\bar\ad\neq\ad_0$. The first-stage estimator also does not have to converge with a particularly fast rate. In view of~\eqref{eq:dist_fixedmu}, it then holds that
	\begin{equation*}
		\htau(h; \wh\ad) \stackrel{a}{\sim}N\left(\tau + h^2B_{\sss base},  (nh)^{-1} V(\bar\ad)\right).
	\end{equation*}
As mentioned above, the variance $V(\bar\ad)$ is  minimized if $\bar\ad =\ad_0$. However, the distributional approximation is also valid if  $\bar\ad \neq\ad_0$  because the moment condition~\eqref{eq:EMu} holds for  \emph{all} adjustment functions, and not just the optimal one. In that sense, our procedure is robust to misspecification or over-regularized estimation of the optimal adjustment function. Moreover, we argue that  $ V(\bar\ad)$ is typically smaller than $V_{\sss base}$ or $V_{\sss lin}$ even if $\bar\ad \neq\ad_0$. 

We also show that other common steps in an empirical RD analysis can easily be implemented by applying existing methods that are devised for settings without covariates to the generated data set $\{(X_i,M_i(\wh\ad_{s(i)}))\}_{i\in[n]}$.  For example, we can construct an estimator of the bandwidth that minimizes the asymptotic mean squared error of $\htau(h; \wh\ad)$ by using the procedures proposed by \citet{calonico2014robust} or \citet{imbens2012optimal}. Similarly, we can generalize the standard error~\eqref{eq:lin_se} and construct a valid nearest-neighbor standard error  $\wh{\textnormal{se}}(h; \hmu)$ as
\begin{align}\label{eq:fca_se}
\widehat{\textnormal{se}}^2(h;\hmu) = \sum_{i=1}^n w_i(h)^2 \widehat\sigma_{i}^2(\wh\ad), \quad \widehat\sigma_{i}^2(\wh\ad) = \frac{R}{R+1}\Big( M_i(\wh\ad_{s(i)})-\frac{1}{R}\sum_{j\in\mathcal{R}_i} M_j(\wh\ad_{s(j)}) \Big)^2,
\end{align}	
and construct ``robust bias correction'' and ``bias-aware'' confidence intervals as in \citet{calonico2014robust} and \citet{ armstrong2020simple}, respectively.
To reduce the sensitivity of  empirical findings to the particular data split in the cross-fitting step, we can proceed as in \citet[Section 3.4]{chernozhukov2018double} by  repeating the respective procedure several times and reporting a summary measure of the results, such as the median. We recommend proceeding like this especially when working with smaller sample sizes.

\subsection{Estimating the Adjustment Function}\label{sec:est_adj_function}

We now discuss some implementation details for the first-stage estimator of the optimal adjustment function $\ad_0$. Our theoretical analysis allows for a variety of different methods to be used in this context. 
If one wishes to maintain the simplicity of the conventional linear adjustment, one can obtain a  ``cross-fitted'' version of $\wh\tau_{\sss lin}(h)$ by setting $\wh\ad_s(z) =  z^\top\wh\gamma_{s,h}$, for $s \in \{1,\ldots,S\}$, where  $\wh\gamma_{s,h}$ is the minimizer w.r.t.\ $\gamma$ in the minimization problem
\begin{equation}\label{eq:llregadj}
\min_{\beta, \gamma} \sum_{i \in I_s^c} K_h(X_i) (Y_i - S_i^\top\beta- Z_i^{\top}\gamma )^2.
\end{equation}
We refer to  this procedure as the  cross-fitted ``localized'' linear adjustment.
The adjustment coefficients, however, do not need to be estimated using the kernel weights from the second-stage regression. 
As an important variant of cross-fitted linear adjustments, we consider $\wh\ad_s(z) =  z^\top\wh\gamma_{s,\infty}$, where $\wh\gamma_{s,\infty}$ is obtained via a version of \eqref{eq:llregadj} without kernel weights.
Since all the observations outside of fold $s$ are used to obtain $\wh\gamma_{s,\infty}$, we refer to this procedure as the cross-fitted ``global'' linear adjustment. In finite samples, the global version can outperform the localized one in terms of the resulting standard deviation of the RD estimator due to the increased stability of the first-stage estimates. This approach can be naturally extended to other parametric models where the components involving $S_i$ and $Z_i$ are additively separable, and it can be combined with lasso regularization. The cross-fitted post-lasso adjustments are then obtained via \eqref{eq:llregadj} with the set of covariates restricted to those ``selected'' by the lasso.

More generally, our approach allows for any parametric or classical nonparametric method, as well as generic modern machine learning methods. To allow for such generality, we consider adjustment functions of the form 
$$
	\wh\ad_s (z) = \frac{1}{2}(\wh\mu_s^+(z)+\wh\mu_s^-(z)),\quad s \in \{1,\ldots,S\},
$$ 
where  $\wh\mu_s^+(z)$ and $\wh\mu_s^-(z)$ are separate estimators of $\mu_0^+(z)=\E[Y_i|X_i=0^+,Z_i=z]$ and  $\mu_0^-(z)=\E[Y_i|X_i=0^-,Z_i=z]$, respectively, using the data outside of fold $s$. With appropriate subject knowledge, one can then, for example, specify parametric models for $\E[Y_i|X_i=x,Z_i=z]$. 
If the number of covariates is small, the functions $\mu_0^+$ and $\mu_0^-$ can be also estimated using classical nonparametric methods under smoothness conditions, with local polynomial regression being particularly suitable due to its good boundary properties. 
If the number of covariates is large, however, we recommend using modern machine learning methods, such as lasso or post-lasso regression, random forests, deep neural networks, boosting, or ensemble combinations thereof.

One issue to consider is that the default implementations of generic machine learning estimators of $\E[Y_i|X_i=x,Z_i=z]$ will  not automatically produce an estimate with a jump at the cutoff. As having this feature is potentially important in our context, we consider two simple variations  of generic machine learning estimators. First, let
\begin{equation}\label{eq:eta_ml1}
	\wh \E_s[Y_i|T_i=t, X_i=x, Z_i=z] =\argmin_{f\in\mathcal{F}}\sum_{i \in I_s^c} l(Y_i,f(T_i, X_i,Z_i))
\end{equation}
be a generic machine learning  estimator of  $\E[Y_i|T_i=t, X_i=x, Z_i=z]$, computed by minimizing some empirical loss function $L(f) =\sum_{i \in I_s^c} l(Y_i,f(T_i, X_i, Z_i))$ over a set of candidate functions $\mathcal{F}$. We can then estimate $\mu^+(z)$ by  $\wh \E_s[Y_i|T_i=1, X_i=0, Z_i=z]$ and $\mu^-(z)$ by $\wh \E_s[Y_i|T_i=0, X_i=0, Z_i=z]$. Here including the seemingly superfluous treatment indicator $T_i=\1{X_i\geq 0}$ as a predictor allows the machine learner to create the ``jump'' in the estimated function at the cutoff value.  
We refer to this type of implementation as ``global'', as it uses all available observations.

To define the second type of implementation of machine learning we consider in this paper, let
\begin{equation}\label{eq:eta_ml2}
	\wh \E_s[Y_i|T_i=t, Z_i=z] =\argmin_{f\in\mathcal{F}}\sum_{i \in I_s^c} K(X_i/b) l(Y_i,f(T_i, Z_i))
\end{equation}
be a generic machine learning estimator of  $\E[Y_i|T_i=t, Z_i=z]$, where $b>0$ is some positive bandwidth and $K$ is again a kernel function. We can then estimate $\mu^+(z)$ as  $\wh \E_s[Y_i|T_i=1, Z_i=z]$ and $\mu^-(z)$ as $\wh \E_s[Y_i|T_i=0, Z_i=z]$. We refer to this type of implementation as ``localized'', as it effectively  only uses data points whose realization of the running variable is close to the cutoff. The idea is to produce an estimate with small empirical loss in the relevant area around the cutoff rather than one with small ``overall'' loss. The downside of this approach is the reduced effective sample size and the need to choose the tuning parameter $b$.\footnote{The choice of $b$ involves a bias-variance trade-off similar to the one encountered in classical nonparametric kernel regression problems.
We are not aware of generic theoretical results for such  estimators in settings with $b\to 0$ as $n\to\infty$. Specific results are given by
\citet{su2019non} for the lasso, and by \citet{colangelo2022double} for series estimators and deep neural networks.
In our applications below, we simply use $b=h$. To make this simultaneous choice feasible, we use an iterative procedure. We first choose a reasonable preliminary first-stage bandwidth, like the one that would be optimal for RD estimation without covariates, and generate preliminary versions of the adjustment terms as described above. Next, we use the preliminary covariate-adjusted outcomes to pick 
 an optimized second-stage bandwidth. Finally, we rerun both stages with this last bandwidth to obtain our empirical results.}

\subsection{Our Proposed Flexible Adjustment}\label{sec:proposed_method}
The specific flexible covariate adjustment that we propose and implement in our empirical analysis and simulations is an ensemble of the following methods:
\begin{enumerate*}[label=(\roman*)]
	\item linear regression; 
	\item post-lasso;
	\item boosted trees; and
	\item random forest.
\end{enumerate*}
All four methods are implemented in localized and global versions discussed above. We use the cross-fitted linear and post-lasso adjustments specified in the discussion following \eqref{eq:llregadj}, and the boosted trees and random forest adjustments are based on the formulations in \eqref{eq:eta_ml1} and~\eqref{eq:eta_ml2}.
Our proposed flexible covariate adjustment is a convex combination of these eight adjustment functions and the trivial no-adjustment function that minimizes the mean squared error for predicting the outcome close to the cutoff. Specifically, we employ the super learning approach of \citet{laan2007super}, where the optimal weights are chosen via cross-validation.\footnote{We implemented our procedure in the {\tt R} programming language. The boosted trees adjustments are obtained using the package {\tt xgboost} with trees of depth 2 and shrinkage rate 0.1. The number of boosting iterations is chosen via cross-validation with a maximum of 1000 iterations, separately for the localized and global versions. The random forest with 1000 trees is implemented using the package {\tt ranger} with the minimal node size set to the maximum of 10 and 0.1\% of the sample size. All other parameters are set to the default values in the respective packages. For post-lasso estimation, we use the function {\tt rlasso} from the package {\tt hdm} with a data-driven penalty parameter. We use the package {\tt SuperLearner} to choose the optimal weights via cross-validation.}

	\section{Main Theoretical Results}\label{sec:Results}
	
	\subsection{Assumptions} We study the theoretical properties of our proposed estimator under a number of conditions that are either standard in the RD literature, or concern the general properties of  the first-stage estimator $\wh\ad$. To describe them, we  denote the support of $Z_i$ by $\mathcal{Z}$, and the support of $X_i$ by $\X$. We  write $\Xh=\X \cap [-h,h]$, and $\Zh$ denotes the support of $Z_i$ given $X_i\in \Xh$. We also define the following class of admissible adjustment functions:
	\begin{equation*}
		\Class = \{\ad: \E[\ad(Z_i)|X_i=x] \text{ exists and is twice continuously differentiable around the cutoff} \}.
	\end{equation*}
	The class $\Class$ implicitly depends on the underlying conditional distribution of the covariates given the running variable. If this conditional distribution changes smoothly around the cutoff, the class $\Class$ contains essentially all functions of the covariates, subject only to technical integrability conditions.\footnote{For example, if the conditional distribution of $Z_i$ given $X_i$  admits a density $f_{Z|X}(z|x)$ that is twice continuously differentiable in $x$ and $ |\partial_x^jf_{Z|X}(z|x)|\leq g_j(z)$ for  all $x$ in a neighborhood of the cutoff, some integrable functions $g_j$, and $j \in \{0,1,2\}$, then $\mc E$ contains all bounded Borel functions. The class $\mc E$ also contains all polynomials if the corresponding conditional moments of $Z_i$ exist and are twice continuously differentiable.\label{footnote:Density}}
	
	\begin{assumption}\label{ass:1stage}
		For all $n\in\IN$, there exist a set $\Tn\subset\Class$ and a function $\bar\ad \in\Tn$ such that:
		(i) $\wh\ad_{s}$ belongs to $\Tn$ with probability approaching 1 for all $s\in[S]$;
		(ii) it holds that: $$ \sup_{\ad \in \Tn} \sup_{x \in \mc X_h} \E\left[ (\ad(Z_i)-\bar\ad(Z_i))^2|X_i=x\right] = O(r_{n}^2)$$
		for some deterministic sequence $r_{n}=o(1)$.
	\end{assumption}
	
	Assumption~\ref{ass:1stage} states that  with high probability the first-stage estimator belongs to some realization set $\Tn \subset \Class$. As discussed above, this requirement seems weak as we generally expect the class $\Class$ to be very large.
	The assumption also states that the sets $\Tn$ contract around a deterministic sequence of functions in a particular $L_2$-type sense. Note that taking the supremum in Assumption~\ref{ass:1stage}  over $\mc X_h$ instead of $\mc X$ suffices as the properties of the first-stage estimator are only relevant for observations with non-zero kernel weights in the second-stage local linear regression. The assumption does not impose any restrictions on the speed at which $\wh\ad$ concentrates around $\bar\ad$. It also allows the function $\bar\ad$ to be different from the target function $\ad_0$, so that $\wh\ad$ could be inconsistent for $\ad_0$.
		
	Mean-square error consistency as in Assumption~\ref{ass:1stage} follows under classical conditions for the parametric and nonparametric procedures for settings in which the number of covariates is fixed. 	
	For the ``localized'' versions of the machine learning methods described in Section~\ref{sec:est_adj_function}, existing results imply that for fixed $b>0$ and $K$ the uniform kernel,
	\begin{equation}\label{eq:ass:MSE1}
		\sup_{ \ad \in \Tn}\E\left[ (\ad(Z_i)-\bar\ad(Z_i))^2|X_i \in (-b, b) \right] = O( r_{n}^{\, 2}),
	\end{equation}
	with $\bar\ad(z) = ( \E[Y_i|X_i\in(- b,0), Z_i=z] + \E[Y_i|X_i\in(0,b), Z_i=z])/2$ and some $ r_{n}=o(1)$, under general conditions. For example, if $\bar\ad(z)$ is contained in a Hölder class of order $s$, then~\eqref{eq:ass:MSE1} can hold with $r_n^2 = n^{-2s/(2s+d)}$ for estimators that exploit smoothness. If $\bar\ad(z)$ is $s$-sparse, then~\eqref{eq:ass:MSE1} can hold with $r_n^2 = s\log(d)/n$ for estimators that exploit sparsity.  
	Assumption~\ref{ass:1stage} then follows from~\eqref{eq:ass:MSE1} if  the conditional distribution of the covariates does not change ``too quickly'' when moving away from the cutoff. For example, if the  covariates are continuously distributed conditional on the running variable, having that
	\begin{equation*}
		\sup_{x \in\Xh} \sup_{z\in\Zh} \frac{f_{Z|X}(z|x)}{f_{Z|X\in (-b,b) }(z)} < C,
	\end{equation*}
	for some constant $C$ and all $n$   sufficiently large, suffices. Similar conditions can be given for discrete conditional covariate distributions, or intermediate cases.
	If $\E[Y_i|X_i=x, Z_i=z]$ is sufficiently smooth in $x$ on both sides of the cutoff, we can also expect that $\bar\ad$ is ``close'' to $\ad_0$ for ``small'' values of $b$. Formal rate results with $b\to 0$ are given by \citet{su2019non} for the Lasso, and by \citet{colangelo2022double} for series and deep neural networks.

	\begin{assumption}\label{ass:derivatives}
		For $j \in \{1,2\}$, it holds that:
		$$
		\sup_{\ad \in \Tn} \sup_{x \in \mc X_h\setminus \{0\} } \left|\partial^j_x \E\left[ \ad(Z_i) - \bar\ad(Z_i) |X_i=x\right] \right| = O(v_{j,n})
		$$ for some deterministic sequences $v_{j,n}=o(1)$.
	\end{assumption}
	
	Assumption~\ref{ass:derivatives} also concerns the first-stage estimator, and requires the first and second derivatives of $\E\left[ \ad(Z_i)-\bar\ad(Z_i) |X_i=x\right]$ to be close to zero in large samples  for all $\ad\in\Tn$. We generally expect this condition to hold with $v_{1,n}=v_{2,n}=r_n$ with $r_n$  as in Assumption~\ref{ass:1stage}.\footnote{For example, this can easily be seen to be the case if  $\wh\eta$ converges to $\bar\eta$ uniformly over $\mc Z$ with rate $r_n$ and the smoothness conditions for $f_{Z|X}(z|x)$ given in footnote~\ref{footnote:Density} hold. Similarly, under regularity conditions on $\E[Z_i|X_i=x]$, these three rates coincide if $\Tn$ contains only linear functions. Without any additional restrictions on first-stage estimators or $\mc T_n$, except that it  contains only bounded functions, Assumption~\ref{ass:derivatives} also follows  from Assumption~\ref{ass:1stage}, again with $v_{1,n}=v_{2,n}=r_n$, under restrictions concerning solely the conditional density $f_{Z|X}(z|x)$. Specifically, it suffices that $\E\big[\big(\partial_x^jf_{Z|X}(Z_i|x)/f_{Z|X}(Z_i|x)\big)^2|X_i=x\big]$ is bounded for $j \in \{1,2\}$ uniformly in $x$ and the conditions from footnote~\ref{footnote:Density} hold.}
		
	\begin{assumption}\label{ass:reg1a}
	 $X_i$ is continuously distributed with density $f_X$, which is continuous and bounded away from zero over an open neighborhood of the cutoff.
	\end{assumption}
	
	Assumption~\ref{ass:reg1a} is a standard condition from the RD literature. Continuity of the running variable's density $f_X$ around the cutoff is strictly speaking not required for an RD analysis. However, a discontinuity in $f_X$ is typically considered to be an indication of a design failure that prevents $\tau$ from being interpreted as a causal parameter \citep{mccrary2008manipulation, gerard2020bounds}. For this reason, we focus on the case of a  continuous running variable density in this paper.
	
	\begin{assumption}\label{ass:reg1b}
		\begin{enumerate*}[label=(\roman*)]
			\item\label{ass:item:kernel} The kernel function $K$ is a bounded and symmetric density
			function that is continuous on its support, and equal to zero outside some compact set, say $[-1,1]$;
			\item The bandwidth satisfies $h\to 0$ and $nh \to \infty$ as  $n \to \infty$.
		\end{enumerate*}
	\end{assumption}
	
	The conditions on the kernel and the bandwidth that are imposed in Assumption~\ref{ass:reg1b} are standard in the RD literature.

	\begin{assumption}\label{ass:reg2}
		There exist constants $C$ and $L$ such that the following conditions hold for all $n\in\IN$.
		\begin{enumerate*}[label=(\roman*)]
			\item $\E[M_i(\bar\ad)|X_i=x]$ is twice continuously differentiable on $\mathcal{X} \setminus \{0\}$ with $L$-Lipschitz continuous second derivative bounded by $C$;
			\item\label{item:ass:qmoment} For all $x\in \mathcal{X}$ and some $q>2$ $\E[| M_i(\bar\ad) - \E[M_i(\bar\ad) |X_i]|^q| X_i=x]$ exists and 
			is bounded by $C$;
			\item\label{item:ass:variance} $\V[M_i(\bar\ad) |X_i=x]$ is $L$-Lipschitz continuous and bounded from below by $1/C$ for all $x \in \mathcal X \setminus \{0\}$. 
		\end{enumerate*}
	\end{assumption}
	
	Assumption~\ref{ass:reg2} collects standard conditions for an RD analysis with $M_i(\bar\ad)$ as the outcome variable. 
	Part~(i) imposes smoothness conditions on $\E[M_i(\bar\ad)|X_i=x]$, and 
	parts \ref{item:ass:qmoment} and \ref{item:ass:variance} impose restrictions on conditional moments of the outcome variable. Throughout, we use constants $C$ and $L$ independent of the sample size to ensure asymptotic normality of the infeasible estimator $\wh\tau(h;\bar\ad)$ even in settings where the distribution of the data, and thus $\bar\ad$, might change with $n$.
	
	\subsection{Asymptotic Properties}
	We give four main results in this subsection. The first shows that our proposed estimator $\htau(h; \hmu)$ is asymptotically equivalent to an infeasible analog $\wh{\tau}(h;\bar\ad)$ that replaces the estimator $\hmu$  with the deterministic sequence $\bar\ad$; the second shows the asymptotic  normality of the estimator; the third characterizes how the asymptotic variance changes with the adjustment function and shows that $\ad_0$ is indeed the optimal adjustment; and the fourth shows the impact of flexible covariate adjustments on the optimal bandwidth and the corresponding mean squared error.

	\begin{theorem}\label{th:Equivalence}
		Suppose that Assumptions~\ref{ass:1stage}--\ref{ass:reg1b}  hold. Then
		$$\htau(h; \hmu)  = \wh{\tau}(h;\bar\ad) + O_P(r_n (nh)^{-1/2}  + v_{1,n} h(nh)^{-1/2} + v_{2,n}h^2 ).$$ 
	\end{theorem}

	Theorem~\ref{th:Equivalence} is easiest to interpret in what is arguably the standard case that  $v_{1,n}=v_{2,n}=r_n$, in which it holds that
	$$\htau(h; \hmu)   = \wh{\tau}(h;\bar\ad) +O_P(r_n(h^2+(nh)^{-1/2})) =  \wh{\tau}(h;\bar\ad) + O_P(r_n|\wh{\tau}(h;\bar\ad)-\tau|).$$
	The accuracy of the approximation that $\htau(h; \hmu) \approx \wh{\tau}(h;\bar\ad)$ thus increases with the rate at which $\wh\ad$ concentrates around $\bar\ad$, but first-order asymptotic equivalence holds even if the first-stage estimator converges arbitrarily slowly. 
	This insensitivity of 
	$\htau(h; \hmu)$ to sampling variation in $\hmu$ occurs because $\htau(h; \hmu)$ is based on the  moment condition
	$$\tau = \E[M_i(\ad)|X_i=0^+] -  \E[M_i(\ad)|X_i=0^-], $$
	which is  insensitive to variation in $\ad$. 
	Moment conditions with a local form of insensitivity with respect to a nuisance function, often called Neyman orthogonality, are used extensively in the recent literature on two-stage estimators that use machine learning in the first stage \citep[e.g.][]{belloni2017program, chernozhukov2018double}. The global insensitivity that arises in our RD setup is stronger, and allows us to work with weaker conditions on the first-stage estimates than those used in papers that work with Neyman orthogonality. Similarly, globally insensitive moment functions exist, for example, in certain types of randomized experiments, and our proposed estimator is in many ways analogous to efficient estimators in such setups; see Section~\ref{sec:ATE} for further discussion.
	
	\begin{theorem}\label{th:Normality}
		Suppose that Assumptions~\ref{ass:1stage}--\ref{ass:reg2} hold. Then
		\[
		\sqrt{nh}\, V(\bar\ad)^{-1/2} \left( \htau(h; \hmu)  - \tau - h^2 B_n \right) \overset{d}{\to}  \mathcal{N}(0, 1),
		\]
		for some $B_n = B_{base} + o_P(1) $, where $B_{base}$ and $V(\cdot)$ are as defined in Sections~\ref{sec:standardrdestimator} and~\ref{sec:motivation}, respectively.
	\end{theorem}

	Theorem~\ref{th:Normality} follows from Theorem~\ref{th:Equivalence} under the additional regularity conditions of Assumption~\ref{ass:reg2}. It shows that our estimator is asymptotically normal,  gives
	explicit expressions for its asymptotic bias and variance, and justifies the distributional approximation given in Section~\ref{sec:Estimator}.

	\begin{theorem}\label{th:var}
		Suppose  $\E[Y_i^2|X_i=x]$ is  uniformly bounded in $x$,  the limit $ \V[Y_i-\mu_0^\star(Z_i)|X_i=0^\star]$ exists for $\star \in\{+,-\}$, and  $\ad_0 \in \ClassVar$, where the function class $\ClassVar$ is defined as 
		\begin{align*}
			\ClassVar \equiv\left\{ \ad: \V[\ad(Z_i) |X_i=x] \text{ and } \Cov[\ad(Z_i),\mu^\star_0(Z_i) |X_i=x]\text{ are continuous for }  \star\in\{+,-\}\right\}.
		\end{align*} Then, for any  $\ad^{(a)},\ad^{(b)}\in \ClassVar$,  
		$$
		V(\ad^{(a)}) - V(\ad^{(b)}) = 2\frac{\bar\kappa}{f_X(0)}\left(\V[\ad_0(Z_i)-\ad^{(a)}(Z_i)|X_i=0] - \V[\ad_0(Z_i)-\ad^{(b)}(Z_i)|X_i=0]\right).
		$$
	\end{theorem}	
	
	Theorem~\ref{th:var}   introduces a function class $\ClassVar$ that, similarly to the class $\mc E$ above, enforces some technical integrability conditions. 
	The theorem shows that $V(\ad^{(a)}) < V(\ad^{(b)})$ for 
 generic adjustment functions $\ad^{(a)}$ and $\ad^{(b)}$
 if and only if $\V[\ad_0(Z_i)-\ad^{(a)}(Z_i)|X_i=0] < \V[\ad_0(Z_i)-\ad^{(b)}(Z_i)|X_i=0]$. That is, the ``closer'' (in a particular $L_2$-sense) the adjustment function is to the optimal one, the smaller the asymptotic variance. In consequence, the lowest possible value of $V(\bar\ad)$  is achieved for $\bar\ad=\ad_0$.
Even if $\bar\ad \neq \ad_0$, our flexible covariate adjustment RD estimators typically still have smaller asymptotic variances than the no covariates and linear adjustment RD estimators. Specifically,   $V(\bar\ad) < V_{base}$ if and only if $\V[\ad_0(Z_i)-\bar\ad(Z_i)|X_i=0] < \V[\ad_0(Z_i)|X_i=0]$, i.e.\ whenever $\bar\ad(Z_i)$ captures \emph{some} of the variance of $\ad_0(Z_i)$ among units near the cutoff; and $V(\bar\ad) < V_{lin}$ if and only if $\V[\ad_0(Z_i)-\bar\ad(Z_i)|X_i=0] < \V[\ad_0(Z_i)-Z_i^\top\gamma_0|X_i=0]$, i.e.\ whenever $\bar\ad$ is ``closer'' to $\ad_0$ in our particular $L_2$-sense than the population linear adjustment function. We note that this property should be satisfied for all commonly used machine learning algorithms; and that it can be satisfied mechanically by considering ensemble methods, such as super learning, that include a simple fixed adjustment or a simple linear adjustment as parts of the ensemble.
	
 \begin{theorem}\label{th:bandwidth}   
 Let 
 $\textnormal{AMSE}(h,\ad) = h^4 B_{\sss base}^2 + (nh)^{-1}V(\ad)$ be the approximate mean squared error of $\htau(h; \ad)$ with $B_{\sss base} \neq0$, and let
 $h_{\sss{AMSE}}(\eta) =\argmin_h \textnormal{AMSE}(h,\ad)= n^{-1/5} \left( V(\ad)/4 B_{\sss base}^2\right)^{1/5}$ be the corresponding optimal bandwidth.
Then for any pair of adjustment functions $\ad^{(a)},\ad^{(b)}\in \ClassVar$ with $v(\ad^{(a)}),v(\ad^{(b)})>0$,
$$\frac{h_{\sss{AMSE}}(\ad^{(a)})}{h_{\sss{AMSE}}(\ad^{(b)})} =\left( \frac{v(\ad^{(a)})}{v(\ad^{(b)})}\right)^{1/5} \quad\text{and}\quad \frac{\textnormal{AMSE}(h_{\sss{AMSE}}(\ad^{(a)}),\ad^{(a)})}{\textnormal{AMSE}(h_{\sss{AMSE}}(\ad^{(b)}),\ad^{(b)})} = \left( \frac{v(\ad^{(a)})}{v(\ad^{(b)})}\right)^{4/5},$$
where $v(\ad) = \V[M_i(\ad) |X_i=0^+]+\V[M_i(\ad)|X_i=0^-]$.
 \end{theorem}
Theorem~\ref{th:bandwidth} implies that flexible covariate adjustments can reduce the approximate mean squared error of our estimator not only directly through a smaller asymptotic variance term but also indirectly through a change in the optimal bandwidth and a corresponding reduction in bias. That is, if $V(\ad^{(a)}) < V(\ad^{(b)})$ for 
 generic adjustment functions $\ad^{(a)}$ and $\ad^{(b)}$, then the optimal bandwidth $h_{\sss{AMSE}}(\ad^{(a)})$ is smaller than $h_{\sss{AMSE}}(\ad^{(b)})$, and the corresponding estimator 
  $\htau(h_{\sss{AMSE}}(\ad^{(a)}); \ad^{(a)})$ has both smaller asymptotic bias \emph{and} smaller asymptotic variance than $\htau(h_{\sss{AMSE}}(\ad^{(b)}); \ad^{(b)})$.
	
	\section{Additional Theoretical Results and Discussions}\label{sec:Implementation}
	
	\subsection{Bandwidth Choice and Inference}\label{sec:Inference_Overview}
We formally show in Appendix~\ref{sec:Inference} that standard methods for bandwidth choice and confidence interval construction based on the no covariates RD estimator maintain their general asymptotic properties when they are applied    to the generated data set $\{(X_i,M_i(\wh\ad_{s(i)}))\}_{i\in[n]}$ without any adjustment for the sampling uncertainty about the estimated adjustment function. Specifically, we derive three groups of results, all under conditions that are rather weak and  analogous to those commonly imposed  in setups without covariates.
 
	 First,  we show  that the nearest neighbor standard error~\eqref{eq:fca_se} is consistent, in the sense that	
	$$nh\,\hse^2( h; \wh \ad) /V(\bar\ad)\overset{
p}{\to}1.$$
 Second,  we show that commonly used methods for confidence interval construction achieve correct asymptotic coverage.	For example, assuming a bound on $|\partial^2_x \E[M_i(\bar\ad)|X_i=x]|$, the absolute value of the second derivative of the conditional expectation of the adjusted outcome  given the running variable, one can construct a ``bias-aware'' confidence interval as in \citet{armstrong2020simple} as
	$$ 
	CI^{ba}_{1-\alpha} = \left[\htau(h; \wh \ad) \pm z_\alpha(\bar b(h)/\hse( h; \wh \ad)) \,\hse( h; \wh \ad)\right].
	$$ 
	Here $z_\alpha(r)$ is the $(1-\alpha)$-quantile of $|N(r,1)|$, the folded normal distribution with mean $r$ and variance one, and $\bar b(h)$ is an explicit bound on the finite-sample bias of $\wh\tau(h;\bar\eta)$ given in the appendix. 	Alternatively, one can also construct a ``robust bias correction'' confidence interval as in \citet{calonico2014robust} by subtracting a local quadratic estimate of the first-order bias of $\htau(h; \wh \ad)$ from the estimator, and adjusting the standard error appropriately. This yields a confidence interval of the form
	\[
	CI^{rbc}_{1-\alpha} = \big[\htau^{rbc}(h; \wh \ad) \pm z_\alpha \hse^{rbc}( h; \wh \ad)\big],
	\]
	where  $z_\alpha = z_\alpha(0)$ and the other terms are formally defined in the appendix.
Third, we show  that the MSE-optimal bandwidth selector $\wh h_{n}$ of \cite{calonico2014robust}, which is similar to that of \citet{imbens2012optimal}, consistently estimates the AMSE-optimal bandwidth $h_{\sss{AMSE}}(\bar \ad)$ defined in Theorem~\ref{th:bandwidth}, in the sense that
$$ \wh h_{n} / h_{\sss{AMSE}}(\bar \ad)  \overset{p}{\to} 1.$$
RD estimation and inference with flexible covariate adjustments are thus easy to implement with existing software packages.

	\subsection{Analogies with Randomized Experiments}\label{sec:ATE}
		
	The results in Section~\ref{sec:Results} are qualitatively similar to ones obtained for efficient influence function (EIF) estimators of the population average treatment effect (PATE) in  randomized experiments  with known and constant propensity scores \citep[e.g.,][]{wager2016high,chernozhukov2018double}.
	To see this, consider a randomized experiment with unconfounded treatment assignment and known constant propensity score $p$. Using our notation in an analogous fashion, the  EIF of the PATE in such a setup is typically written  in the form 
	$$\psi_i(m_0^0, m_0^1) =  m_0^1(Z_i) - m_0^0(Z_i) + \frac{T_i(Y_i-m_0^1(Z_i))}{p} - \frac{(1-T_i)(Y_i-m_0^0(Z_i))}{1-p}, $$
	where $m_0^t(z) = \E[Y_i|Z_i=z,T_i=t]$ for $t \in \{0,1\}$ \citep[e.g.,][]{hahn1998role}. The minimum variance that any regular estimator of the PATE can achieve is thus $V_{\sss \textnormal{PATE}}=\V(\psi_i(m_0^0, m_0^1))$. By randomization, it also holds that
	$\tau_{\sss \textnormal{PATE}} = \E[\psi_i(m^0, m^1 )]$ for all (suitably integrable) functions $m^0 $ and $m^1$. The PATE is thus identified by a moment function that satisfies a global invariance property.
	A sample analog estimator of $\tau_{\sss \textnormal{PATE}}$ based on this moment function has asymptotic variance $V_{\sss \textnormal{PATE}}$   if $\wh m^t$ is a consistent estimator of $m_0^t$ for $t\in\{0,1\}$, but remains consistent and asymptotically normal with asymptotic variance $\V(\psi_i(\bar m^0, \bar m^1))$ if $\wh m^t$ is consistent for some other function $\bar{m}^t$, $t\in\{0,1\}$. The convergence of $\wh m^t$ to  $\bar{m}^t$ can be arbitrarily slow for these results \citep[e.g.][]{wager2016high,chernozhukov2018double}. 
	
	The qualitative parallels between these findings and ours in Section~\ref{sec:Results} arise because our covariate-adjusted RD estimator is in many ways a direct analog of such EIF estimators. To show this,   write $ m(z)=  (1- p){m}^1(z) + p\,{m}^0(z)$ for any two functions $m^0$ and $m^1$, so that $ m_0(z)=  (1- p){m}_0^1(z) + p\,{m}_0^0(z)$. The PATE's influence function can then be expressed as
	$$\psi_i({m}_0^0, {m}_0^1) =  \frac{T_i(Y_i- m_0(Z_i))}{p} - \frac{(1-T_i)(Y_i-m_0(Z_i))}{1-p}, $$
	and it holds that $$\E[\psi_i(m^0, m^1 )]=\E [Y_i- m(Z_i)|T_i=1] - \E[Y_i-  m(Z_i)|T_i=0],$$ which is the difference in average covariate-adjusted outcomes between treated and untreated units. This last equation is fully analogous to our equation~\eqref{eq:EMu}, with $p=1/2$, and conditioning on $T_i=1$ and $T_i=0$ replaced by conditioning on $X_i$ in infinitesimal right and left neighborhoods of the cutoff (the value $p=1/2$ is appropriate here because  continuity of the running variable's density implies that an equal share of units close to the cutoff can be found on either side). An EIF estimator of $\tau_{\sss \textnormal{PATE}}$  is thus analogous to our estimator $\htau(h; \wh\ad)$, as they are both sample analogs of a moment function with the same basic properties.\footnote{We note that neither  in our setting nor with EIF estimation of the PATE would replacing the known propensity score with some empirical estimate result in any efficiency gains. The finding from \cite{hahn1998role} that using an estimated propensity score can be more efficient than using the true one refers to inverse probability weighting (IPW) type estimators and does not apply to the EIF type estimators we consider here.}

	\subsection{Fuzzy RD Designs}\label{sec:fuzzy}
	In fuzzy RD designs, units are assigned to treatment if their realization of the running variable falls above the threshold value, but might not comply with their assignment. The conditional treatment probability given the running variable hence changes discontinuously at the cutoff, but in contrast to sharp RD designs it does not jump from zero to one.  The parameter of interest in fuzzy RD designs is 
	$$
	\theta = \frac{\tau_Y}{\tau_T}   \equiv \frac{\E [Y_i|X_i=0^+] - \E [Y_i|X_i=0^-]}{\E [T_i|X_i=0^+]- \E [T_i|X_i=0^-]},
	$$
	which is the ratio of two sharp RD estimands (throughout this subsection, the notation is analogous to that used before, with the subscripts $Y$ and $T$ referencing the respective outcome variable). Under standard conditions \citep{hahn2001identification, dong2014alternative}, one can interpret $\theta$ as the average causal effect of the treatment among units at the cutoff whose treatment decision is affected by whether their value of the running variable is above or below the cutoff.

	Similarly to sharp RD designs, predetermined covariates can be used in fuzzy RD designs to improve efficiency. Building on our proposed method, we consider estimating $\theta$ by the ratio of two generic flexible covariate-adjusted sharp RD estimators:
	\[
	\htheta(h; \wh \ad_Y,\wh \ad_T) =  \frac{ \htau_{Y} (h; \wh \ad_Y )}{\htau_{T} (h; \wh \ad_T ) }  =\frac{ \Sum w_{i}(h) (Y_i- \wh \ad_{Y, s(i)}(Z_i)) }{ \Sum w_{i}(h) (T_i- \wh  \ad_{T, s(i)}(Z_i)) }.
	\]

	\begin{proposition}\label{prop:fuzzy}
		Suppose that Assumptions~\ref{ass:1stage}--\ref{ass:reg2} hold also with $T_i$ replacing $Y_i$, mutatis mutandis, and $\tau_T \neq 0$. 		\begin{itemize}
			\item[(i)]  It holds that
			\[
			\sqrt{nh}\, V_\theta(\bar \ad_Y,\bar \ad_T)^{-1/2} \left( \htheta(h;\wh \ad_Y, \wh \ad_{\sss T})  -\theta - B_\theta(\bar \ad_Y,\bar\ad_T)h^2 \right) \stackrel{d}{\to} \mathcal{N}(0, 1),
			\]
			where
			\begin{align*}
				B_\theta(\bar\ad_Y,\bar\ad_T) & =\frac{ \bar{\nu}}{2 \tau_T}  \left(\partial_x^2 \E[ Y_i- \theta T_i|X_i=x]\big|_{x=0^+}-\partial_x^2 \E[ Y_i- \theta T_i |X_i=x]\big|_{x=0^-}\right) + o_P(1),\\
				V_\theta(\bar\ad_Y,\bar\ad_T) & =\frac{ \bar{\kappa} }{f_X(0)}  \left( \V[ U_i(\bar\ad_Y,\bar\ad_T)|X_i=0^+] + \V[ U_i(\bar\ad_Y,\bar\ad_T)|X_i=0^-] \right),
			\end{align*}
			and $U_i(\bar\ad_Y,\bar\ad_T)=\left(Y_i- \theta T_i - (\bar\ad_Y(Z_i) - \theta \bar\ad_T(Z_i))\right)/\tau_T$.
			\item[(ii)] 
			Suppose additionally that the assumptions of Theorem~\ref{th:var} hold, mutatis mutandis, also with $T_i$ replacing $Y_i$ and the definition of $\ClassVar$ adjusted accordingly. Then, for any  $\ad_Y^{(a)},\ad_Y^{(b)}, \ad_T^{(a)},\ad_T^{(b)} \in \ClassVar$,  
			it holds that
			\begin{align*}
				V_\theta& (\ad_Y^{(a)},\ad_T^{(a)}) - V_\theta(\ad_Y^{(b)},\ad_T^{(b)}) \\
				& = 	\frac{2 \bar\kappa}{\tau_T^2 f_X(0)}\left(  \V[\ad_{Y,0}(Z_i)- \theta \ad_{T,0}(Z_i) -(\ad_{Y}^{(a)}(Z_i)- \theta \ad_{T}^{(a)}(Z_i))|X_i=0] \right.\\
				& \qquad \left. -  \V[\ad_{Y,0}(Z_i)- \theta \ad_{T,0}(Z_i) -(\ad_{Y}^{(b)}(Z_i)- \theta \ad_{T}^{(b)}(Z_i))|X_i=0]\right) .
			\end{align*}
		\end{itemize}
	\end{proposition}
	
The first part of the proposition shows that our flexible covariate-adjusted fuzzy RD estimator is asymptotically normal, with asymptotic variance that depends on the population counterparts $\bar\eta_Y$ and $\bar\eta_T$ of the two estimated adjustment functions. This result can then be used to construct a confidence interval for $\theta$ based on the t-statistic. Alternatively, confidence sets for $\theta$ can be constructed via an Anderson-Rubin-type approach, which circumvents certain problems of ratio estimators \citep{noack2021bias}.
	
The second part of the proposition shows that the asymptotic variance of our estimator is minimized if the estimated adjustment functions concentrate around $\bar\eta_Y=\eta_{Y,0}$ and $\bar\eta_T=\eta_{T,0}$, respectively. That is, the optimal adjustment functions for fuzzy RD designs can be obtained by separately considering two covariate-adjusted sharp RD problems with outcomes $Y_i$ and $T_i$, respectively. This holds because for fixed adjustment functions $\ad_Y$ and $\ad_T$ we have that $\htheta(h; \ad_Y, \ad_T)  -\theta$ is 	first-order asymptotically equivalent to a sharp RD estimator with the infeasible outcome $U_i(\ad_Y,\ad_T)=\left(Y_i- \theta T_i - (\ad_Y(Z_i) - \theta \ad_T(Z_i))\right)/\tau_T$. By our Theorem~\ref{th:var},  the asymptotic variance of $\htheta(h; \ad_Y, \ad_T)$ is minimized if $(\ad_Y(Z_i) - \theta \ad_T(Z_i))/\tau_T$ equals the optimal adjustment function for the outcome $\left(Y_i- \theta T_i\right)/\tau_T$. By linearity of conditional expectations, this holds if $\eta_Y=\eta_{Y,0}$ and $\eta_T=\eta_{T,0}$.

	\subsection{Variants of Cross-Fitting} \label{sec:alternative_crossfitting}
	
	We note that instead of the type of cross-fitting described in Section~\ref{sec:Estimator}, which is analogous to the DML2 method in  \citet{chernozhukov2018double}, one could also consider an analog of their DML1 method,  which creates an overall  estimate by averaging separate estimates from each data fold. In our context, this would yield an estimator of the form
	$$ \htau_{alt}( h; \wh \ad) = \frac{1}{S}\sum_{s \in[S]} \sum_{i \in I_s} w_{i, s}(h) M_i(\wh \ad_s),$$
	where $w_{i, s}(h)$ is the local linear regression weight of unit $i$ using only data from the $s$-th fold; see Appendix~\ref{A:sec:weights}. Under the conditions of Theorem~\ref{th:Equivalence}~and~\ref{th:Normality}, we see from their proofs that
	\begin{equation}\label{eq:DML1}
		\htau_{alt}(h; \wh \ad)  -  \wh \tau_{alt}(h; \bar\ad) = O_P(r_n(nh)^{-1/2} + v_{2,n}h^2 ),
	\end{equation}
	and the estimators $\htau(h; \wh \ad)$ and $\htau_{alt}(h; \wh \ad)$ have the same first-order asymptotic distribution. Comparing the rate in~\eqref{eq:DML1} to that in Theorem~\ref{th:Equivalence} shows that the alternative implementation removes a term of order   $O_P(v_{1,n}h(nh)^{-1/2})$ from the expansion of the feasible estimator about its respective infeasible analogue.  We still prefer our proposed implementation of cross-fitting despite this improvement in second-order asymptotic properties	because it allows existing routines for bandwidth selection and confidence interval construction to be applied directly to the generated data set $\{(X_i,M_i(\wh\ad_{s(i)}))\}_{i\in[n]}$, as discussed in Section~\ref{sec:Inference_Overview}.\footnote{\cite{velez2024asymptotic} also argues in the context of a setting with regular parameters DML2 should be preferred to DML1 due to better bias and mean-squared error properties under particular asymptotic regimes.}

\section{Practical Performance}\label{sec:practical_performance}
To illustrate the scope for efficiency improvements that flexible covariate adjustments can achieve in practically relevant settings, we applied our method to a number of  recent RD studies. Specifically, we collected data from all articles that appeared between 2018 and 2023 in the main AEA journals for applied microeconomic research, fit into our general framework, use covariates, and have directly available public replication data. We found 16 such papers with a total of 56 main specifications.
For each of these specifications, we computed the length of bias-aware 95\% confidence intervals for the respective RD parameter based on local linear estimators that use our flexible covariate adjustments, linear adjustments, and no covariate adjustments, respectively.\footnote{We implement the linear adjustment using cross-fitting to ensure a fair benchmark for our flexible adjustment. In Section~\ref{sec:Simulations}, we illustrate in simulations  that the standard error based on the conventional linear adjustment may be downward-biased in settings where the number of covariates is large relative to the effective sample size.} For illustration purposes, we use the same ``no covariates'' smoothness bound for all estimators within each specification.\footnote{Specifically, we use the rule of thumb for the smoothness bound of \cite{imbens2017optimized}, which equals twice the maximal second derivative of a second-order global polynomial fitted on each side of the cutoff. The smoothness bound is calibrated based on the original outcomes. The results in the main text are not very sensitive to specific choices of the smoothness bounds.} This analysis captures the effect of covariate adjustments on the confidence interval length that is due to a reduction in variance, abstracting from other issues. In Section~\ref{sec:additional_empirical_results} of the Online~Supplement we also provide results where the smoothness bound is calibrated based on the adjusted outcomes as well as results based on the robust bias correction.  Appendix~\ref{sec:appendix_empirical_results} contains further details on the implementation and the data collection process, and Table~\ref{table::overview_of_papers} in the Online Supplement provides the complete list of papers and specifications used. 

\begin{figure}
	\centering
	\includegraphics[scale=0.8]{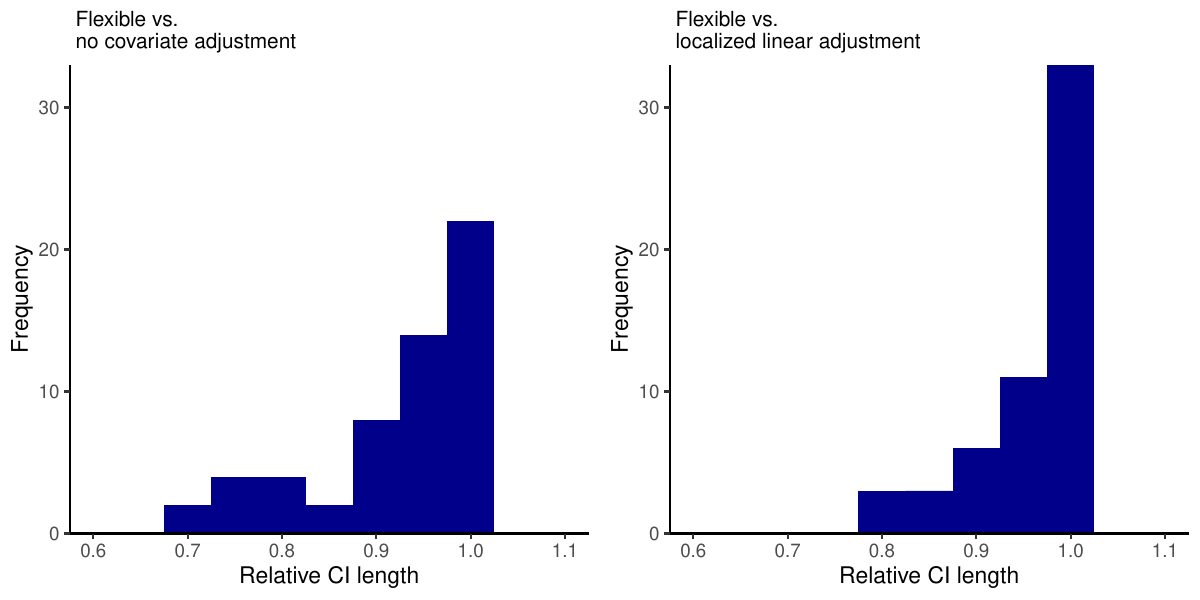}
	\caption{Bias-aware confidence interval (CI) lengths with flexible covariate adjustment relative to bias-aware  confidence interval length with no covariates and with cross-fitted localized linear adjustment for all specifications of the empirical literature survey. The smoothness constants are calibrated via the rule of thumb of \cite{imbens2017optimized} using the original outcomes.}
	\label{fig:applications_flexCA}
\end{figure}

 Figure~\ref{fig:applications_flexCA} shows the distribution of the ratio of confidence interval lengths for flexible adjustments relative to no adjustments in its left panel, and for flexible adjustments relative to linear adjustments in its right panel.  We first note  that the confidence intervals with flexible adjustments are never noticeably wider than those of the no covariates or linear adjustment  RD estimators. From the left panel, we can see that in nearly half of our specifications the flexible covariate adjustments yield confidence intervals that are not noticeably shorter than the ones obtained without covariate adjustments. Given the flexibility of our methods, this suggests that the covariates are not informative about the outcome in these specifications, and hence, there is no scope for efficiency gains. In many specifications, however, flexible covariate adjustments lead to substantially shorter confidence intervals, with the biggest reduction exceeding 30\%. To put this into perspective, note that one would have to increase the sample size used by the no covariates RD estimator by a factor of about 2.4 to achieve a similar reduction in the length of the confidence interval.
 
 The confidence interval length reductions shown in the right panel are smaller than those in the left panel, which is due to the fact that linear adjustments already capture some variation of the outcome given covariates. Nevertheless, the confidence intervals based on flexible adjustments, which account also for nonlinear relationships between the outcomes and the covariates, can be substantially shorter, with the biggest reduction reaching 20\% in our empirical exercise.

\section{Simulations}\label{sec:Simulations}

In this section, we investigate the finite-sample properties of our  proposed flexible covariate adjustment RD estimators under realistic conditions in two simulation studies. The first study's purpose is to show that our theoretical results provide accurate approximations to our estimator's actual finite sample properties, whereas the second study's purpose is to document how the properties of our estimator and those of existing methods are affected if the number of covariates becomes large relative to the effective sample size.

\subsection{General Setup}\label{sec:NumericalImplementation}
Our simulations are based on real data from \cite{londono2020upstream}, who study the impact of  merit-based college financial aid for low-income  students in a sharp RD design. Their data contain the outcome variable, a dummy for immediate enrollment in any post-secondary education, the running variable, a test score,\footnote{\cite{londono2020upstream} consider two different test scores as running variables. We focus on the \textit{SABER~11} test score in this section as it is available for a larger number of data points.} and 21 covariates, namely age, family size, indicators for gender,  ethnicity, employment status, parent's education, household residential stratum, high school schedule, and high school type. Our simulations involve repeatedly drawing random samples from a version of the data that is restricted to the $n=259,419$ observations with test scores below the original treatment threshold (so that none of the students remaining in the data set are actually assigned to treatment), and then estimating the effect of a placebo treatment ``received'' by students with test scores above the median test score value. We use either the original outcome (enrollment in any post-secondary education) or age (one of the original covariates) as the dependent variable. These two dependent variables correspond to settings in which covariate adjustments achieve almost no and quite substantial efficiency gains, respectively; see the RD estimates from the entire restricted data in Table~\ref{table::sim_full_sample_results} in the Online Supplement for details. In the main text, we conduct inference using the bias-aware approach. All simulations are based on 10,000 Monte Carlo draws.
			
\afterpage{
	\begin{table}
		\centering
		\caption{Main results for Simulation I with bias-aware inference.}\label{table::sim_baseline}
		\resizebox*{1\textwidth}{!}{
			\begin{threeparttable}
				\begin{tabular}{p{3.8cm}p{1.3cm}p{1cm}p{1cm}p{1cm}p{1cm}p{1cm}p{1.5cm}p{1.8cm}p{1.8cm}}
					\toprule
					\multicolumn{2}{l}{Adjustment Method} & SE x100 & SD x100 & Bias x100 & RMSE x100 & Band-width & CI Cov in \% & CI Length x100 &   CI Length Reduction in \%  \\ 
					\midrule
					\multicolumn{10}{l}{\textbf{Panel A - Original Outcome}}\\[1ex] 
					No Covariates &  & 2.15 & 2.17 & 0.41 & 2.21 & 37.73 & 96.95 & 9.41 & 0 \\[1ex] 
					Conventional Linear &  & 2.09 & 2.13 & 0.47 & 2.19 & 38.71 & 96.51 & 9.14 & 2.86\\[1ex] 
					Localized Linear & Feasible & 2.10 & 2.13 & 0.45 & 2.18 & 39.25 & 96.63 & 9.18 & 2.39 \\ 
					& Oracle   & 2.10 & 2.13 & 0.48 & 2.18 & 39.08 & 96.67 & 9.18 & 2.44  \\[1ex] 
					\multirow[t]{2}{*}{Flexible} & Feasible  & 2.11 & 2.13 & 0.44 & 2.18 & 39.03 & 96.69 & 9.19 & 2.28 \\ 
					& Oracle  & 2.09 & 2.12 & 0.44 & 2.16 & 39.37 & 96.73 & 9.14 & 2.81  \\[2ex] 
					\multicolumn{10}{l}{\textbf{Panel B - Age}}\\[1ex] 
				No Covariates &  & 24.95 & 25.69 & 2.87 & 25.84 & 42.25 & 96.84 & 110.79 & 0 \\[1ex] 
					Conventional Linear &  & 21.55 & 22.50 & 2.64 & 22.66 & 41.96 & 96.24 & 
					\phantom{1}95.90 & 13.44   \\[1ex] 
					Localized Linear & Feasible  & 21.84 & 22.56 & 2.64 & 22.72 & 42.09 & 96.34 & 
					\phantom{1}97.05 & 12.40\\ 
					& Oracle  & 21.71 & 22.43 & 2.69 & 22.59 & 42.07 & 96.52 & 
					\phantom{1}96.53 & 12.87 \\[1ex] 
					\multirow[t]{2}{*}{Flexible} & Feasible  & 21.35 & 22.10 & 2.65 & 22.25 & 42.25 & 96.48 & 
					\phantom{1}94.87 & 14.37  \\ 
					& Oracle  & 20.94 & 21.82 & 2.96 & 22.02 & 42.75 & 96.26 & 
					\phantom{1}93.06 & 16.00  \\ 
					\bottomrule
				\end{tabular}
				\begin{tablenotes}
					\footnotesize	\setlength\labelsep{0pt}
					\item \textit{Notes:} 
					Results are based on 10,000 Monte Carlo draws and a sample size of $n=5000$ (see Section~\ref{sec:Simulations} for details). The data-generating process is based on \cite{londono2020upstream}.					
					The bandwidth is chosen and the confidence sets are constructed based on bias-aware inference. The columns show results for  simulated mean standard error (SE);   standard deviation (SD);   bias (Bias);  root mean squared error (RMSE); the average  bandwidth (Bandwidth); 
					coverage of confidence intervals with 95\% nominal level (CI Cov); the average  confidence interval length (CI Length); and the mean confidence interval length relative to the no covariates confidence interval length (CI Length Reduction). Estimators are described in Section~\ref{sec:Covariate_adj}. 
				\end{tablenotes}
		\end{threeparttable}}
	\end{table}
}

\subsection{Simulation I: Moderate number of covariates}\label{sec:simulations_big}
In this simulation study, we evaluate the finite-sample performance of our methods in a typical RD setting with a moderate number of covariates and a relatively large number of observations. Specifically, we consider estimation with the original baseline covariates and samples of size 5000, which is around the median of the sample sizes of the empirical applications of our literature survey.
We apply our flexible covariate adjustment discussed in Section~\ref{sec:proposed_method}. Additionally, we consider the deterministic approximations of all the feasible adjustment methods, which were obtained by running the respective method on the full restricted dataset. By comparing the feasible adjustments and their respective deterministic approximation, we can assess the quality of the approximation in our equivalence result of Theorem~\ref{th:Equivalence}. For comparison, we also report the results without covariate adjustments and with  conventional linear adjustments. 
For each adjustment method, we select the bandwidth and construct a confidence interval using the bias-aware approach with the smoothness bound calibrated using the adjusted outcomes via the rule of thumb of \cite{imbens2017optimized} in each Monte Carlo draw.\footnote{The number of effective observations used in the second stage is on average around 3447 for the original outcome and around 3622 for age as the dependent variable.}
The results are based on 5-fold cross-fitting with one random data split. 

Table~\ref{table::sim_baseline} reports the main results of this simulation study. Our methods work very well for both dependent variables and all types of adjustments in that the mean simulated standard errors are close to the simulated standard deviations and the confidence intervals have simulated coverage rates close to the nominal one. The confidence intervals are slightly conservative, which is typical in bias-aware inference. The changes in the mean bias for different types of adjustments are negligible relative to the standard deviation, which is consistent with our conjecture that covariate adjustments should typically have no first-order effect on the leading bias constant. 

In Panel~A, the covariates have essentially no impact on the dependent variable, and so none of the methods leads to noticeable reductions in the standard deviation. In Panel~B, where the covariates have some explanatory power for the dependent variable, the cross-fitted RD estimator with localized linear adjustment yields a confidence interval that is on average 12\% shorter than the no covariates confidence interval. The flexible adjustment improves this performance even further. As can be expected in a setting with a small number of covariates relative to the sample size, the conventional and cross-fitted localized linear covariate adjustments yield similar results. 

In Appendix~\ref{sec:AppendixSimulation} of the Online Supplement, we present additional simulation results for all individual adjustment methods described in Section~\ref{sec:proposed_method}. We further investigate the asymptotic equivalence result presented in Theorem~\ref{th:Equivalence} within this simulation design and we illustrate that the estimation uncertainty of the adjustment functions is indeed asymptotically  negligible relative to the overall estimation uncertainty of the RD estimators. Additionally, we  present  estimation and inference results based on robust bias corrections. The qualitative conclusions remain very similar to those presented above.

\subsection{Simulation II: Many covariates}\label{sec::simulations_many_cov} This simulation setting is motivated by the fact that in some empirical applications of our literature survey, researchers estimate specifications where the ratio of the effective sample size to the number of covariates is relatively small. Indeed, in Table~\ref{table::overview_of_papers}, the minimal value of this ratio is 2.7 and  it falls below 20 in about one fifth of the specifications. 

To mimic settings where there are many covariates relative to the effective sample size, we sample 500 observations without replacement within a distance of 25 from the  placebo cutoff and use all of them in the RD regressions, i.e. we use a fixed bandwidth $h=25$. We create additional covariates by generating all second-order interaction terms of the original covariates, and we consider different settings by including the first 2, 10, 50, 100, and 150 of these covariates.\footnote{The first 21	of the technical covariates correspond to the original covariates, followed by the interaction terms. Since the covariates have essentially no explanatory power for the original outcome, the exact order of inclusion does not affect the results in this section.}
 In this setting, the ratio of the effective sample size to the number of covariates lies in the range between 250 and 3.33, which corresponds to the settings in our literature analysis with small values of this ratio. 
For each subsample, we estimate the RD parameter using the no covariates, the conventional linear adjustment,  and our cross-fitted RD estimator with localized linear adjustments and localized random forest adjustments.\footnote{We chose the random forest to represent the machine learning adjustments here, but the qualitative results are similar when employing other methods. In this section, we focus on the individual adjustment methods, rather than on the flexible ensemble, to offer more direct insights into the mechanics of the linear and regularized adjustments.} In this simulation, we calibrated the smoothness constant for each estimator and number of covariates using the full sample and we kept them fixed across simulation draws. The results are based on the bias-aware inference approach and $B=11$ data splits for each Monte Carlo draw.\footnote{In Simulation II, we calibrated the smoothness bound for each method and number of covariates only once, using the rule of thumb of \citet{imbens2017optimized} and oracle adjusted outcomes based on the full restricted sample described in Section~\ref{sec:simulations_big}.}

\subsubsection{Results on efficiency} 

\begin{figure}[!h]
	\centering
	\includegraphics{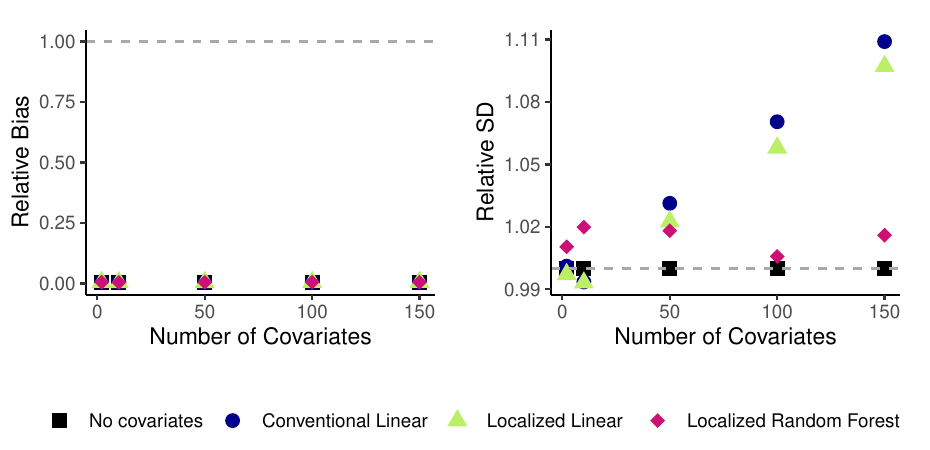}
	\vspace{-3pt}	

	\caption{Results for Simulation II - Bias and standard deviation of the respective estimator relative to the standard deviation of the no covariates RD estimator in the left and right panels, respectively. The results are based on samples of size $n=500$ within the estimation window with $h=25$ and 10,000 Monte Carlo draws.}\label{fig:sim_manyn_sd}
\end{figure}

Figure~\ref{fig:sim_manyn_sd} shows the bias and the standard deviation of the four estimation methods,  normalized by the standard deviation of the no covariates RD estimator, for a varying number of covariates.
We note that the simulated bias is insensitive to including many covariates, and we therefore focus on the standard deviation. 
In this setting, the covariates seem to have essentially no explanatory power for the dependent variable, and so adjustments based on them cannot lead to a reduction in the asymptotic variance of the RD estimator; see estimation results in Table~\ref{table::sim_full_sample_results}. As predicted by our theory, when the number of covariates remains moderate, all estimators perform very similarly, meaning that all the adjustments concentrate around the optimal function of no adjustment.
However, as the number of covariates increases, the standard deviations of both the conventional and cross-fitted localized linear adjustment estimators become substantially larger than those of the no covariates RD estimator. 
The reason for that is that the linear regression with a large number of covariates is very variable, such that the estimated adjustments are no longer close to a constant\footnote{Such finite-sample behavior renders our asymptotic theory as well as the results of \citet{calonico2019regression} inapplicable in this setting.} and the high-dimensional linear adjustments effectively add non-negligible noise to the outcome variable in this setting.

In contrast, the RD estimator with random forest adjustments, due to built-in regularization,  does not become more variable as the number of covariates increases, meaning that the estimated adjustment function remains close to the optimal function of no adjustment. In general, it is therefore advisable to always rely on regularized adjustments in high-dimensional settings.

\subsubsection{Results on inference}

\begin{figure}[!h]
	\centering	
	\includegraphics{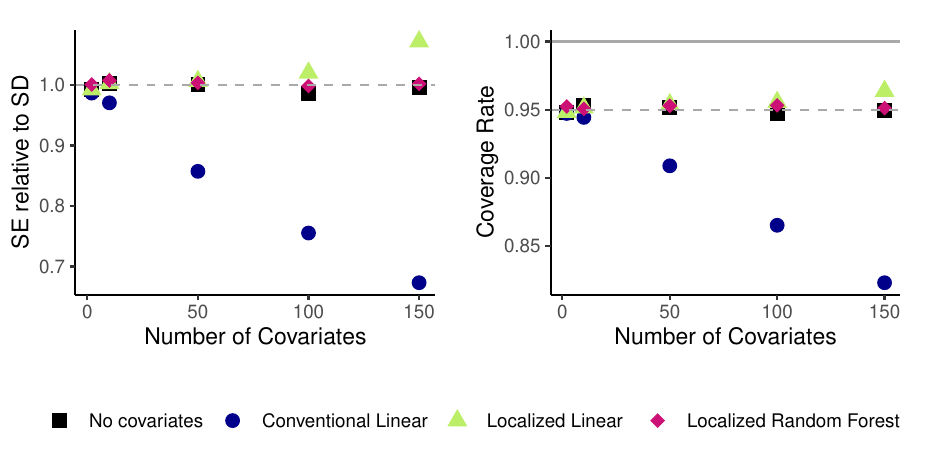}
	\vspace{-3pt}	
	
	\caption{Results for Simulation II - Mean standard error relative to the standard deviation of the respective estimator and simulated confidence interval coverage for nominal confidence level 95\%. We consider bias-aware confidence intervals and nearest neighbor standard errors.  The results are based on samples of size $n=500$ within the estimation window with $h=25$ and 10,000 Monte Carlo draws.}\label{fig:sim_manyn_inference}
\end{figure}

We now turn to the standard error and coverage of the confidence intervals for the respective methods. 
The left panel of Figure~\ref{fig:sim_manyn_inference} shows that the standard  error of the conventional linear adjustment estimator exhibits a downward bias that increases substantially with the number of covariates, with its ratio to the estimator's standard deviation reaching less than 70\% for 150 covariates. This effect is due to overfitting: with many covariates, the regression residuals that enter the standard error formula become ``too close to zero'', and standard errors therefore become ``too small''. With cross-fitting, this issue occurs neither for the linear adjustment nor for the random forest adjustment.

The right panel of Figure~\ref{fig:sim_manyn_inference} shows that, due to  increasingly biased standard errors, the coverage of the conventional linear adjustment bias-aware confidence intervals with nominal level $95\%$ falls below 85\% for 150 covariates. With cross-fitting, bias-aware confidence intervals have close to nominal coverage for both adjustment methods and all numbers of covariates under consideration. 
These simulation results demonstrate that cross-fitting yields consistent standard errors and valid inference even in high-dimensional settings where conventional methods may fail.

\section{Conclusions}\label{sec:Conclusions}
We have proposed a novel class of estimators that can make use of covariate information more efficiently than the conventional linear adjustment estimators that are currently  used widely in practice. In particular, our approach allows the use of modern machine learning tools to adjust for covariates, and is at the same time largely unaffected by the ``curse of dimensionality''. Our estimator is also easy to implement in practice, and can be combined in a straightforward manner with existing methods for bandwidth choice and the construction of confidence intervals. In our reanalysis of the literature, we show that our proposed estimator yields shorter confidence intervals in almost all empirical applications, and in some cases, these reductions can be substantial. We  therefore expect our proposed estimator to be very attractive for a wide range of future economic applications.

\appendix
\section{Proofs of the Main Results}\label{A:main}
In this section, we prove Theorems~\ref{th:Equivalence}--\ref{th:bandwidth} and Proposition~\ref{prop:fuzzy}. To this end, we show a more general result that allows for a local polynomial regression of an arbitrary order $p$. We use this result also in Appendix~\ref{sec:Inference} to establish the validity of the inference methods discussed in Section~\ref{sec:Inference_Overview}.

\subsection{Additional Notation}\label{A:sec:weights}
Let $\Xn=(X_i)_{i\in[n]}$ denote the realizations of the running variable. For $0\leq v \leq p$, we define feasible and infeasible estimators of the jump in the $v$-th derivative of the conditional expectation of the modified outcome $M(\bar\eta)$ at the cutoff using the $p$-th order local polynomial regression as
\begin{align*}
	\htau_{v,p}(h; \wh\ad) &= \Sum w_{i,v,p}(h) M_i(\wh\ad_{s(i)}) \quad\text{and}\quad	\wh{\tau}_{v,p}(h;\bar\ad) = \Sum w_{i,v,p}(h) M_i(\bar\ad),\\
	w_{i,v,p}(h)&=w_{i,v,p}^+(h)-w_{i,v,p}^-(h),\\
	w_{i,v,p}^\star(h) &= e_{v+1}^\top \left(\Sum K^\star_h(X_i) \widetilde{X}_{p,i} \widetilde{X}_{p,i}^\top\right)^{-1} K^\star_h(X_i) \widetilde{X}_{p,i} \quad\text{ for } \star \in \{+,-\},
\end{align*}
where $\widetilde{X}_{p,i}=(1,X_i/(1!),\dots,X_i^p/(p!))^\top$, $K_h(v)=K(v/h)/h$, $K^+_h(v)=K_h(v)\1{v\geq 0}$,  $K^-_h(v)=K_h(v)\1{v<0}$. The corresponding estimates of  $\beta^\star_v(\bar\ad) = \partial_x^v\E[M_i(\bar \ad)|X_i=x]|_{x=0^\star}$ are given by 
\[
\wh\beta^\star_{v,p}(h; \wh\ad) = \Sum w^\star_{i,v,p}(h) M_i(\wh\ad_{s(i)}) \quad\text{and}\quad 	\wh\beta^\star_{v,p}(h;\bar\ad) = \Sum w^\star_{i,v,p}(h) M_i(\bar\ad) \quad\text{for } \star \in \{+,-\}.
\]

\subsection{General Result}
We state and prove two theorems that generalize our Theorems~\ref{th:Equivalence} and~\ref{th:Normality}.

\begin{theorem}\label{th:Equivalence_General}
	Suppose that Assumptions~\ref{ass:1stage}--\ref{ass:reg1b} hold with $j\in\{1,...,p+1\}$ in Assumption~\ref{ass:derivatives} and functions in $\mathcal E$ are $p+1$ times continuously differentiable. Then
	$$
	\htau_{0,p}(h; \hmu)  = \wh{\tau}_{0,p}(h;\bar\ad) + O_P(t_p),
	$$
	where $t_p = r_n (nh)^{-1/2}  + \sum_{j=1}^p v_{j,n} h^j(nh)^{-1/2} + v_{p+1,n}h^{p+1} $.
\end{theorem}

In the proof of Theorem~\ref{th:Equivalence_General}, we will use the following lemma that collects some standard intermediate steps in the analysis of local polynomial estimators, taking into account cross-fitting.
\begin{lemma}\label{lemma:X} Suppose that Assumptions~\ref{ass:reg1a}~and~\ref{ass:reg1b} hold. For $s \in [S]$ and $\star\in\{-,+\}$, it holds that: 
	\begin{enumerate}[label=(\roman*),wide, labelindent=0pt]
		\item $\displaystyle \frac{S}{n}\sum_{i \in I_s} K_h^\star(X_i)(X_i/h)^j = \E[K_h^\star(X_i)(X_i/h)^j] + O_p((nh)^{-1/2})$  for  $j \in \mathbb{N}$,
		\item $\displaystyle\sum_{i \in I_s} w_{i,0,p}^\star(h) = 1/S + O_p((nh)^{-1/2})$,
		\item $\displaystyle\sum_{i \in I_s} w_{i,0,p}^\star(h)X_i^j =  O_p(h^j(nh)^{-1/2})$ for $1 \leq j \leq p $,
		\item $\displaystyle\sumI |w_{i,0,p}^\star(h)X_i^j|=O_P(h^j)$ for  $j \in \mathbb{N}$,
		\item $\displaystyle\sumI w_{i,0,p}^\star(h)^2 = O_P((nh)^{-1})$.
	\end{enumerate}
\end{lemma}
\begin{proof}
	The results follow from standard kernel calculations.
\end{proof}

\begin{proof}[Proof of Theorem~\ref{th:Equivalence_General}]
	To begin with, note that
	\[
	\wh{\tau}_{0,p}(h;\bar\ad) - \htau_{0,p}(h;\hmu) = \sum_{s=1}^S G_s(p),\quad G_s(p) \equiv \sum_{i \in I_s}w_{i,0,p}(h)(\wh\ad_s(Z_i) - \bar\ad(Z_i)).
	\]
	Since $S$ is a fixed number, it suffices to show that $G_s(p)=O_p(t_p)$ for $s \in \{1,...,S\}$. 	
	We analyze the expectation and variance of $G_s(p)$ conditional on $\Xn$ and $(W_j)_{j \in I^c_s}$.  
	We begin with the expectation. It holds with probability approaching one that 	
	\begin{align*}
		\left|\E[G_{s}(p) | \Xn, (W_j)_{j \in I^c_s}]\right| 
		& =  \Big|\sum_{i \in I_s} w_{i,0,p}(h) \E[\wh\ad_{s}(Z_i) - \bar\ad(Z_i)  |X_i, (W_j)_{j \in I^c_s}] \Big| \\
		& \leq \sup_{\ad\in\Tn} \Big|  \sum_{i \in I_s} w_{i,0,p}(h) \E[ \ad(Z_i) - \bar\ad(Z_i) |X_i] \Big|.
	\end{align*}
	Let $m(x;\ad) = \E[\ad(Z_i)-\bar\ad(Z_i)|X_i=x]$.
	Taylor's theorem yields
	\[
	m(X_i;\ad)= m(0;\ad)  + \sum_{j=1}^p \frac{1}{j!}  \partial^j_x m(0;\ad)X_i^j + \frac{1}{(p+1)!}\partial^{p+1}_x m(\wt{x}_{i,p};\ad)X_i^{p+1}
	\]
	for some $\wt{x}_{i,p}$ between $0$ and $X_i$. We analyze the three terms associated with different terms of Taylor's expansion separately. We make use of Lemma~\ref{lemma:X} in each step.
	
	First, using the Cauchy-Schwarz inequality, we obtain that
	\[
	\sup_{\ad\in\Tn} \Big| m(0;\ad)   \sum_{i \in I_s} w_{i,0,p}(h)   \Big| = \sup_{\ad\in\Tn} | m(0;\ad)| O_p((nh)^{-1/2})= O_p(r_n (nh)^{-1/2}).
	\]

	Second, for $j \in \{1,...,p\}$, we have that
	\[
	\sup_{\ad\in\Tn} \Big| \partial^j_x m(0;\ad)   \sum_{i \in I_s} w_{i,0,p}(h)X_i^j   \Big| = \sup_{\ad\in\Tn} | \partial^j_x m(0;\ad) | h^j  O_p((nh)^{-1/2}) =  O_p(h^j(nh)^{-1/2}v_{j,n}).
	\]
	
	Third, we note that
	\begin{align*}
		\sup_{\ad\in\Tn} \Big|   \sum_{i \in I_s} w_{i,0,p}(h) \partial^{p+1}_x m(\wt{x}_{i,p};\ad) X_i^{p+1} \Big|  
		\leq   \sum_{i \in I_s} |w_{i,0,p}(h)X_i^{p+1}| \sup_{\ad\in\Tn} |\partial^{p+1}_x m(\wt{x}_{i,p};\ad)| = O_p(h^{p+1}v_{p+1,n}).
	\end{align*}
	
	Next, we consider the conditional variance. It holds with probability approaching one  that 
	\begin{align*}
		\V\left[G_s(p)| \Xn, (W_j)_{j \in I^c_s}\right] 
		& = \sum_{i \in I_s} w_{i,0,p}(h)^2\V\left[\bar\ad(Z_i)  - \wh\ad_{s}(Z_i)|\Xn, (W_j)_{j \in I^c_s}\right]\\
		& \leq \sup_{\ad\in\Tn}  \sum_{i \in I_s}w_{i,0,p}(h)^2\E[(\bar\ad(Z_i) - \ad(Z_i))^2 |X_i] \\
		& \leq \sup_{\ad\in\Tn} \sup_{x \in \Xh} \E[(\bar\ad(Z_i) - \ad(Z_i))^2 |X_i=x] \sum_{i \in I_s}w_{i,0,p}(h)^2 \\
		& = O_p(r_{n}^{\; 2}  (nh)^{-1} ),
	\end{align*}
	where we use Lemma~\ref{lemma:X} and Assumption~\ref{ass:1stage} in the last step. The conditional convergence then implies the unconditional one \citep[see][Lemma 6.1]{chernozhukov2018double}. 
\end{proof}

\begin{theorem}\label{th:Normality_General}
	Suppose that the assumptions of Theorem~\ref{th:Equivalence_General} hold, Assumption~\ref{ass:reg2} holds, and $\E[M_i(\bar\ad)|X_i=x]$ is $p+1$ times continuously differentiable with an $L$-Lipschitz continuous $(p+1)$-st derivative bounded by $C$ on each side of the cutoff. Then 
	$$
	\sqrt{nh}\, V_p(\bar\ad)^{-1/2} \left( \htau_{0,p}(h; \hmu)  - \tau - h^{p+1} B_{p}   \right) \overset{d}{\to}  \mathcal{N}(0, 1),
	$$
	where, for some kernel constants $\bar{\nu}_p$ and $\bar{\kappa}_p$,
	\begin{align*}
		B_{p}  & =\frac{ \bar{\nu}_p}{(p+1)!}  \left(\partial_x^{p+1} \E[ M_i(\bar\ad) |X_i=x]\big|_{x=0^+} + (-1)^p \partial_x^{p+1} \E[ M_i(\bar\ad) |X_i=x]\big|_{x=0^-}\right) + o_P(1),\\
		V_p(\bar\ad) & =\frac{ \bar{\kappa}_p }{f_X(0)}  \left( \V[ M_i(\bar\ad)|X_i=0^+] + \V[ M_i(\bar\ad)|X_i=0^-] \right).\\
	\end{align*}
\end{theorem}

\begin{proof}[Proof of Theorem~\ref{th:Normality_General}]
	By the conditional version of Lyapunov's CLT, we obtain that 
	\[
	\text{se}_{0,p}(h;\bar\ad)^{-1}\left(\wh{\tau}_{0,p}(h;\bar\ad) - \E[\wh{\tau}_{0,p}(h;\bar\ad)|\Xn]\right) \to \mc N(0,1).
	\]
	where $\text{se}_{0,p}^2(h;\bar\ad) = \Sum w_{i,0,p}(h)^2 \V[M_i(\bar\ad)|X_i]$.
	By $L$-Lipschitz continuity of $\V[M_i(\bar\ad)|X_i=x]$ in $x$, we obtain that
	\begin{align*}
		\text{se}_{0,p}^2&(h;\bar\ad) = \Sum w_{i,0,p}^-(h)^2 \V[M_i(\bar\ad)|X_i=0^-] + \Sum w_{i,0,p}^+(h)^2 \V[M_i(\bar\ad)|X_i=0^+] + o_p((nh)^{-1}).
	\end{align*}
	It then follows from standard kernel calculations that $nh\,\text{se}_{0,p}^2(h;\bar\ad)-V_p(\bar\ad)=o_P(1)$ and $\E[\wh{\tau}_{0,p}(h;\bar\ad)|\Xn]-\tau = B_ph^{p+1} + o_p(h^{p+1})$ for some constant $B_p$.
\end{proof}

\subsection{Proofs of Theorems~\ref{th:Equivalence}--\ref{th:bandwidth}} 
Theorems~\ref{th:Equivalence} and~\ref{th:Normality} follow directly from the general results in Theorems~\ref{th:Equivalence_General} and~\ref{th:Normality_General} with $p=1$; and Theorem~\ref{th:bandwidth} follows from simple calculations. It remains to prove Theorem~\ref{th:var}. For any $\ad\in\ClassVar$, it holds that
$$
\frac{f_X(0)}{\bar\kappa} V(\ad)  =  \V[Y_i-\mu_0^+(Z_i)|X_i=0^+] + \V[Y_i-\mu_0^-(Z_i)|X_i=0^-] + R(\ad),
$$
where the first two terms on the right-hand side do not depend on $\ad$, and	
$$
R(\ad) =  \V[\mu_0^+(Z_i) -\ad(Z_i) |X_i=0^+] +  \V[\mu_0^-(Z_i) - \ad(Z_i)|X_i=0^-].
$$
Further, it holds that
\begin{align*}
	R(\ad) = R(\ad_0+ \ad - \ad_0 ) & = \V\left[ \frac{1}{2}(\mu_0^+(Z_i)-\mu_0^-(Z_i)) - (\ad(Z_i)-\ad_0(Z_i))  |X_i=0^+ \right]\\
	&\quad +  \V\left[ - \frac{1}{2}(\mu_0^+(Z_i)-\mu_0^-(Z_i)) - (\ad(Z_i)-\ad_0(Z_i))  |X_i=0^- \right] \\
	&  = R(\ad_0) + 2 \V[\ad(Z_i)-\ad_0(Z_i)|X_i=0],
\end{align*}
where in the last step we use the assumption on continuity of conditional covariances. The theorem follows from the above decomposition by taking the difference $V(\ad^{(a)}) - V(\ad^{(b)})$ for arbitrary $\ad^{(a)}$ and $\ad^{(b)}$ in $\ClassVar$.
\qed

\subsection{Proof of Proposition~\ref{prop:fuzzy}}\label{sec::appendixfuzzy}
We first note that 
$$\htheta(h; \wh \ad_Y,\wh \ad_T) - \htheta(h; \bar \ad_Y, \bar \ad_T) = O_P(r_n (nh)^{-1/2}  + v_{1,n} h(nh)^{-1/2} + v_{2,n}h^2 ). $$
This equality is an immediate consequence of Theorem~\ref{th:Equivalence} and an application of the continuous mapping theorem as $|\tau_T|>0$. Further, using a mean-value expansion, it follows that  
$$\htheta(h;  \bar\ad_Y, \bar\ad_T) - \theta = \frac{1}{\tau_T}( \wh \tau_Y(h; \bar \ad_Y ) -\tau_Y ) - \frac{\tau_Y}{\tau_T^2}( \wh \tau_T(h; \bar \ad_T ) -\tau_T ) + \wh \rho( \bar \ad_T,  \bar \ad_Y) $$
with
$$
\wh \rho( \bar \ad_T,  \bar\ad_Y) =\frac{\widehat\tau_Y(h;  \bar \ad_Y)(\widehat\tau_T(h;  \bar \ad_T)-\tau_T)^2}{\widehat\tau_T^*(h; \bar \ad_T)^3}-\frac{(\widehat\tau_Y(h; \bar \ad_Y)-\tau_Y)(\widehat\tau_T(h; \bar \ad_T)-\tau_T)}{ \tau_T^2}, 
$$
where $\wh\tau_T^*(h; \ad_T)$ is some intermediate value between $\tau_T$ and $\wh \tau_T(h; \bar \ad_T)$.
Given our assumptions, it follows that
$$ \wh \rho(\bar \ad_T, \bar \ad_Y) = O_P\big(((nh)^{-1/2}   + h^2)^2\big).$$
Part~(i) follows analogously to  Theorems~\ref{th:Equivalence} and~\ref{th:Normality} and Part~(ii) follows from Theorem~\ref{th:var}.\qed

\section{Details on Section~5.1}\label{sec:Inference}
In this section, we formally show that, under suitable assumptions, existing procedures for bandwidth selection and construction of confidence intervals devised for settings without covariates can be directly applied to the modified data $\{(X_i,M_i(\wh\ad))\}_{i\in[n]}$.

\subsection{Standard Errors}\label{sec:standarderror}
We generalize the nearest-neighbors standard error from the main text to the local polynomial regression of an arbitrary order $p$. Let 
$$ 
\wh{se}^2_{v,p}(h; \hmu)  = \sum_{i =1}^n w_{i, v,p}^2(h)  \wh \sigma_{i}^2(\wh \ad) , \quad 	\wh{\sigma}^2_{i}(\wh\ad) =  \frac{R}{R+1} \Big( M_i(\wh\ad_{s(i)})  -  \frac{1}{R} \sum_{j \in \mc R_i}  M_j(\wh\ad_{s(j)})  \Big)^2,
$$
where  $\mc R_i$ is the set of $R$ nearest neighbors of unit $i$ in terms of their running variable realization on the respective side of the cutoff. Establishing consistency of this standard error requires the following technical assumption on the first-stage estimator, which is implied by our main assumptions, for example,  if  $M_i(\bar\ad)$ is bounded.

\begin{assumption}\label{ass:technical_se_nn}
	For all $s \in [S]$, it holds that $\Sumn  w_{i,v,p}^2(h) \iota_i(\widehat \eta)  = o_P( (nh^{1+2v})^{-1})$ for $0 \leq v \leq p$, where 
	$$ 
	\iota_i(\widehat \eta) =   \sum_{\substack{ (j, l) \in \mc R_i^2\\ (j,l) \notin I^2_{s(i)}}}  \left( (\wh \ad_{s(i)}(Z_i) - \bar \ad(Z_i))  -   (\wh \ad_{s(j)}(Z_j) - \bar \ad(Z_j))\right)(M_i(\bar \ad) -   M_l(\bar \ad)) .	$$
\end{assumption}

\begin{proposition}\label{prop:standarderror}	
	Suppose that Assumptions~\ref{ass:1stage}--\ref{ass:reg2} and~\ref{ass:technical_se_nn} hold. Moreover, suppose that Assumption~\ref{ass:1stage} also holds with $\mc X_h$ replaced by $\widetilde{\mc X}_h$ that is an open set s.t. $ \mc X_h  \subset \widetilde{\mc X}_h $, and $  \sup_{\ad \in \Tn} \sup_{x \in \wt{\mc X}_h}  \E[( M_i(\ad) - \E[M_i(\ad) |X_i])^4| X_i=x]$ is bounded by $B$ for all $n \in \mathbb{N}$. Suppose further that $\E[M_i(\bar\ad) |X_i=x]$ and  $\V[M_i(\bar\ad) |X_i=x]$ are L-Lipschitz continuous on each side of the cutoff. 
	Then for all $0 \leq v \leq p$, it holds that 
	$$
	nh^{1+2v} \left( \hs_{v,p}^2(h; \wh \ad) - se^2_{v,p}(h;\bar\ad)\right) = o_P(1),
	$$
	where $se^2_{v,p}(h;\bar\ad)  = \sum_{i =1}^n w_{i, v,p}^2(h)  \sigma_{i}^2(\bar \ad)$ and ${\sigma_i^2(\bar\ad)= \V [M_i(\bar\ad)|X_i ] } $.
\end{proposition}

We note that Assumption~\ref{ass:technical_se_nn} could be dropped if we were to study a slight variation
of $\wh{se}^2_{v,p}(h; \hmu)$ in which we take the $R$ nearest neighbors of unit $i$ in terms of running variable values \emph{among units in the same fold} to compute $\wh{\sigma}^2_{i}(\wh\ad)$. However, proceeding like this would mean that existing software packages that compute nearest neighbor standard errors would have to be adapted, and could not be applied directly to the modified data $\{(X_i,M_i(\wh\ad))\}_{i\in[n]}$.

\subsection{Confidence Intervals}\label{sec:confidencesets} In this subsection, we discuss three types of confidence intervals for the RD parameter, based on undersmoothing, robust bias correction, and bias-aware critical values, respectively.

\subsubsection{Undersmoothing}
We first consider confidence intervals that are based on an undersmoothing bandwidth of order  $o(n^{-1/5})$. This choice of bandwidth implies that the  smoothing bias shrinks to zero at a faster rate than the standard deviation and can hence be ignored when constructing confidence intervals.  Let
$$ 
CI^{us}_{1-\alpha} = \big[\htau(h; \wh \ad) \pm z_\alpha \hse( h; \wh \ad)\big], 
$$
where $z_\alpha$ is the $(1-\alpha/2)$-quantile of the standard normal distribution. Proposition~\ref{prop::CIus} shows that $CI^{us}_{1-\alpha}$ is asymptotically valid.

\begin{proposition}\label{prop::CIus}
	Suppose that the assumptions of Proposition~\ref{prop:standarderror} hold for $p=1$. If $nh^5=o(1)$, then
	$
	\mathbb P\big(\tau \in CI^{us}_{1-\alpha}\big) \geq 1-\alpha+o(1).
	$	 
\end{proposition}
\subsubsection{Robust bias correction}
 We now adapt the robust bias corrections of \citet{calonico2014robust} to our setting.
To keep the exposition transparent, we focus on the important special case where the bandwidth used to obtain the bias correction is the same as the main bandwidth. In this case, the local linear estimator with a bias correction is numerically equal to the local quadratic estimator (with the same bandwidth), i.e. $ \wh \tau_{0,2}(h; \wh \ad)$. Let
$$ 
	CI^{rbc}_{1-\alpha} = \big[\htau_{0,2}(h; \wh \ad) \pm z_\alpha \hse_{0,2}( h; \wh \ad)\big].
$$
Proposition~\ref{prop::CI_rb} shows that $CI^{rbc}_{1-\alpha}$ is asymptotically valid.

\begin{proposition}\label{prop::CI_rb}
	Suppose that the assumptions of Theorem~\ref{th:Normality_General} and~Proposition~\ref{prop:standarderror} hold for $p=2$. If $nh^7=o(1)$, then
$
\mathbb P\big(\tau \in CI^{rbc}_{1-\alpha}\big) \geq 1-\alpha+o(1).
$		 
\end{proposition}

\subsubsection{Bias-awareness}
We consider a version of the bias-aware approach of \citet{armstrong2018optimal} which adjusts the critical values so as to account for the maximal possible smoothing bias of $\wh \tau(h; \bar \ad)$.
Suppose  that  $|\partial^2_x \mathbb E[M_i(\bar\ad)|X_i=x]|$  is bounded  by some constant $B_M$ on either side of the cutoff. Then it follows from the results of \citet{armstrong2020simple} and our Theorem~\ref{th:Normality} that the asymptotic bias of our covariate-adjusted RD estimator is bounded in absolute value  by $\bar b(h) + o_P(h^2)$, where
$$ \bar b(h) =- \frac{B_M}{2} \sum_{i=1}^n w_i(h) X_i^2 \sgn(X_i).$$ 
For implementation, the smoothness bound $B_M$ can be selected by inspecting the scatter plot of  $\{(X_i,M_i(\wh\ad))\}_{i\in[n]}$; see \citet{kolesar2018discrete}, \citet{imbens2017optimized}, \citet{armstrong2020simple}, and \citet{noack2021bias} for a discussion of various procedures for choosing the smoothness bound.
The proposed confidence interval is
$$ 
CI^{ba}_{1-\alpha} = \left[\htau(h; \wh \ad) \pm z_\alpha(\bar b(h)/\hse_{0,1}( h; \wh \ad)) \,\hse_{0,1}( h; \wh \ad)\right],
$$
where $z_\alpha(r)$ is the $(1-\alpha)$-quantile of $|N(r,1)|$, the folded normal distribution with mean $r$ and variance one. 
Proposition~\ref{prop::CIba} shows that $CI^{ba}_{1-\alpha}$ is asymptotically valid.
\begin{proposition}\label{prop::CIba}
	Suppose that the assumptions of Proposition~\ref{prop:standarderror} hold for $p=1$. If $nh^5=O(1)$, then
	$
	\mathbb P\big(\tau \in CI^{ba}_{1-\alpha}\big) \geq 1-\alpha+o(1).
	$		 
\end{proposition}

We note that the above bias-aware approach accounts for the smoothing bias separately for the regressions on either side of the cutoff, in line with the original proposal of \citet{armstrong2018optimal} for settings without covariates. This procedure does not rely on the asymptotic bias formula in Theorem~\ref{th:Normality}, where the bias contributions of the adjustment terms on both sides cancel out asymptotically.
To exploit this cancellation of asymptotic biases, one could construct a confidence interval using a bound on  $|\partial^2_x \mathbb E[Y_i|X_i=x]|$ instead of $|\partial^2_x \mathbb E[M_i(\bar\eta)|X_i=x]|$ in the expression of the worst-case bias $\bar b(h)$. Such a confidence interval, however, can be expected to be valid only in the pointwise-in-DGP asymptotics.\footnote{Suppose that both $\E[Y_i|X_i=x]$ and $\E[\bar\ad(Z_i)|X_i=x]$ are twice continuously differentiable on either side of the cutoff, with their second derivatives being bounded in absolute value by constants $B_Y$ and $B_\eta$, respectively. Then one can show that the conditional bias $\E[\wh\tau(h;\bar\ad)|\mathcal{X}]-\tau$ is bounded in absolute value uniformly over these function classes by $- (B_M/2) \sum_{i=1}^n w_i(h) X_i^2 \sgn(X_i)$, with $B_M=B_Y+B_\eta$. Under pointwise asymptotics, the contribution of the adjustment term to the bias vanishes in large samples under our assumptions, and in finite samples we expect this ``worst case'' bound to generally be somewhat pessimistic and the contribution to be relatively small. We still recommend using this general bound in practice.}

\subsection{Optimal Bandwidth}\label{sec:bandwidthestimation}
In our Theorem~\ref{th:bandwidth}, we show that the bandwidth that minimizes the approximate mean squared error (AMSE) of our proposed estimator is 
\[
	h_{AMSE} = \left(\frac{V(\bar\ad)}{4B_{\textnormal{base}}^2}\right)^{1/5} n^{-1/5}.
\]
This optimal bandwidth can be consistently estimated by applying the procedure of \citet[][S.2.6]{calonico2014robust} to the modified data $\{(X_i,M_i(\wh\ad_{s(i)}))\}_{i\in[n]}$ using the following three steps.

\noindent\textit{Step 0.} Initial bandwidths.
\begin{itemize}
	\item [(i)] Let $v_n$ be such that $v_n \to 0$ and $n v_n \to \infty$. In practice, one can set
	$	\wh v_n = 2.58 \min\{ S_X, IQR_X/1.349 \} n^{-1/5},$
	where $S_X^2$ and $IQR_X$ denote, respectively, the sample variance and interquartile range of $\{X_i:1\leq i \leq n\}$.
	
	\item [(ii)] Choose $c_n$ such that $c_n \to 0$ and $n c_n^7 \to \infty$. In practice, let
	\begin{align*}
		\wh c_n  = \wh C_n^{1/9} n^{-1/9},\qquad
		\wh C_n  = \frac{7 n v_n^7 \wh{se}^2_{3,3}(v_n;\wh\ad)}{2 \mc B_{3,3}^2\left(\wh\gamma^+_{4,4}(\wh\ad) -  \wh\gamma^-_{4,4}(\wh\ad) \right)^2 },
	\end{align*}
	where  $\wh \gamma^\star_{4,4}(\wh\ad)$ is the coefficient on $(1/4!)X_i^4$ in the  fourth-order global polynomial regression of $M_i(\wh\ad_{s(i)})$ on a constant, $X_i$, $(1/2!)X_i^2$, $(1/3!)X_i^3$, and $(1/4!)X_i^4$, using the data on the respective side of the cutoff, and $\mc B^\star_{v,p}$ for $\star \in \{+,-\}$ is the kernel constant in the leading bias term of $\wh\beta^\star_{v,p}(h;\wh\ad)$.
\end{itemize}

\noindent\textit{Step 1.} Choose a pilot bandwidth $b_n$ such that $b_n \to 0$ and $ n b_n^5 \to \infty $. In practice, use the following estimate of the bandwidth that minimizes the AMSE of the estimates of the second derivative terms in a local quadratic regression:
\begin{align*}
	\wh b_n = \wh B_n^{1/7}n^{-1/7},\qquad
	\wh B_n = \frac{5 n v_n^5 \wh{se}^2_{2,2}(v_n;\wh\ad)}{ 2 \mc B^2_{2,2} \Big( \left( \wh \beta^+_{3,3}(c_n; \wh\ad) + \wh \beta^-_{3,3}(c_n;\wh\ad)\right)^2 + 3 \wh{se}^2_{3,3}(c_n;\wh\ad)    \Big)   }.
\end{align*}

\noindent\textit{Step 2.} Estimate $h_{AMSE}$ by
\begin{align*}
	\wh h_n = \wh H_n^{1/5} n^{-1/5},\qquad
	\wh H_n = \frac{n v_n \wh{se}^2_{0,1}(v_n;\wh\ad)}{ 4 \mc B^2_{0,1} \Big( \left( \wh \beta^+_{2,2}(b_n;\wh\ad) -  \wh \beta^-_{2,2}(b_n;\wh\ad) \right)^2 + 3 \wh{se}^2_{2,2}(b_n;\wh\ad) \Big) }.
\end{align*}

\begin{proposition}\label{prop:bandwidth}
	Suppose that the assumptions of Theorem~\ref{th:Normality_General} and Proposition~\ref{prop:standarderror} hold for $p=3$, $\mc X$ is bounded, $\IP \big[1/C \leq |\wh\gamma^+_{4,4}(\bar\ad)-\wh\gamma^-_{4,4}(\bar\ad)|\leq C ] \to 1$ for some $C>0$, and Assumption~\ref{ass:1stage} holds with $\Xh$ replaced by $\X$. Suppose that $\beta^+_{v}(\bar\ad) + (-1)^{v+1} \beta^-_{v}(\bar\ad)$ is bounded and bounded away from zero for $v \in \{2,3\}$. Then
	$\wh c_n \overset{p}{\to} 0,\, n \wh c_n^{\,7} \overset{p}{\to} \infty,\, \wh b_n \overset{p}{\to} 0,\, n \wh b_n^{\,5} \overset{p}{\to} \infty$, and  
	$\wh h_n/h_{AMSE} \overset{p}{\to} 1.$
\end{proposition}

\subsection{Proofs of Propositions~\ref{prop:standarderror}--\ref{prop:bandwidth}}

\subsubsection{Proof of Proposition~\ref{prop:standarderror}} To begin with, note that standard kernel calculations show that: 
\begin{enumerate*}[ label=(\roman*)]
	\item $\sum_{i \in [n] } w_{i,v,p}(h)^2 = O_P( (nh^{1+2v})^{-1})$ and
	\item $\max_{i \in [n]} w_{i,v,p}(h)^2  = o_P( (nh^{1+2v})^{-1})$.
\end{enumerate*}
The proof of Proposition~\ref{prop:standarderror} then requires showing that $\hs_{v,p}^2( h; \wh\ad )$ is asymptotically equivalent to the following infeasible version of itself, which uses the deterministic function $\bar\ad$. Let $r=R/(R+1)$:
\[
 \hs_{v,p}^2( h; \bar\ad ) =	\Sumn w_{i,v,p}^2(h) r \bigg(M_i(\bar \ad)  -  \Sumr M_j(\bar \ad)\bigg)^2.
\]
Using arguments as in the proof of  Theorem~4 in \citet{noack2021bias}, one can show that  $\hs_{v,p}^2( h; \bar\ad ) -  \text{se}_{v,p}^2( h; \bar\ad )  = o_P((nh^{1+2v})^{-1})$.
It therefore remains to show that $\hs_{v,p}^2( h; \wh\ad ) -  \hs_{v,p}^2( h; \bar\ad ) = o_P( (nh^{1+2v})^{-1})$. We express this difference as the sum of terms that are linear in $M_i(\wh \ad_{s(i)}) - M_i(\bar \ad) = \bar \ad(Z_i) - \wh \ad_{s(i)}(Z_i)$ and a quadratic remainder:
	\begin{align*}
		&\hs_{v,p}^2( h; \wh\ad ) -  \hs_{v,p}^2( h; \bar\ad ) \\   
		&\;=	2 \Sumn w_{i,v,p}^2(h) r \bigg(M_i(\bar \ad)  -  \Sumr M_j(\bar \ad)\bigg) \bigg(M_i(\wh \ad_{s(i)}) - M_i(\bar \ad)  -  \Sumr (  M_j(\wh \ad_{s(j)}) -  M_j(\bar \ad))\bigg)\\
		&\;\quad + \Sumn w_{i,v,p}^2(h) r \bigg(M_i(\wh \ad_{s(i)}) - M_i(\bar \ad)  -  \Sumr (  M_j(\wh \ad_{s(j)}) -  M_j(\bar \ad))\bigg)^2  \\
		&\;\equiv 2 A_1 + A_2.
	\end{align*} 

We first consider $A_2$. Let $C$ denote a generic constant that might change from line to line. It holds that
	\begin{align*}
	\frac{1}{C}	A_2 &  \leq   \sum_{i=1}^n  w_{i,v,p}^2(h) \bigg( (\wh \ad_{s(i)}(Z_i) - \bar \ad(Z_i) )^2  +   \Sumr (\wh \ad_{s(j)}(Z_j) - \bar \ad(Z_j))^2  \bigg)  \\
	& \leq   \sum_{i=1}^n  \bigg( w_{i,v,p}^2(h) + \frac{	C}{R} \sum_{j:\; i \in R_j}w_{j,v,p}^2(h)\bigg)  (\wh \ad_{s(i)}(Z_i) - \bar \ad(Z_i) )^2     \\
	& = \sum_{s \in [S]} \sumI \bigg( w_{i,v,p}^2(h) + \frac{	C}{R} \sum_{j:\;   i \in R_j}w_{j,v,p}^2(h)\bigg)  (\wh \ad_{s(i)}(Z_i) - \bar \ad(Z_i) )^2  \\
	& \equiv \sum_{s \in [S]} A_{2,s}.
	\end{align*}
For all $s \in [S]$,  it holds with probability approaching one that
\begin{align*}
\mathbb E[& 	A_{2,s}    | \Xn, \{W_i\}_{i \in I_s^c}]  \\	
	&\leq 
	\sumI  \left( w_{i,v,p}^2(h) + \frac{	C}{R} \sum_{j:\; i \in \mc R_j}w_{j,v,p}^2(h)\right)      \sup_{\ad \in \Tn} \sup_{x \in \mc X_h} \E\left[ (\ad(Z_i)-\bar\ad(Z_i))^2|X_i=x\right]\\
	&\leq 	C \Sum  w_{i,v,p}^2(h)    \sup_{\ad \in \Tn} \sup_{x \in \mc X_h} \E\left[ (\ad(Z_i)-\bar\ad(Z_i))^2|X_i=x\right]      =  O_P((nh^{1+2v})^{-1} r_n^2 ).
\end{align*}
	As $S$ is finite and $A_{2,s}$ is a positive random variable, it follows that $A_2=o_P((nh^{1+2v})^{-1})$.
	
To show that $A_1$ is of order $o_P((nh^{1+2v})^{-1})$, we separate the terms involving the nearest neighbors in the fold of unit $i$ and those that involve at least one neighbor from a different fold. Specifically, we have that:
\begin{align*}
	A_1	 & = \frac{r}{R^2}	\Sumn  w_{i,v,p}^2(h) \bigg(  \sum_{j, l\in \mc R_i} (M_i(\bar \ad) -   M_l(\bar \ad))  \left( (\wh \ad_{s(i)}(Z_i) - \bar \ad(Z_i))  -   (\wh \ad_{s(j)}(Z_j) - \bar \ad(Z_j))\right) \bigg)\\
& =	\frac{r}{R^2}	\Sumn  w_{i,v,p}^2(h) \left(  \sum_{\substack{ (j, l) \in \mc R_i^2 \\ (j,l) \notin I^2_{s(i)}}}  (M_i(\bar \ad) -   M_l(\bar \ad))  \left( (\wh \ad_{s(i)}(Z_i) - \bar \ad(Z_i))  -   (\wh \ad_{s(j)}(Z_j) - \bar \ad(Z_j))\right) \right)    \\
	& \quad +  \frac{r}{R^2}  \sum_{s \in [S]}	\sumI  w_{i,v,p}^2(h) \bigg(  \sum_{j, l\in \mc R_i \cap I_s}  (M_i(\bar \ad) -   M_l(\bar \ad))  \left( (\wh \ad_{s(i)}(Z_i) - \bar \ad(Z_i))  -   (\wh \ad_{s(j)}(Z_j) - \bar \ad(Z_j))\right) \bigg)   \\		
& \equiv A_{1,1} + \frac{r}{R^2}\sum_{s \in [S]} A_{1,2,s}. 
	\end{align*}
By Assumption~\ref{ass:technical_se_nn}, it holds that $A_{1,1}=o_P((nh^{1+2v})^{-1})$. For all $s \in [S]$, it holds with probability approaching one that 
\begin{align*}
\frac{1}{C} &\mathbb E[ |	A_{1,2,s}   | | \Xn, \{W_i\}_{i \in I_s^c}]  \\	
& \leq \sumI  w_{i,v,p}^2(h) \sum_{j, l\in (\mc R_i \cap I_s)\cup \{i\}}    \mathbb E[| (M_i(\bar \ad) -   M_l(\bar \ad))  (\wh \ad_{s(j)}(Z_j) - \bar \ad(Z_j))|   | \Xn, \{W_i\}_{i \in I_s^c}]   \\ 
& \leq  \sumI  w_{i,v,p}^2(h) \sum_{j, l\in (\mc R_i \cap I_s)\cup \{i\}}  \sup_{\ad \in \mc T_n}   \mathbb E[|  (M_i(\bar \ad) -   M_l(\bar \ad)) ( \ad(Z_j) - \bar \ad(Z_j)) | | \Xn]   \\
& \leq  \sumI  w_{i,v,p}^2(h) \sum_{j, l\in (\mc R_i \cap I_s)\cup \{i\}}  \left( \mathbb E[  (M_i(\bar \ad) -   M_l(\bar \ad))^2   | \Xn]   \sup_{\ad \in \mc T_n}   \mathbb E[ ( \ad(Z_j) - \bar \ad(Z_j))^2  | \Xn]  \right)^{1/2} \\
&= O_P((nh^{1+2v})^{-1}r_n),
\end{align*}
where the last equality follows from Assumption~\ref{ass:1stage} and the assumption of bounded second moments. 
Hence, $A_{1,2,s} = o_p((nh^{1+2v})^{-1})$, which concludes this proof. \qed

\subsubsection{Proof of Proposition~\ref{prop::CIus}}
The validity of the CI follows directly from the asymptotic normality of the local linear estimator established in Theorem~\ref{th:Normality_General} and the fact that the standard error is consistent. \qed

\subsubsection{Proof of Proposition~\ref{prop::CI_rb}}
Validity of the CI follows directly from asymptotic normality of the local quadratic estimator  established in Theorem~\ref{th:Normality_General} and the fact that the standard error is consistent. \qed

\subsubsection{Proof of Proposition~\ref{prop::CIba}}
Validity of the CI follows directly from asymptotic normality of the local linear estimator established in Theorem~\ref{th:Normality_General}, the fact that the standard error is consistent, and that the asymptotic bias is bounded in absolute value by $\bar b(h) + o_P(h^2)$. \qed

\subsubsection{Proof of Proposition~\ref{prop:bandwidth}}
The proposition follows, using the consistency of the standard error established in Proposition~\ref{prop:standarderror}, if the following claims hold:
		\begin{enumerate}[label=(\roman*)]
		\item $\wh \gamma^\star_{4,4}(\wh\ad) - \wh\gamma^\star_{4,4}(\bar\ad) = o_P(1)$,
		\item $\wh \beta^+_{3,3}(c_n;\wh\ad) + \wh \beta^-_{3,3}(c_n;\wh\ad) = \beta^+_{3}(\bar\ad) + \beta^-_{3}(\bar\ad) +o_P(1)$,
		\item $\wh \beta^+_{2,2}(b_n;\wh\ad) - \wh \beta^-_{2,2}(b_n;\wh\ad) = \beta^+_{2}(\bar\ad) - \beta^-_{2}(\bar\ad) +o_P(1)$.
	\end{enumerate}
\textit{Part (i).} First, note that 
	$$
	\wh \gamma^\star_{4,4}(\wh\ad) -  \wh \gamma^\star_{4,4}(\bar\ad) = e_5' \left(\Sum \wt X_{4,i}^\star {\wt X_{4,i}^{\star^{\top}}} \right)^{-1} \Sum \wt X_{4,i}^\star (\bar\ad(Z_i) - \wh\ad_{s(i)}(Z_i)),
	$$
	where $\wt X^+_{4,i} = \wt X_{4,i} \1{X_i \geq 0} $ and $\wt X^-_{4,i} = \wt X_{4,i} \1{X_i<0} $. Further, for $s \in [S]$, we have that 
	\[
	\Big| \frac{S}{n}\sum_{i \in I_s} X_i^j (\bar\ad(Z_i) - \wh\ad_{s}(Z_i)) \Big| \leq \sqrt{ \frac{S}{n}\sum_{i \in I_s} X_i^{2j} } \sqrt{ \frac{S}{n}\sum_{i \in I_s} (\bar\ad(Z_i) - \wh\ad_{s}(Z_i))^2}.
	\]
	Note that, with probability approaching one,
	\begin{align*}
		\E\left[ \frac{S}{n}\sum_{i \in I_s} (\bar\ad(Z_i) - \wh\ad_{s}(Z_i))^2 \Big| \mathbb X_n, (W_j)_{j\in I_s^c} \right] & \leq  \sup_{\ad \in \mc T_n} \E\left[ \frac{S}{n}\sum_{i \in I_s} (\bar\ad(Z_i) - \ad(Z_i))^2 \Big| \mathbb X_n \right] \\
		& \leq \sup_{\ad \in \mc T_n}  \sup_{x \in \mc X}  \E[(\bar\ad(Z_i) - \ad(Z_i))^2|X_i=x ] = o(1).
	\end{align*}
	It follows that $\Big| \frac{S}{n}\sum_{i \in I_s} X_i^j (\bar\ad(Z_i) - \wh\ad_{s(i)}(Z_i)) \Big|=o_p(1)$. Since $\X$ is bounded, the claim follows.
	
	\noindent\textit{Part~(ii) and~(iii).} Using steps as in the proof of Theorem~\ref{th:Equivalence_General}, for $p\in\{2,3\}$, we obtain that $\wh \beta^\star_{p,p}(h;\wh\ad) - \wh \beta^\star_{p,p}(h; \bar\ad) = o_P(1)$. Moreover, under the assumptions made, $\wh\beta^\star_{p,p}(h;\bar\ad) - \beta^\star_{p}(\bar\ad) = O_P(h +(nh^{1+2p})^{-1/2})$. The claims follow using the conditions on $b_n$ and $c_n$.\qed

\section{Details on the Empirical Analysis}\label{sec:appendix_empirical_results}

In this section, we provide additional details on our empirical analysis described in Section~\ref{sec:practical_performance}.

\subsection{Data Collection}
We conducted an extensive literature search in order to document how covariates are used in empirical RD designs and to collect data sets on which to compare our proposed method to the existing approaches. We focused on the publications in AER, AER Insights, AEJ: Applied Economics, AEJ: Economic Policy, and AEA Papers and Proceedings between 2018 and 2023. Starting from a Google Scholar search for the keywords ``regression discontinuity'',\footnote{A majority of papers found through the Google Scholar search did not
	conduct an original RD analysis, but only cited other RD papers, and were hence excluded.} we first identified 74 articles that appeared to fit into our theoretical framework,\footnote{We
	excluded geographic RD designs where boundary fixed effects were included as part of the identification strategy, and a small number of other nonstandard RD analyses where the outcome variable was measured at a higher level of aggregation than
	the running variable or a donut design was used.} and then retained those 16 papers for which the journal's replication package contained all the data used in the empirical analysis. In 14 of these papers, covariates were used in at least one of the reported RD regressions, while in two papers the available covariates were used only for balance checks but could in principle have been used in the RD regressions, too. For each paper, we identified the main specification (or a version thereof) that includes covariates. These specifications often involve multiple outcomes or running variables, which yielded a total of 56 specifications. The details on all of them are given in Table~\ref{table::overview_of_papers} in the Online Supplement. In our reanalysis of these papers, we focus on these main specifications. In the two cases where only a no covariates RD analysis is reported, we included the covariates that were used for covariate balance checks.

\subsection{Implementation Details}
We apply our flexible adjustment RD estimator proposed in Section~\ref{sec:Covariate_adj}, and we contrast it with the no covariates, conventional linear, and cross-fitted localized linear adjustment RD estimators. In general, the flexible adjustment is implemented as an ensemble of eight learners listed in Section~\ref{sec:proposed_method} and we use 5-fold cross-fitting with $B=25$ data splits. For three specifications with more than 100,000 observations, we speed up the computations by restricting the set of methods used to the local versions of machine learning methods and using 2-fold cross-fitting with $B=5$ data splits. For one specification where the number of observations times the number of covariates exceeds 500,000,000, we use only the local version of random forest as our flexible adjustment and consider 2-fold cross-fitting with $B=1$ data split. We used the rule of thumb of \cite{imbens2017optimized} for the smoothness bound. This choice was dictated by practical considerations, as it would not be possible to separately discuss the choice of smoothness bound for each of the 56 specifications. While one can argue whether
the resulting smoothness bound is always appropriate, the qualitative conclusions about the relative reductions in the confidence interval length are not too sensitive to the choice of the smoothness bound.

\section*{Acknowledgments}
We thank Sebastian Calonico, Michal Koles{\'a}r, Thomas Lemieux,  Jonathan Roth, Vira Semenova, Stefan Wager, Daniel Wilhelm, Andrei Zeleneev, and numerous conference and seminar participants  for  helpful comments and suggestions. We thank Tobias Gro{\ss}b{\"o}lting and Merve {\"O}\u{g}retmek  for excellent research assistance. 
The authors gratefully acknowledge financial support by the European Research Council (ERC) through grant SH1-77202.
The second author also gratefully acknowledges support from the European Research Council (ERC) through grant SH-1852332. 
The authors gratefully acknowledge access to the Marvin cluster of the University of Bonn.

\singlespacing
\bibliography{bibl}    

@article{DoubleML2022,
      title   = {{DoubleML} -- {A}n Object-Oriented Implementation of Double Machine Learning in {P}ython},
      author  = {Philipp Bach and Victor Chernozhukov and Malte S. Kurz and Martin Spindler},
      journal = {Journal of Machine Learning Research},
      year    = {2022},
      volume  = {23},
      number  = {53},
      pages   = {1--6},
      url     = {http://jmlr.org/papers/v23/21-0862.html}
}

@article{velez2024asymptotic,
	title={On the Asymptotic Properties of Debiased Machine Learning Estimators},
	author={Velez, Amilcar},
	journal={arXiv preprint arXiv:2411.01864},
	year={2026}
}

@article{lei2021regression,
	title={Regression adjustment in completely randomized experiments with a diverging number of covariates},
	author={Lei, Lihua and Ding, Peng},
	journal={Biometrika},
	volume={108},
	number={4},
	pages={815--828},
	year={2021},
	publisher={Oxford University Press}
}

@article{lin2013agnostic,
	author = {Winston Lin},
	title = {{Agnostic notes on regression adjustments to experimental data: Reexamining Freedman’s critique}},
	volume = {7},
	journal = {The Annals of Applied Statistics},
	number = {1},
	publisher = {Institute of Mathematical Statistics},
	pages = {295--318},
	year = {2013}
}

@article{chiang2023regression,
	title={Regression adjustment in randomized controlled trials with many covariates},
	author={Chiang, Harold D and Matsushita, Yukitoshi and Otsu, Taisuke},
	journal={arXiv preprint arXiv:2302.00469},
	year={2023}
}

@article{freedman2008regression,
	title={On regression adjustments to experimental data},
	author={Freedman, David A},
	journal={Advances in Applied Mathematics},
	volume={40},
	number={2},
	pages={180--193},
	year={2008},
	publisher={Elsevier}
}

@article{Freedman2008regressionadjustmentsseveraltreatments,
	author = {David A. Freedman},
	title = {{On regression adjustments in experiments with several treatments}},
	volume = {2},
	journal = {The Annals of Applied Statistics},
	number = {1},
	publisher = {Institute of Mathematical Statistics},
	pages = {176 -- 196},
	year = {2008}
	}

@article{chang2024exact,
	title={Exact bias correction for linear adjustment of randomized controlled trials},
	author={Chang, Haoge and Middleton, Joel A and Aronow, PM},
	journal={Econometrica},
	volume={92},
	number={5},
	pages={1503--1519},
	year={2024},
	publisher={Wiley Online Library}
}

@article{greenstone2022can,
	title={Can technology solve the principal-agent problem? Evidence from China’s war on air pollution},
	author={Greenstone, Michael and He, Guojun and Jia, Ruixue and Liu, Tong},
	journal={American Economic Review: Insights},
	volume={4},
	number={1},
	pages={54--70},
	year={2022},
	publisher={American Economic Association 2014 Broadway, Suite 305, Nashville, TN 37203}
}

@article{johnson2020regulation,
	title={Regulation by shaming: Deterrence effects of publicizing violations of workplace safety and health laws},
	author={Johnson, Matthew S},
	journal={American economic review},
	volume={110},
	number={6},
	pages={1866--1904},
	year={2020},
	publisher={American Economic Association 2014 Broadway, Suite 305, Nashville, TN 37203}
}

@article{wasserman2021up,
	title={Up the political ladder: Gender parity in the effects of electoral defeats},
	author={Wasserman, Melanie},
	journal={AEA Papers and Proceedings},
	volume={111},
	pages={169--173},
	year={2021},
	organization={American Economic Association 2014 Broadway, Suite 305, Nashville, TN 37203}
}

@article{akhtari2022political,
	title={Political turnover, bureaucratic turnover, and the quality of public services},
	author={Akhtari, Mitra and Moreira, Diana and Trucco, Laura},
	journal={American Economic Review},
	volume={112},
	number={2},
	pages={442--493},
	year={2022},
	publisher={American Economic Association 2014 Broadway, Suite 305, Nashville, TN 37203}
}

@article{becker2020forced,
	title={Forced migration and human capital: Evidence from post-WWII population transfers},
	author={Becker, Sascha O and Grosfeld, Irena and Grosjean, Pauline and Voigtl{\"a}nder, Nico and Zhuravskaya, Ekaterina},
	journal={American Economic Review},
	volume={110},
	number={5},
	pages={1430--1463},
	year={2020},
	publisher={American Economic Association 2014 Broadway, Suite 305, Nashville, TN 37203}
}

@article{londono2020upstream,
	title={Upstream and downstream impacts of college merit-based financial aid for low-income students: Ser Pilo Paga in Colombia},
	author={Londo{\~n}o-V{\'e}lez, Juliana and Rodr{\'\i}guez, Catherine and S{\'a}nchez, Fabio},
	journal={American Economic Journal: Economic Policy},
	volume={12},
	number={2},
	pages={193--227},
	year={2020},
	publisher={American Economic Association 2014 Broadway, Suite 305, Nashville, TN 37203-2425}
}

@article{del2020rules,
	title={Rules for recovery: Impact of indexed disaster funds on shock coping in Mexico},
	author={Del Valle, Alejandro and de Janvry, Alain and Sadoulet, Elisabeth},
	journal={American Economic Journal: Applied Economics},
	volume={12},
	number={4},
	pages={164--195},
	year={2020},
	publisher={American Economic Association 2014 Broadway, Suite 305, Nashville, TN 37203-2425}
}

@article{granzier2023coordination,
	title={Coordination and bandwagon effects: How past rankings shape the behavior of voters and candidates},
	author={Granzier, Riako and Pons, Vincent and Tricaud, Clemence},
	journal={American Economic Journal: Applied Economics},
	volume={15},
	number={4},
	pages={177--217},
	year={2023},
	publisher={American Economic Association 2014 Broadway, Suite 305, Nashville, TN 37203-2425}
}

@article{curto2018does,
	title={Does electoral competition curb party favoritism?},
	author={Curto-Grau, Marta and Sol{\'e}-Oll{\'e}, Albert and Sorribas-Navarro, Pilar},
	journal={American Economic Journal: Applied Economics},
	volume={10},
	number={4},
	pages={378--407},
	year={2018},
	publisher={American Economic Association 2014 Broadway, Suite 305, Nashville, TN 37203-2425}
}

@article{chin2023politicians,
	title={When do politicians appeal broadly? The economic consequences of electoral rules in Brazil},
	author={Chin, Moya},
	journal={American Economic Journal: Applied Economics},
	volume={15},
	number={3},
	pages={183--209},
	year={2023},
	publisher={American Economic Association 2014 Broadway, Suite 305, Nashville, TN 37203-2425}
}

@article{asher2020rural,
	title={Rural roads and local economic development},
	author={Asher, Sam and Novosad, Paul},
	journal={American economic review},
	volume={110},
	number={3},
	pages={797--823},
	year={2020},
	publisher={American Economic Association 2014 Broadway, Suite 305, Nashville, TN 37203}
}

@article{ambrus2020loss,
	title={Loss in the time of cholera: Long-run impact of a disease epidemic on the urban landscape},
	author={Ambrus, Attila and Field, Erica and Gonzalez, Robert},
	journal={American Economic Review},
	volume={110},
	number={2},
	pages={475--525},
	year={2020},
	publisher={American Economic Association 2014 Broadway, Suite 305, Nashville, TN 37203}
}

@article{altindag2022mental,
	title={Mental health costs of lockdowns: Evidence from age-specific curfews in Turkey},
	author={Altindag, Onur and Erten, Bilge and Keskin, Pinar},
	journal={American Economic Journal: Applied Economics},
	volume={14},
	number={2},
	pages={320--343},
	year={2022},
	publisher={American Economic Association 2014 Broadway, Suite 305, Nashville, TN 37203-2425}
}

@article{avis2022money,
	title={Money and politics: The effects of campaign spending limits on political entry and competition},
	author={Avis, Eric and Ferraz, Claudio and Finan, Frederico and Varj{\~a}o, Carlos},
	journal={American Economic Journal: Applied Economics},
	volume={14},
	number={4},
	pages={167--199},
	year={2022},
	publisher={American Economic Association 2014 Broadway, Suite 305, Nashville, TN 37203-2425}
}

@article{baskaran2018does,
	title={Does the election of a female leader clear the way for more women in politics?},
	author={Baskaran, Thushyanthan and Hessami, Zohal},
	journal={American Economic Journal: Economic Policy},
	volume={10},
	number={3},
	pages={95--121},
	year={2018},
	publisher={American Economic Association}
}

@article{tuttle2019snapping,
Author = {Tuttle, Cody},
Title = {Snapping Back: Food Stamp Bans and Criminal Recidivism},
Journal = {American Economic Journal: Economic Policy},
Volume = {11},
Number = {2},
Year = {2019},
Month = {May},
Pages = {301-27},
DOI = {10.1257/pol.20170490},
URL = {https://www.aeaweb.org/articles?id=10.1257/pol.20170490}}

@article{colangelo2022double,
	title={Double debiased machine learning nonparametric inference with continuous treatments},
	author={Colangelo, Kyle and Lee, Ying-Ying},
	journal={Journal of Business \& Economic Statistics},
	volume={44},
	number={1},
	pages={67--79},
	year={2026},
	publisher={Taylor \& Francis}
}

@article{su2019non,
	title={Non-separable models with high-dimensional data},
	author={Su, Liangjun and Ura, Takuya and Zhang, Yichong},
	journal={Journal of Econometrics},
	volume={212},
	number={2},
	pages={646--677},
	year={2019},
	publisher={Elsevier}
}

@article{arai2024regression,
	title={Regression discontinuity design with potentially many covariates},
	author={Arai, Yoichi and Otsu, Taisuke and Seo, Myung Hwan},
	journal={Econometric Theory},
	pages={1--36},
	year={2025},
	publisher={Cambridge University Press}
}

@article{laan2007super,
  title={Super Learner},
  author={Van der Laan, Mark J and Polley, Eric C and Hubbard, Alan E},
  journal={Statistical applications in genetics and molecular biology},
  volume={6},
  number={1},
  year={2007},
  publisher={De Gruyter}
}

@article{mccrary2008manipulation,
	title={Manipulation of the running variable in the regression discontinuity design: A density test},
	author={McCrary, Justin},
	journal={Journal of Econometrics},
	volume={142},
	number={2},
	pages={698--714},
	year={2008},
	publisher={Elsevier}
}

@article{gerard2020bounds,
	title={Bounds on treatment effects in regression discontinuity designs with a manipulated running variable},
	author={Gerard, Fran{\c{c}}ois and Rokkanen, Miikka and Rothe, Christoph},
	journal={Quantitative Economics},
	volume={11},
	number={3},
	pages={839--870},
	year={2020},
	publisher={Wiley Online Library}
}

@article{chernozhukov2019double,
	title={Debiased machine learning of global and local parameters using regularized Riesz representers},
	author={Chernozhukov, Victor and Newey, Whitney K and Singh, Rahul},
	journal={Econometrics Journal},
	volume={25},
	number={3},
	pages={576--601},
	year={2022},
	publisher={Oxford University Press}
}

@article{kreiss2021regression,
	title={Inference in regression discontinuity designs with high-dimensional covariates},
	author={Kreiss, Alexander and Rothe, Christoph},
	journal={The Econometrics Journal},
	volume={26},
	number={2},
	pages={105--123},
	year={2023},
	publisher={Oxford University Press}
}

@article{wager2016high,
	title={High-dimensional regression adjustments in randomized experiments},
	author={Wager, Stefan and Du, Wenfei and Taylor, Jonathan and Tibshirani, Robert J},
	journal={Proceedings of the National Academy of Sciences},
	volume={113},
	number={45},
	pages={12673--12678},
	year={2016},
	publisher={National Acad Sciences}
}

@article{kennedy2020optimal,
	title={Towards optimal doubly robust estimation of heterogeneous causal effects},
	author={Kennedy, Edward H},
	journal={Electronic Journal of Statistics},
	volume={17},
	number={2},
	pages={3008--3049},
	year={2023},
	publisher={The Institute of Mathematical Statistics and the Bernoulli Society}
}

@article{kennedy2017nonparametric,
	title={Nonparametric methods for doubly robust estimation of continuous treatment effects},
	author={Kennedy, Edward H and Ma, Zongming and McHugh, Matthew D and Small, Dylan S},
	journal={Journal of the Royal Statistical Society. Series B, Statistical Methodology},
	volume={79},
	number={4},
	pages={1229--–1245},
	year={2017},
	publisher={NIH Public Access}
}

@article{armstrong2018optimal,
	title={Optimal inference in a class of regression models},
	author={Armstrong, Timothy B and Koles{\'a}r, Michal},
	journal={Econometrica},
	volume={86},
	number={2},
	pages={655--683},
	year={2018},
	publisher={Wiley Online Library}
}

@article{calonico2014robust,
	title={Robust nonparametric confidence intervals for regression-discontinuity designs},
	author={Calonico, Sebastian and Cattaneo, Matias D and Titiunik, Rocio},
	journal={Econometrica},
	volume={82},
	number={6},
	pages={2295--2326},
	year={2014},
	publisher={Wiley Online Library}
}

@article{imbens2012optimal,
	title={Optimal bandwidth choice for the regression discontinuity estimator},
	author={Imbens, Guido and Kalyanaraman, Karthik},
	journal={Review of Economic Studies},
	volume={79},
	number={3},
	pages={933--959},
	year={2012},
	publisher={Oxford University Press}
}

@book{cattaneo2019practical,
	title={A practical introduction to regression discontinuity designs: Foundations},
	author={Cattaneo, Matias D and Idrobo, Nicol{\'a}s and Titiunik, Roc{\'\i}o},
	year={2019},
	publisher={Cambridge University Press}
}

@article{imbens2008regression,
	title={Regression discontinuity designs: A guide to practice},
	author={Imbens, Guido W and Lemieux, Thomas},
	journal={Journal of Econometrics},
	volume={142},
	number={2},
	pages={615--635},
	year={2008},
	publisher={Elsevier}
}

@article{lee2010regression,
	title={Regression discontinuity designs in economics},
	author={Lee, David S and Lemieux, Thomas},
	journal={Journal of Economic Literature},
	volume={48},
	number={2},
	pages={281--355},
	year={2010}
}

@article{armstrong2020simple,
	title={Simple and honest confidence intervals in nonparametric regression},
	author={Armstrong, Timothy B and Koles{\'a}r, Michal},
	journal={Quantitative Economics},
	volume={11},
	number={1},
	pages={1--39},
	year={2020},
	publisher={Wiley Online Library}
}

@article{imbens2017optimized,
	title={Optimized regression discontinuity designs},
	author={Imbens, Guido and Wager, Stefan},
	journal={Review of Economics and Statistics},
	volume={101},
	number={2},
	pages={264--278},
	year={2019},
	publisher={MIT Press}
}

@article{noack2021bias,
  title={Bias-aware inference in fuzzy regression discontinuity designs},
  author={Noack, Claudia and Rothe, Christoph},
  journal={Econometrica},
  volume = {92},
	number = {3},
	pages = {687--711},
  year={2024}
}

@ARTICLE{andrews1994asymptotics,
  author = {Andrews, D.W.K.},
  title = {{Asymptotics for semiparametric econometric models via stochastic
	equicontinuity}},
  journal = {Econometrica},
  year = {1994},
  volume = {62},
  pages = {43--72},
  number = {1},
  publisher = {The Econometric Society}
}

@ARTICLE{dong2014alternative,
  author = {Dong, Yingying},
  title = {Alternative Assumptions to Identify LATE in Fuzzy Regression Discontinuity
	Designs},
  journal = {Oxford Bulletin of Economics and Statistics},
  volume = {80},
  pages = {1020--1027},
  number = {5},
  year = {2018}
}

@BOOK{fan1996local,
  title = {Local polynomial modelling and its applications},
  publisher = {Chapman \& Hall/CRC},
  year = {1996},
  author = {Fan, J. and Gijbels, I.}
}

@ARTICLE{hahn1998role,
  author = {Hahn, J.},
  title = {On the role of the propensity score in efficient semiparametric estimation
	of average treatment effects},
  journal = {Econometrica},
  year = {1998},
  volume = {66},
  pages = {315--331},
  number = {2},
  publisher = {JSTOR}
}

@ARTICLE{hahn2001identification,
  author = {Hahn, Jinyong and Todd, Petra and Van der Klaauw, Wilbert},
  title = {Identification and Estimation of Treatment Effects with a Regression-Discontinuity
	Design},
  journal = {Econometrica},
  year = {2001},
  volume = {69},
  pages = {201--209},
  number = {1},
  publisher = {Wiley Online Library}
}

@ARTICLE{kolesar2018discrete,
  author = {Kolesár, Michal and Rothe, Christoph},
  title = {Inference in Regression Discontinuity Designs with a Discrete Running
	Variable},
  journal = {American Economic Review},
  pages = {2277–-2304},
    volume = {108},
  year = {2018}
}

@ARTICLE{newey1994variance,
  author = {Newey, W.},
  title = {{The Asymptotic Variance of Semiparametric Estimators}},
  journal = {Econometrica},
  year = {1994},
  volume = {62},
  pages = {1349--1382}
}

@ARTICLE{robins01rotnitsky,
  author = {Robins, James M. and Rotnitzky, Andrea},
  title = {Comment on ``Inference for semiparametric models: some questions
	and an answer'' by P. Bickel and J. Kwon},
  journal = {Statistica Sinica},
  year = {2001},
  volume = {11},
  pages = {920--936},
  number = {4}
}

@BOOK{van2000asymptotic,
  title = {Asymptotic Statistics},
  publisher = {Cambridge University Press},
  year = {1998},
  author = {van der Vaart, Aad }
}

@article{fan2020estimation,
	author = {Qingliang Fan and Yu-Chin Hsu and Robert P. Lieli and Yichong Zhang},
	title = {Estimation of Conditional Average Treatment Effects With High-Dimensional Data},
	journal = {Journal of Business \& Economic Statistics},
	volume = {40},
	number = {1},
	pages = {313–-327},
	year  = {2022},
	publisher = {Taylor \& Francis}	
}

@article{belloni2017program,
	author = {Belloni, Alexandre and Chernozhukov, Victor and Fernández-Val, Iván and Hansen, Christian},
	title = {Program Evaluation and Causal Inference With High-Dimensional Data},
	journal = {Econometrica},
	volume = {85},
	number = {1},
	pages = {233-298},
	year = {2017}
}

@article{chernozhukov2018double,
	author = {Chernozhukov, Victor and Chetverikov, Denis and Demirer, Mert and Duflo, Esther and Hansen, Christian and Newey, Whitney and Robins, James},
	title = "{Double/debiased machine learning for treatment and structural parameters}",
	journal = {Econometrics Journal},
	volume = {21},
	number = {1},
	pages = {C1-C68},
	year = {2018},
	month = {01}
}

@article{calonico2019regression,
	author={Sebastian Calonico and Matias D. Cattaneo and Max H. Farrell and Rocío Titiunik},
	title={{Regression Discontinuity Designs Using Covariates}},
	journal={Review of Economics and Statistics},
	year={2019},
	volume={101},
	number={3},
	pages={442-451},
	month={July}
}

@article{froelich2019including,
	author = {Markus Frölich and Martin Huber},
	title = {Including Covariates in the Regression Discontinuity Design},
	journal = {Journal of Business \& Economic Statistics},
	volume = {37},
	number = {4},
	pages = {736-748},
	year  = {2019},
	publisher = {Taylor \& Francis}	
}

\end{document}


\pagestyle{plain}

\bibliographystyle{ecta}
 
\maketitle 

\begin{abstract}
	This Online Supplement contains additional empirical and simulation results.
\end{abstract}

\doublespacing
\appendix
\supplementarysection

\section{Additional Empirical Results}\label{sec:additional_empirical_results}
In Figure~\ref{fig::appendi_empirical_full_results}, we present the full results of our empirical analysis for bias-aware inference and robust bias correction. The first two graphs of Panel~A are discussed in the main text. The third and fourth graphs illustrate the length of the confidence intervals associated with the cross-fitted and conventional linear covariate adjustments, respectively, relative to the no covariates confidence intervals. We note that the conventional linear adjustment yields on average slightly shorter confidence intervals, but this effect might be due the downward bias of the associated standard error documented in Simulation~II in Section~\ref{sec:Simulations}. Panel~B presents the results based on the bias-aware approach with the second-stage smoothness bound calibrated based on the adjusted outcomes using the rule of thumb of \cite{imbens2017optimized}. This choice was dictated by practical considerations, as it would not be possible to separately discuss the choice of smoothness bound for each of the 56 specifications. In comparison to Panel~A, all four histograms are more spread out, which reflects the differences between the smoothness bounds calibrated based on the original and adjusted outcomes. In some cases, the confidence intervals based on the flexible adjustment are wider than the no covariates and linear adjustment confidence intervals, which is due to an increase in the smoothness bound, but the average reductions in the confidence interval length are larger then in Panel~A. We note, however, that these comparisons are sensitive to the method of choosing the smoothness bound.
Panel~C presents the results based on the robust bias correction. They are qualitatively similar to those in Panel~B. In particular, the flexible and linear adjustments lead to wider confidence intervals in some cases. By inspecting the results, we saw that these increases occurred in cases where the adjusted bandwidth was smaller than the no covariates bandwidths.

In 13 out of 16 papers in our literature analysis, the standard errors were clustered. To account for that, in Figure~\ref{fig::appendi_empirical_full_results_clustering}, we present the results of our empirical analysis with clustered standard errors. In our second-stage RD regression, we cluster the standard error based on the same variable as in the original application. Additionally, we adjust the data splitting procedure such that all observations within a cluster belong to the same fold. Clustering substantially increases the length of all confidence intervals, but the relative patterns displayed in Figure~\ref{fig::appendi_empirical_full_results_clustering} are broadly similar to those in Figure~\ref{fig::appendi_empirical_full_results}.

\afterpage{\begin{sidewaysfigure}
	\centering 
	\includegraphics[scale=0.8]{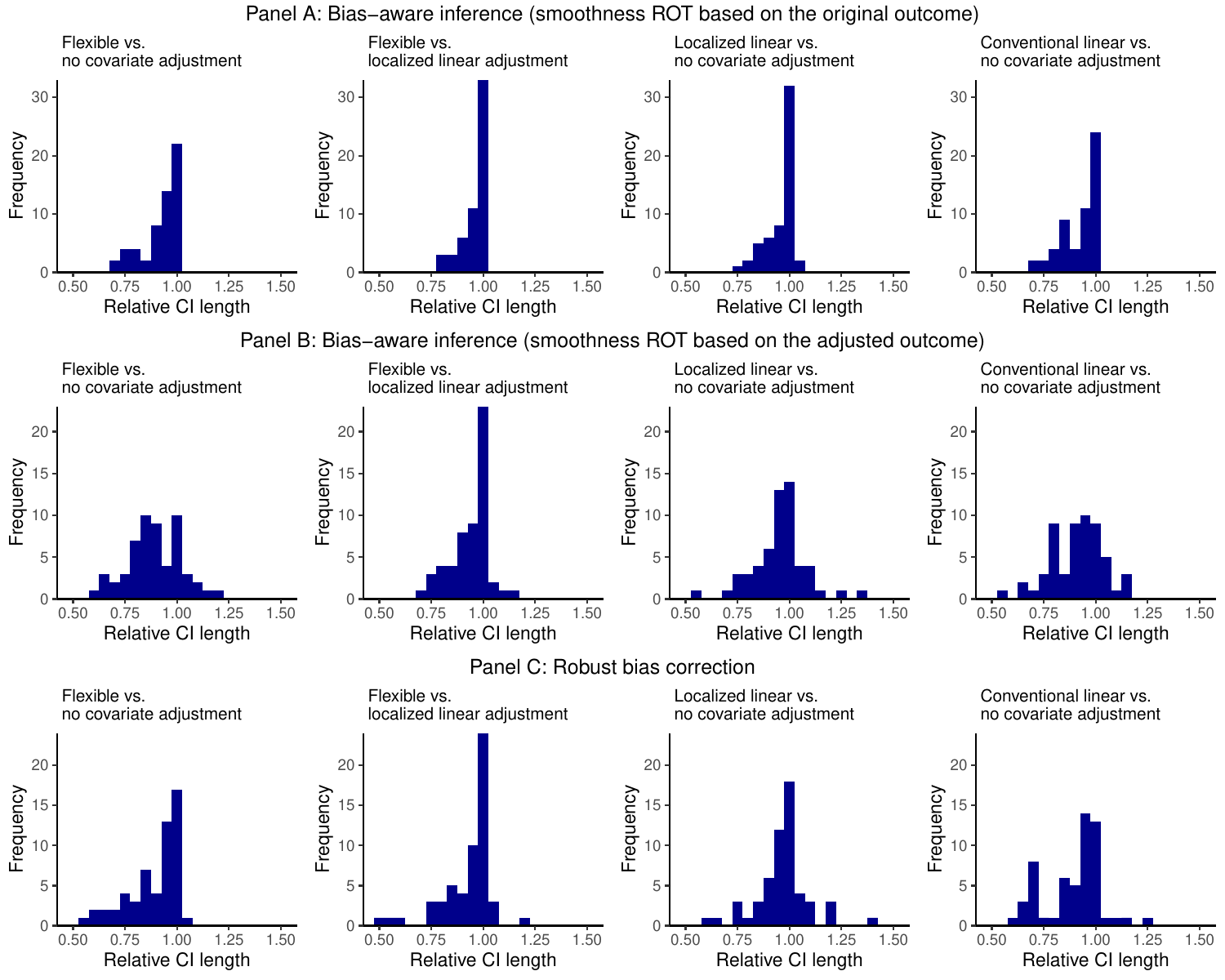}
	\caption{Full empirical results without clustering.}\label{fig::appendi_empirical_full_results}
	\footnotesize{\textit{Notes:} Results of our empirical analysis for bias-aware inference (in two variants) and robust bias correction.
	See Section~\ref{sec:Covariate_adj} and Appendix~\ref{sec:Inference} for details on the estimators and the confidence intervals; and  Section~\ref{sec:practical_performance} and~Appendix~\ref{sec:appendix_empirical_results} for details on the data sets.}
\end{sidewaysfigure}}

\afterpage{\begin{sidewaysfigure}
	\centering 
	\includegraphics[scale=0.8]{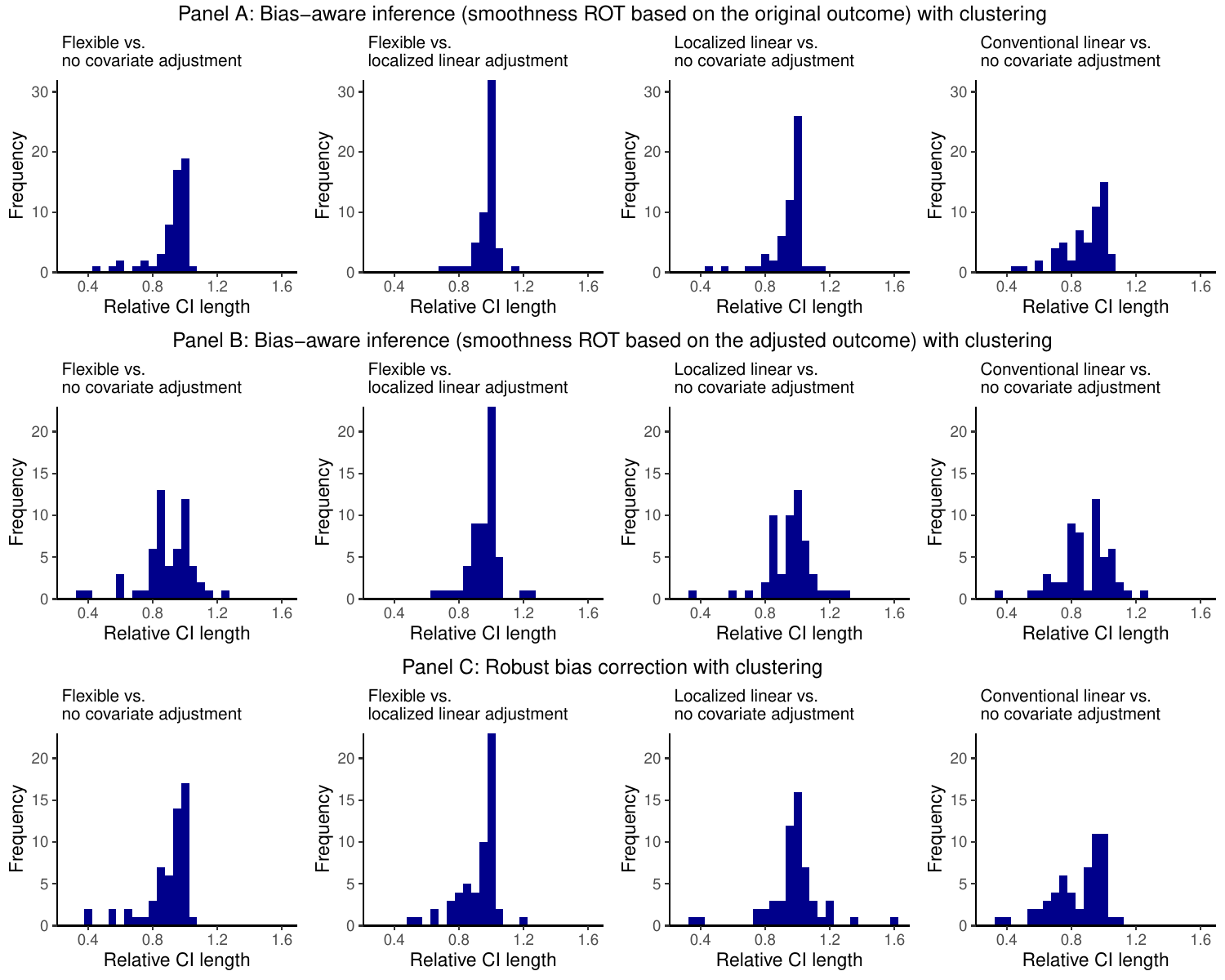}
	\caption{Full empirical results with clustering.}\label{fig::appendi_empirical_full_results_clustering}
	\footnotesize{\textit{Notes:} Results of our empirical analysis for bias-aware inference (in two variants) and robust bias correction with clusered standard errors.
	See Section~\ref{sec:Covariate_adj} and Appendix~\ref{sec:Inference} for details on the estimators and the confidence intervals; and  Section~\ref{sec:practical_performance} and~Appendix~\ref{sec:appendix_empirical_results} for details on the data sets.}
\end{sidewaysfigure}}

\afterpage{
	\clearpage
	\begin{landscape}
	\footnotesize
	\begin{longtable}[l]{p{0.2cm}p{5.5cm}p{2.9cm}p{4.5cm}p{1.5cm}p{1.2cm}p{1.3cm}p{1.6cm}p{1.1cm}p{1cm}}
		\caption{Overview of the papers of the literature analysis.}\label{table::overview_of_papers}\\
		\toprule
		& Outcome & Running variable & Covariates &\raggedright  \#Covs (not 0/1)  & \raggedright \#Obs & \raggedright Eff \#Obs &\raggedright Eff \#Obs /\#Covs &  \raggedright \#Clusters  & Covs in RD \\ 
		\midrule	
		\endhead 
		\multicolumn{10}{l}{  \textbf{\citet[AER][]{akhtari2022political}}}\\
		\multicolumn{10}{l}{Table 3: \textquotedblleft Political Turnover and Fourth Grade and Eighth-grade Test Scores\textquotedblright}\\[2ex]	   
		1	 &  Fourth-grade test scores & \multirow{2}{=}{Incumbent's vote margin} &  \multirow{2}{=}{Baseline school-level average test scores; school- and individual-level controls; 
			election-cycle indicator} & 14 (1) & 1,088,553 &  325,554  &  23,254 & 3737 & Yes \\   
		2	 & Eighth grade test scores &  & & 14 (1) & 446,451 & 234,629   &  17,545   &   2368 & Yes \\[6ex]  
		\multicolumn{10}{l}{\textbf{\citet[AEJAE][]{altindag2022mental}}}\\
		\multicolumn{10}{l}{  Table 4: \textquotedblleft Effects of Curfew on Mental Health Outcomes\textquotedblright}\\[2ex]    
		3  & Mental distress  & \multirow{4}{=}{number of months
			older than  index month} &  \multirow{5}{=}{month, province, and surveyor fixed effects, indicators for education levels,
			ethnicity, and gender} & 175 (0) & 1868 & 475 & 2.7 &  144 &  Yes \\ 
		4	  & Somatic symptoms of distress  &  &   & 175 (0) & 1868 & 503 & 2.8 & 144 &  Yes \\ 
		5	   & Nonsomatic symptoms of distress &  &  &  175  (0)  & 1868 & 478 &  2.7 &  144 &  Yes \\ 
		6	  & Sum of “Yes” answers in SRQ-20 &  &  &  175 (0) & 1868 &  475 & 2.7 &  144 &  Yes \\[2ex]
		\multicolumn{10}{l}{\textit{Notes:} The authors present results for different bandwidths. Here, the reported effective sample sizes correspond to the bandwidth calculated via the algorithm } \\
		\multicolumn{10}{l}{\phantom{\textit{Notes:}} of \cite{calonico2014robust}. In the main specification, month fixed effects are included and the standard error is clustered on the running variable.  In our  }\\ 
		\multicolumn{10}{l}{\phantom{\textit{Notes:}} reanalysis, we do not include covariates that are a deterministic function of the running variable}\\[4ex] 
		\multicolumn{10}{l}{\textbf{\citet[AER][]{ambrus2020loss}}}
		\\
		\multicolumn{10}{l}{ Table 3: \textquotedblleft Boundary Effects of Rental Prices\textquotedblright}\\[2ex]
		7 & Log rental prices, 1853 & \multirow{4}{=}{Distance to boundary}  & \multirow{4}{=}{Determinants of rental values, distance to various amenities,  distance to presumed plague pit, and sewer access} & 14 (12) & 1738 & 469 & 34 & 179 & Yes \\ 
		8  & Log rental prices, 1864 &  & & 14 (12) & 1738 & 510 & 36 & 179 &   Yes \\ 
		9 & Log rental prices, 1894 &   & & 5 (5)   & 1879 & 363 & 73 & 179 &   Yes\\ 
		10 & Log rental prices, 1936 &  &    & 6 (6)&  793 & 221 & 37 & 90 & Yes \\[4ex] 
		\multicolumn{10}{l}{\textbf{\citet[AER][]{asher2020rural}}}\\*
		\multicolumn{10}{l}{Table 3: \textquotedblleft Impact of New Road on Indices of Major Outcomes\textquotedblright}\\*[2ex]
		11	 & Transportation & \multirow{5}{=}{Village population} &  \multirow{5}{=}{baseline village-level controls for amenities and economic indicators, as well as district-cutoff fixed effects} & 225 (8)  & 11432 & 11432 & 51 & - &  Yes \\* 
		12	 & Occupation &  & &  225 (8) & 11432 & 11432 & 51 & - &   Yes \\* 
		13	 & Firms &  &   &  225 (8) & 10678 & 10678 & 48 & -  &  Yes \\* 
		14	 & Production &  &   & 225  (8) & 11432 & 11432 & 51 & -  &    Yes \\* 
		15	 & Consumption &  &   &  225 (8) & 11432 & 11432 & 51 & - & Yes    \\[4ex] 
		\multicolumn{10}{l}{\textbf{\citet[AEJAE][]{avis2022money}}}\\*
		\multicolumn{10}{l}{  Table 4: \textquotedblleft Effects of Campaign Spending Limits on Candidate Entry\textquotedblright}\\*[2ex]
		16	  & \# of candidates & \multirow{12}{=}{maximum amount a
			candidate spent in municipality election} &  \multirow{12}{=}{Municipal controls: GDP per capita, illiteracy, share urban, Gini coefficient, population}  & 5 (5) & 5562 & 3080    & 616 & - &  Yes  \\* 
		17	 & Eff. \# of candidates &  &  &  5 (5)  & 5558 & 3052    &  610  &  - & Yes     \\* 		
		18	 &  Small party &  &   &  5 (5)  & 5562 &  	3116    & 623  & -  & Yes    \\* 
		19	 & Small party w/o incumbent &  &  &  5 (5)  & 5562 & 2804    &  561 & -  &  Yes   \\* 
		20  &  Party's ideology index &  &  &  5 (5) & 5562 &  2783    &  557  & - & Yes    \\*
		21 & Candidate's prop. to win &  &  &  5 (5) & 5459 & 3074    & 615 &  - & Yes    \\* 
		22 & Candidate's wealth &  &  &  5 (5)  & 5562 &  3218    & 644 &  - & Yes    \\* 
		23 & Candidate's political experience &  &  &  5 (5) & 5562 & 2849    & 570  &  - &  Yes   \\ 
		24 & Candidate's gender &  &  &  5 (5)& 5562 &   3080    & 616 & - & Yes    \\* 
		25 & Candidate's age  &  &  &   5 (5)  & 5562 & 3259    & 652 & - & Yes    \\* 
		26 & Candidate's college degree &  &  &  5 (5) & 5562 & 2881   & 576 &  - & Yes    \\* 
		27 & Candidate: white &  &  &  5 (5) & 5562 & 2668    & 534 & - &   Yes  \\[4ex] 
		\multicolumn{10}{l}{ \textbf{\citet[AEJEP][]{baskaran2018does}}}\\*
		\multicolumn{10}{l}{  Table 2: \textquotedblleft Baseline Results: Rank Improvement of Female Candidates\textquotedblright}\\*[2ex]
		
		28  &  Rank improvement & vote margin  & municipality characteristics & - & 6472 & 2878 & - & 134 & No  \\*
		\multicolumn{10}{l}{\textit{Notes:} We use 24 (24 non-binary) covariates from the robustness check in Table A.4.}\\[4ex] 
		\multicolumn{10}{l}{\textbf{\citet[AER][]{becker2020forced}}}\\*
		\multicolumn{10}{l}{Table~A.10: \textquotedblleft  Border Sample from the Diagnoza Survey\textquotedblright}\\*[2ex]
		29  & Years of education   &  distance to boarder & Respondents’ gender, age, squared age, dummies for six age groups, indicators for Western Territories, rural places and urban counties  & 20 (17)  & 33160  & 8760 & 438 &  11734 & Yes  \\*
		\multicolumn{10}{l}{\textit{Notes:} All RD results  are in the appendix.}\\[4ex] 
		\multicolumn{10}{l}{\textbf{\citet[AEJAE][]{chin2023politicians} }}\\*
		\multicolumn{10}{l}{Table 2: \textquotedblleft Effect on the Geographic Concentration of Voters \textquotedblright and Table 4 Panel C in the Appendix}\\*[2ex]
		30 & Coefficient of
		variation  of voters & \multirow{7}{=}{Number of registered voters} &  \multirow{7}{=}{Election-year fixed effects,	municipality characteristics}  & 17  (11) &    22915 & 230 & 14 & 5568  & Yes   \\* 
		31 & Fractionalization   of voters& & & 17 (11)   & 33187  & 230 & 14 & 5568  &  Yes   \\* 
		32 &  Entropy  of voters  & & &  17 (11) &   33187 & 230 & 14 &  5568 & Yes   \\*
		33 & SD in vote shares for first placed  && & 17 (11)  &  33187 & 230 & 14 & 5568 & Yes \\* 
		34 & SD in vote shares for second placed & &&  17 (11) &  33187 & 230 & 14 & 5568 &  Yes  \\* 
		35 &  SD in vote shares for third placed & & & 17  (11) &  33187 &  217 & 13 & 5568 &    Yes \\* 
		36 & SD in vote shares for  fourth placed &  &  & 17 (11) &   33187 &   185 & 11 & 5568 &  Yes  \\* 	
		\multicolumn{10}{l}{\textit{Notes:} Additionally to the covariates used in the main text, we use all covariates that were used  in Table 4 Panel C in the Appendix. As the number of observations  }\\*
		\multicolumn{10}{l}{\phantom{\textit{Notes:}} of the original data set is very large and its distribution   is very skewed around the cutoff, we restricted the sample to lie  within three times of the bandwidth}\\*
		\multicolumn{10}{l}{\phantom{\textit{Notes:}}  used in the main analysis around the cutoff. In the main specification, the author include the density of the population as a control, but we don't do this.}\\[4ex] 
		\multicolumn{10}{l}{ \textbf{\citet[AEJAE][]{curto2018does}  }}\\
		\multicolumn{10}{l}{Table 1 A: \textquotedblleft Average Effect of Partisan Alignment on Capital Transfers\textquotedblright}\\ 
		37 & Alignment & incumbent's vote margin  & Financial and demographic municipality characteristics,
		& 14 (0) & 6050 & 2553  & 102 & 2592 & Yes  \\ 
		\multicolumn{10}{l}{\textit{Notes:} In their main specification, they include 14 fixed effects. We do not use them in our no-covariates RD estimator. For our RD estimators that use covariates, }\\
		\multicolumn{10}{l}{\phantom{\textit{Notes:}} we also include all covariates that are used for from the falsification check of Figure~A.10. This gives us a total of 25 (10 non-binary) covariates.}\\[4ex] 
		\multicolumn{10}{l}{ \textbf{\citet[AEJAE][]{granzier2023coordination}}}\\*
		\multicolumn{10}{l}{Table 2: \textquotedblleft Impact on Running in the Second Round and Winning\textquotedblright and Table~C4}\\*[2ex]
		38	&  Running   &     Vote Margin 1 vs 2 & \multirow{6}{=}{gender, characteristics of previous election and party,  incumbent, strength} &  23 (8) & 45064 & 24544 & 1067 & 8970 & OA  \\* 
		39	&   Winning & Vote Margin 1 vs 2 &   	& 23 (8)  & 45064 & 16054 & 698 & 8970 & OA  \\*  
		40	& Running  &  Vote Margin 2 vs 3   &     	& 23 (8)  & 17730 & 10694 & 465 & 4810 & OA  \\* 
		41	& Winning &  Vote Margin 2 vs 3  &   	& 23 (8)  & 17730 & 8796 &  382 & 4810 & OA  \\*  
		42	& Running &  Vote Margin 3 vs 4    &     	& 23 (8) & 3956 & 2338 & 102 & 1243 & OA  \\* 
		43	& Winning &  Vote Margin 3 vs 4  &    	&  23 (8)  & 3956 & 2232 & 97 & 1243  & OA \\[4ex] 
		\multicolumn{10}{l}{ \textbf{\citet[AER Insigts][]{greenstone2022can}}}\\*
		\multicolumn{10}{l}{Table 1: \textquotedblleft Automating Air Quality Monitoring System and Reported PM$_{10}$\textquotedblright, Column 2}\\*[2ex]
		44	& PM$_{10}$ concentration & Days to automation & weather controls, and   station and month fixed effects	&   670 (4)  & 1,049,325  &  49,843 & 74 &  123 & Yes  \\*
		\multicolumn{10}{l}{\textit{Notes:} We do not include covariates that are determined based on the running variable and therefore exclude month fixed effects from our analysis. }\\[4ex] 
		\multicolumn{10}{l}{ \textbf{\citet[AER][]{johnson2020regulation}}}\\*
		\multicolumn{10}{l}{Table 2: \textquotedblleft Instrumental Variables (IV) Estimate of the General Deterrence Effect of a Press Release on Compliance of Other Facilities \textquotedblright and Table A.1}\\*[2ex]
		45	&  Number of Violations & Focal penalty  &    \multirow{3}{=}{construction, programmed} &  2 (0) & 60,416  &  3302 & 1651 & 2746  &   Yes \\* 
		46	&  Number of Violations & Focal penalty &   	&  1 (0)  & 39,058    &  10,873 & 10,873 &  2455 & Yes  \\[6ex]
		\multicolumn{10}{l}{\textit{Notes:} Second specification excludes inspections initiated by a serious accident worker complaint, or referral.}\\
		\multicolumn{10}{l}{\phantom{\textit{Notes:}} We further consider 30 covariates  that are from Table~A.1 (press release,  cfr, union, \# inspection prior tc, total violations prior tc).}\\	
		\multicolumn{10}{l}{  \textbf{\citet[AEJEP][]{tuttle2019snapping} }}  \\*
		\multicolumn{10}{l}{Table 3: \textquotedblleft Main Results: Effect of the SNAP Ban on Recidivism\textquotedblright}  \\[2ex] 	
		47 &   Recidivism & \multirow{3}{=}{Date}	&  \multirow{3}{=}{offender characteristics 
		}  & 14 (4) & 18850 & 790  & 56 &  5385 &  OA \\
		48	&   Financially motivated recidivism & 	&   &  14 (4) & 18850 &    936 & 67 & 5385  & OA  \\
		49	&   Non-financially motivated recidivism & 	&   & 14 (4) &  18850 &  980 & 70 & 5385 & OA   \\
		\multicolumn{10}{l}{\textit{Notes:} In the main specification, dummies for weekdays are included and the standard error is clustered on the running variable.  We do not cluster}\\ 
		\multicolumn{10}{l}{\phantom{\textit{Notes:}} the standard errors  on the running variable and we do not include covariates that are a deterministic functions of the running variable }\\[4ex] 
		\multicolumn{10}{l}{  \textbf{\citet[AEJAE][]{del2020rules}}}\\
		\multicolumn{10}{l}{ Table 2: \textquotedblleft Impact of Fonden on Night Lights\textquotedblright}\\[2ex] 
		50 & difference in night lights & heavy rainfall index & characteristics of dwellings quality,
		health care system, 
		education system,  
		municipal indicators, 
		night lights, location indicators, historic mean annual rainfall & - & 2708 &  1563 & - & 1198 & No \\
		\multicolumn{10}{l}{\textit{Notes:} We use 24 (24 non-binary)  covariates from the robustness check of Figure~4.}\\[4ex] 	
		\multicolumn{10}{l}{\textbf{\citet[AEJEP][]{londono2020upstream} }}\\*
		\multicolumn{10}{l}{ Table 2: \textquotedblleft Immediate Enrollment in any Postsecondary Education, by Type of Institution\textquotedblright and Table A.4}\\*[2ex]
		51 & Immediate Enrollment in Any Postsecondary Education & \multirow{1}{=}{SABER 11 test score}  & \multirow{2}{=}{indicators for gender, age,       ethic,  employment status,   
			family size,  parent's education, household residential stratum,  high school schedule, 
			and  private high school}   
		&   21 (2)  & 273361 & 37882 & 1804 & - &  \multirow{1}{=}{OA}   \\*
		52 & Immediate Enrollment in Any Postsecondary Education & SISBEN wealth index  & &    21 (2) & 21071 & 8201 & 391 & - &  OA  \\[10ex]
		\multicolumn{10}{l}{\textbf{\citet[AEJ-P\&P][]{wasserman2021up}}}\\*
		\multicolumn{10}{l}{Table 2: \textquotedblleft The Effect of Losing, by Gender\textquotedblright}\\*[2ex]
		53 & Prob. Run again - Female & Margin of victory & \multirow{4}{=}{fixed effects for state,
			election-year, political party, and legislative chamber (upper/lower)}   
		& 92 (0)   & 13092 & 3652 &  121  & 50 & Yes   \\*
		54 & Prob. Run again and Win - Female & Margin of victory &    
		& 92 (0)    &  13092 & 3512 & 121 & 50 & Yes   \\*
		55 & Prob. Run again - Male & Margin of victory & &  92 (0)   &  50058 & 	12679 & 459 & 50 & Yes   \\*
		56 & Prob. Run again and win - Male  & Margin of victory & & 92 (0)   &  50058 & 12140 & 457 & 50 & Yes   \\
		\multicolumn{10}{l}{\textit{Notes:} The parameter of interest is the difference of the RD estimands of female and male candidates. Here, we consider these as two separate RD regressions. }\\[10ex]
		\bottomrule
	\end{longtable}
	\textit{Notes: } The table shows description of the respective variables (``Outcome'',  ``Running variable'', ``Covariates''); the number of covariates with the number of nonbinary covariates in parentheses  (``\#Covs (not 0/1)'');  the total sample size ( ``\#Obs'');  the number of observations within the bandwidth of the respective specification (``Eff \#Obs'');  the effective sample size relative to the number of covariates (``Eff \#Obs / \#Covs''); the number of clusters ( ``\#Clusters'' ); and   whether the covariates were used in the RD regression;  ``Yes'' if  they were used in the main text, ``OA'' if they were only used in the online appendix and ``No'' if  they were not used (``Covs in RD'' ).
\end{landscape}}

\section{Additional Simulation Results}\label{sec:AppendixSimulation}
In this section, we provide more details and additional results for the simulation studies in Section~\ref{sec:Simulations}.

\subsection{Scope for Efficiency Gains}
To gauge the scope for efficiency gains due to covariate adjustments in this simulation setting, in Table~\ref{table::sim_full_sample_results}, we present RD estimates at the placebo cutoff using all the observations in the restricted data set of \cite{londono2020upstream} described in Section~\ref{sec:NumericalImplementation}. As in in the main text, we consider the original outcome and age as the dependent variables. We now employ the robust bias correction approach in addition to the bias-aware inference. In Panel~A, the results are very similar in all rows, which indicates that the covariates have virtually no explanatory power for the outcome and so the covariate adjustments do not lead to meaningful changes in the length of confidence intervals. When considering the age as the dependent variable in Panel~B, the machine learning adjustments improve upon the no covariates and linear adjustment RD estimators, with our proposed flexible adjustment leading to the shortest confidence intervals.

\subsection{Additional Results for Simulation I}
Table~\ref{table::appendix_full_simulationI_RDHonest} extends the results in Table~\ref{table::sim_baseline} from the main text and displays the results for all individual methods considered in our flexible adjustment. For all methods that employ cross-fitting, we consider their oracle versions obtained on the restricted data set. The observations about the performance of the flexible adjustment discussed in the main text apply here too. The confidence intervals are slightly conservative, the average standard error is very close to the standard deviation in all cases, and the changes in the bias across different adjustments are minimal relative to the standard deviation. The feasible and infeasible, oracle versions of the estimators perform very similarly. The flexible adjustment consistently leads to the shortest confidence intervals among all the adjustments employing cross-fitting. The results in Table~\ref{table::appendix_full_simulationI_rdrobust} are based on the robust bias correction but are otherwise analogous to the results in Table~\ref{table::appendix_full_simulationI_RDHonest}.

\begin{table}[b]
	\centering
	\caption{Estimation results for the full restricted sample in the simulation setting of Section~\ref{sec:Simulations}.}
	
	\resizebox*{\textwidth}{!}{	\begin{threeparttable}
			\begin{tabular}{p{5.2cm}@{\hspace{0.1cm}}p{1.1cm}p{0.7cm}p{0.8cm}p{0.9cm}p{1.2cm}@{\hspace{0.5cm}}p{1.1cm}p{0.7cm}p{0.8cm}p{0.9cm}p{1.2cm}}
				\toprule
				& \multicolumn{5}{c}{Bias-Aware Inference} &  \multicolumn{5}{c}{Robust Bias Correction}\\
				\cmidrule(r){2-6}\cmidrule(l){7-11}
				Adjustment Method	& Est x100 & SE x100 & Band-width & CI Length x100 & CI Length \% Red. & Est x100 & SE x100 &  Band-width  &  CI Length x100 & CI Length \% Red. \\ 
				\midrule
				
				\multicolumn{11}{l}{\textbf{Panel A - Original outcome as the dependent variable}} \\[0.5ex]
				No Covariates &  0.62 & 0.45 & 14.49 & 1.98 & 0.00 & 0.41 & 0.37 & 21.49 & 1.75 & 0.00 \\ 
				Conventional Linear &  0.73 & 0.44 & 14.55 & 1.93 & 2.40 & 0.56 & 0.38 & 19.30 & 1.80 & -2.58 \\ 
				Localized Linear &  0.75 & 0.44 & 14.57 & 1.93 & 2.43 & 0.57 & 0.39 & 19.16 & 1.80 & -2.92 \\ 
				Global Linear & 0.75 & 0.44 & 14.59 & 1.93 & 2.50 & 0.59 & 0.39 & 19.14 & 1.80 & -2.89 \\ 
				Localized Random Forest &  0.68 & 0.44 & 14.43 & 1.93 & 2.17 & 0.50 & 0.38 & 19.31 & 1.79 & -2.38 \\ 
				Global Random Forest & 0.68 & 0.44 & 14.47 & 1.93 & 2.38 & 0.50 & 0.38 & 19.29 & 1.79 & -2.43 \\ 
				Localized Boosted Trees & 0.67 & 0.44 & 14.59 & 1.92 & 2.75 & 0.48 & 0.38 & 19.75 & 1.77 & -1.31 \\ 
				Global Boosted Tree & 0.67 & 0.44 & 14.50 & 1.93 & 2.38 & 0.49 & 0.38 & 19.83 & 1.77 & -1.18 \\ 
				Localized Post-Lasso &  0.71 & 0.44 & 14.72 & 1.91 & 3.09 & 0.55 & 0.38 & 19.56 & 1.78 & -1.76 \\ 
				Global Post-Lasso &  0.72 & 0.44 & 14.51 & 1.93 & 2.35 & 0.54 & 0.38 & 19.98 & 1.77 & -0.87 \\ 
				Flexible & 0.69 & 0.44 & 14.56 & 1.92 & 2.69 & 0.50 & 0.38 & 19.55 & 1.78 & -1.71  \\[2ex] 
				
				\multicolumn{11}{l}{\textbf{Panel B - Age as the dependent variable}} \\[0.5ex]
				No Covariates & -4.98 & 4.59 & 20.14 & 20.21 & 0.00 & -7.42 & 5.24 & 15.52 & 23.46 & 0.00 \\ 
				Conventional Linear & -4.37 & 4.23 & 17.65 & 18.58 & 8.05 & -5.77 & 4.49 & 15.69 & 20.23 & 13.77 \\ 
				Localized Linear & -4.39 & 4.24 & 17.65 & 18.60 & 7.95 & -5.81 & 4.53 & 15.43 & 20.36 & 13.23 \\ 
				Global Linear & -4.05 & 4.24 & 17.62 & 18.64 & 7.77 & -5.31 & 4.50 & 15.72 & 20.28 & 13.57 \\ 
				Localized Random Forest &  -3.90 & 4.14 & 17.52 & 18.14 & 10.25 & -5.10 & 4.29 & 16.14 & 19.25 & 17.96 \\ 
				Global Random Forest & -3.67 & 4.11 & 17.48 & 18.03 & 10.77 & -4.51 & 4.25 & 16.36 & 19.06 & 18.73 \\ 
				Localized Boosted Trees & -4.29 & 4.18 & 17.56 & 18.33 & 9.29 & -5.57 & 4.39 & 15.96 & 19.74 & 15.87 \\ 
				Global Boosted Tree & -4.21 & 4.18 & 17.51 & 18.36 & 9.18 & -5.05 & 4.36 & 16.14 & 19.66 & 16.21 \\ 
				Localized Post-Lasso & -4.63 & 4.24 & 17.58 & 18.65 & 7.75 & -5.91 & 4.52 & 15.49 & 20.32 & 13.38 \\ 
				Global Post-Lasso & -4.09 & 4.24 & 17.62 & 18.64 & 7.77 & -5.36 & 4.50 & 15.71 & 20.28 & 13.54  \\ 
				Flexible & -3.68 & 4.11 & 17.47 & 18.03 & 10.77 & -4.49 & 4.25 & 16.33 & 19.09 & 18.63 \\ 
				\bottomrule
			\end{tabular}
			\begin{tablenotes}
				\item \textit{Notes:} 
				Results are based on the restricted dataset of \cite{londono2020upstream} described in Section~\ref{sec:Simulations}. Sample size is $n=259,419$.  The columns show the estimate (Est), the standard error (SE), the bandwidth (Bandwidth),  
				the length of confidence intervals with 95\% nominal coverage (CI Length), and the percentage reduction in CI length relative to the no covariates CI length (CI Length \% Red.). Estimators are described in Section~\ref{sec:Covariate_adj}. 
			\end{tablenotes}
	\end{threeparttable}}\label{table::sim_full_sample_results}
\end{table}

\afterpage{\begin{landscape}
		\begin{table}
			\centering
			\caption{Full results for Simulation I with bias-aware inference.}\label{table::appendix_full_simulationI_RDHonest}
			\resizebox{!}{0.44\textwidth}{
				\begin{threeparttable}
					\addtolength{\tabcolsep}{-0.1em}\begin{tabular}{p{1.5cm}p{1.8cm}@{\hspace{0.5cm}}p{0.9cm}p{0.9cm}p{0.9cm}p{0.9cm}p{0.9cm}p{0.9cm}p{1.3cm}p{1.3cm}@{\hspace{0.5cm}}p{0.9cm}p{0.9cm}p{0.9cm}p{0.9cm}p{0.9cm}p{0.9cm}p{1.3cm}p{1.3cm}}
						\toprule
						&& \multicolumn{8}{c}{Original Outcome } &  \multicolumn{8}{c}{Age }\\
						\cmidrule(r){3-10}\cmidrule(l){11-18}
						\multicolumn{2}{l}{Adjustment method} & Mean SE x100 & SD x100 & Bias x100  &  RMSE x100 & Mean Band-width & CI Cov in \% & Mean CI Length x100 &  Mean CI Length \% Red. 
						& Mean SE x100 & SD x100 &  Bias x100 &  RMSE x100 & Mean Band-width & CI Cov in \% & Mean CI Length x100 & Mean CI Length \% Red.  \\ 
						\midrule
						\multicolumn{2}{l}{\textbf{No Covariates}}  & 2.15 & 2.17 & 0.41 & 2.21 & 37.73 & 96.95 & 9.41 & 0.00 & 24.95 & 25.69 & 2.87 & 25.84 & 42.25 & 96.84 & 110.79 & 0.00 \\ [1ex] 
						\multicolumn{2}{l}{\textbf{Conventional Linear}}   &  2.09 & 2.13 & 0.47 & 2.19 & 38.71 & 96.51 & 9.14 & 2.86 & 21.55 & 22.50 & 2.64 & 22.66 & 41.96 & 96.24 & 95.90 & 13.44 \\[1ex] 
						\multicolumn{18}{l}{\textbf{Linear Regression}}\\
						Localized & Feasible &   2.10 & 2.13 & 0.45 & 2.18 & 39.25 & 96.63 & 9.18 & 2.39 & 21.84 & 22.56 & 2.64 & 22.72 & 42.09 & 96.34 & 97.05 & 12.40  \\ 
						& Oracle & 2.10 & 2.13 & 0.48 & 2.18 & 39.08 & 96.67 & 9.18 & 2.44 & 21.71 & 22.43 & 2.69 & 22.59 & 42.07 & 96.52 & 96.53 & 12.87 \\  
						Global & Feasible  &     2.10 & 2.13 & 0.45 & 2.18 & 39.26 & 96.81 & 9.17 & 2.52 & 21.79 & 22.50 & 2.79 & 22.67 & 42.35 & 96.36 & 96.82 & 12.61            \\
						& Oracle  & 2.10 & 2.13 & 0.48 & 2.18 & 39.06 & 96.67 & 9.17 & 2.50 & 21.72 & 22.43 & 2.80 & 22.61 & 42.30 & 96.44 & 96.54 & 12.86 \\[1ex]				
						\multicolumn{18}{l}{\textbf{Random Forest}}\\
						Localized & Feasible &   2.13 & 2.16 & 0.48 & 2.21 & 40.05 & 96.85 & 9.29 & 1.26 & 21.51 & 22.30 & 2.81 & 22.47 & 42.55 & 96.42 & 95.55 & 13.76  \\ 
						&  Oracle &  2.10 & 2.13 & 0.43 & 2.17 & 39.26 & 96.75 & 9.18 & 2.45 & 21.10 & 21.97 & 2.96 & 22.17 & 42.61 & 96.18 & 93.80 & 15.33  \\ 
						Global  &  Feasible & 2.13 & 2.16 & 0.48 & 2.21 & 39.50 & 96.72 & 9.28 & 1.31 & 21.42 & 22.21 & 2.94 & 22.40 & 42.93 & 96.36 & 95.16 & 14.11 \\ 
						&  Oracle &  2.10 & 2.12 & 0.43 & 2.17 & 39.44 & 96.77 & 9.16 & 2.64 & 20.94 & 21.83 & 2.96 & 22.03 & 42.76 & 96.20 & 93.06 & 16.00  \\[1ex] 
						\multicolumn{18}{l}{\textbf{Boosted Trees}}\\ 
						Localized & Feasible  &   2.13 & 2.15 & 0.43 & 2.19 & 38.99 & 96.70 & 9.29 & 1.22 & 21.58 & 22.31 & 2.47 & 22.45 & 42.25 & 96.60 & 95.89 & 13.45           \\  
						& Oracle & 2.09 & 2.12 & 0.42 & 2.16 & 39.42 & 96.68 & 9.14 & 2.83 & 21.35 & 22.15 & 2.70 & 22.31 & 42.57 & 96.38 & 94.93 & 14.31 \\ 
						Global & Feasible  & 2.11 & 2.13 & 0.44 & 2.18 & 38.93 & 96.76 & 9.21 & 2.05 & 21.53 & 22.29 & 2.63 & 22.44 & 42.71 & 96.60 & 95.68 & 13.63      \\ 
						&  Oracle &  2.10 & 2.12 & 0.43 & 2.17 & 39.14 & 96.75 & 9.17 & 2.56 & 21.32 & 22.11 & 2.59 & 22.26 & 42.64 & 96.50 & 94.78 & 14.45  \\ [1ex]
						\multicolumn{18}{l}{\textbf{Post-lasso}}\\
						Localized &  Feasible &  2.11 & 2.14 & 0.45 & 2.18 & 38.88 & 96.80 & 9.24 & 1.80 & 21.85 & 22.58 & 2.37 & 22.71 & 41.82 & 96.44 & 97.07 & 12.38 \\  
						& Oracle & 2.09 & 2.12 & 0.46 & 2.17 & 39.56 & 96.71 & 9.12 & 2.99 & 21.72 & 22.49 & 2.49 & 22.63 & 42.33 & 96.40 & 96.57 & 12.83 \\ 
						Global & Feasible  & 2.11 & 2.13 & 0.45 & 2.18 & 38.67 & 96.81 & 9.22 & 1.96 & 21.84 & 22.58 & 2.41 & 22.71 & 42.15 & 96.56 & 97.05 & 12.40  \\ 
						& Oracle &             2.10 & 2.13 & 0.47 & 2.18 & 38.89 & 96.70 & 9.18 & 2.43 & 21.72 & 22.43 & 2.77 & 22.60 & 42.27 & 96.42 & 96.54 & 12.86\\[1ex] 			
						\multicolumn{18}{l}{\textbf{Flexible}}\\
						& Feasile & 2.11 & 2.13 & 0.44 & 2.18 & 39.03 & 96.69 & 9.19 & 2.28 & 21.35 & 22.10 & 2.65 & 22.25 & 42.25 & 96.48 & 94.87 & 14.37  \\ 
						& Oracle  & 2.09 & 2.12 & 0.44 & 2.16 & 39.37 & 96.73 & 9.14 & 2.81 & 20.94 & 21.82 & 2.96 & 22.02 & 42.75 & 96.26 & 93.06 & 16.00 \\
						\bottomrule\end{tabular}
					
					\begin{tablenotes}	
						
					\item\textit{Notes:} 
					Results are based on 10,000 Monte Carlo draws.  The left panel shows the results for the original outcome and the right panel for  age as the dependent variable based on \cite{londono2020upstream} and a sample size of $n=5000$ (see Section~\ref{sec:Simulations} for details).						
					The bandwidth is chosen and the confidence sets are constructed based on bias-aware inference. 
					The columns show the simulated mean standard error (Mean SE),   standard deviation (SD); simulated  bias (Bias);  root mean squared error (RMSE); average  bandwidth (Mean Bandwidth), 
					coverage of confidence intervals with 95\% nominal level (CI Cov); the average  confidence interval length (Mean CI Length); and the reduction in mean CI length relative to the no covariates CI length (Mean CI Length \% Red.). The estimators are described in Section~\ref{sec:Covariate_adj}. 
					\end{tablenotes}
			\end{threeparttable}}
		\end{table}
\end{landscape}}

\afterpage{\begin{landscape}
	\begin{table}
		\centering
		\caption{Full results for Simulation I with robust bias correction.}\label{table::appendix_full_simulationI_rdrobust}
		\resizebox{!}{0.44\textwidth}{
			\begin{threeparttable}
				\addtolength{\tabcolsep}{-0.1em}\begin{tabular}{p{1.5cm}p{1.8cm}@{\hspace{0.5cm}}p{0.9cm}p{0.9cm}p{0.9cm}p{0.9cm}p{0.9cm}p{0.9cm}p{1.3cm}p{1.3cm}@{\hspace{0.5cm}}p{0.9cm}p{0.9cm}p{0.9cm}p{0.9cm}p{0.9cm}p{0.9cm}p{1.3cm}p{1.3cm}}
					\toprule
					&& \multicolumn{8}{c}{Original Outcome } &  \multicolumn{8}{c}{Age }\\
					\cmidrule(r){3-10}\cmidrule(l){11-18}
					\multicolumn{2}{l}{Adjustment method} & Mean SE x100 & SD x100 & Bias x100  &  RMSE x100 & Mean Band-width & CI Cov in \% & Mean CI Length x100 &  Mean CI Length \% Red. 
					& Mean SE x100 & SD x100 &  Bias x100 &  RMSE x100 & Mean Band-width & CI Cov in \% & Mean CI Length x100 & Mean CI Length \% Red.  \\ 
					\midrule
					\multicolumn{2}{l}{\textbf{No Covariates}}  & 	2.72                  & 2.94 & 0.49 & 2.98 & 21.30 & 94.50 & 12.72 & 0.00 & 	32.28                  & 34.90 & -5.45 & 35.32 & 21.08 & 94.80 & 150.42 & 0.00   \\[1ex] 
					\multicolumn{2}{l}{\textbf{Conventional Linear}}   &   2.64                 & 2.89 & 0.58 & 2.95 & 21.15 & 93.97 & 12.35 & 2.91 & 27.68                 & 29.96 & -4.08 & 30.24 & 20.79 & 94.39 & 129.04 & 14.22  \\[1ex] 
					\multicolumn{18}{l}{\textbf{Linear Regression}}\\
					Localized & Feasible &   2.69                   & 2.90 & 0.61 & 2.96 & 21.32 & 94.09 & 12.56 & 1.27 &28.16                                             & 30.12 & -4.06 & 30.39 & 21.04 & 94.78 & 131.26 & 12.73 \\ 
					& Oracle & 2.67 & 2.87 & 0.60 & 2.94 & 21.30 & 94.21 & 12.45 & 2.13 & 	27.87                                    & 29.80 & -4.01 & 30.07 & 21.01 & 94.83 & 129.92 & 13.63 \\  
					Global & Feasible  & 	2.67                    & 2.89 & 0.62 & 2.95 & 21.31 & 94.08 & 12.48 & 1.86 &      28.03                                                & 29.92 & -3.75 & 30.16 & 21.05 & 94.95 & 130.70 & 13.11 \\
					& Oracle  & 2.67                & 2.88 & 0.61 & 2.94 & 21.30 & 94.22 & 12.45 & 2.11 &                  27.93     & 29.81 & -3.70 & 30.03 & 21.04 & 94.91 & 130.21 & 13.44 \\[1ex]				
					\multicolumn{18}{l}{\textbf{Random Forest}}\\
					Localized & Feasible &    2.71 & 2.92 & 0.61 & 2.99 & 21.33 & 94.30 & 12.65 & 0.55 & 		              27.91  & 29.92 & -3.64 & 30.14 & 21.03 & 94.99 & 130.11 & 13.50 \\ 
					&  Oracle & 	2.66                    & 2.86 & 0.53 & 2.91 & 21.29 & 94.23 & 12.42 & 2.35 & 26.96       & 28.87 & -3.97 & 29.14 & 21.01 & 94.82 & 125.65 \\ 
					Global  &  Feasible & 2.69                  & 2.90 & 0.60 & 2.96 & 21.33 & 94.25 & 12.57 & 1.15 &            27.59  & 29.58 & -3.48 & 29.78 & 21.05 & 94.80 & 128.62 & 14.50 \\ 
					&  Oracle &  2.66  & 2.87 & 0.53 & 2.91 & 21.29 & 94.25 & 12.42 & 2.36 & 		                26.89  & 28.77 & -3.58 & 28.99 & 21.03 & 94.82 & 125.33 & 16.68  \\[1ex] 
					\multicolumn{18}{l}{\textbf{Boosted Trees}}\\ 
					Localized & Feasible  & 	2.70                    & 2.92 & 0.51 & 2.96 & 21.28 & 94.43 & 12.61 & 0.87 & 	27.90           
					& 29.84 & -4.04 & 30.11 & 21.04 & 94.93 & 130.07 & 13.53 \\  
					& Oracle & 2.66               & 2.87 & 0.52 & 2.92 & 21.29 & 94.23 & 12.41 & 2.37 &       27.41                 & 29.31 & -4.13 & 29.59 & 21.03 & 94.98 & 127.77 & 15.06 \\ 
					Global & Feasible  &     2.67 & 2.88 & 0.54 & 2.93 & 21.30 & 94.34 & 12.48 & 1.88 &  27.69  & 29.67 & -3.76 & 29.90 & 21.05 & 94.84 & 129.08 & 14.18 \\ 
					&  Oracle & 	2.66                    & 2.87 & 0.54 & 2.92 & 21.29 & 94.20 & 12.42 & 2.30 & 	27.40                   & 29.25 & -3.75 & 29.48 & 21.03 & 94.97 & 127.75 & 15.07 \\ [1ex]
					\multicolumn{18}{l}{\textbf{Post-lasso}}\\
					Localized &  Feasible &  2.69                   & 2.91 & 0.52 & 2.96 & 21.31 & 94.57 & 12.58 & 1.08 & 28.12                 & 30.16 & -4.26 & 30.46 & 21.00 & 94.79 & 131.06 & 12.87 \\  
					& Oracle & 2.66 & 2.87 & 0.58 & 2.93 & 21.30 & 94.28 & 12.43 & 2.29 &  27.88 & 29.79 & -4.21 & 30.09 & 21.01 & 94.82 & 129.95 & 13.61 \\ 
					Global & Feasible  & 2.68       & 2.89 & 0.54 & 2.93 & 21.30 & 94.41 & 12.49 & 1.76 & 	28.06          & 29.99 & -4.14 & 30.27 & 21.03 & 94.80 & 130.83 & 13.02 \\ 
					& Oracle &             2.66               & 2.87 & 0.58 & 2.93 & 21.30 & 94.20 & 12.43 & 2.24 &         27.93             & 29.81 & -3.75 & 30.04 & 21.04 & 94.94 & 130.22 & 13.43\\[1ex] 			
					\multicolumn{18}{l}{\textbf{Flexible}}\\
					& Feasile & 0.56 & 2.88 & 0.56 & 2.93 & 21.29 & 94.25 & 12.44 & 2.16 &      27.44  & 29.39 & -3.76 & 29.63 & 21.02 & 94.92 & 127.91 & 14.96 \\ 
					& Oracle  & 2.66 & 2.87 & 0.54 & 2.92 & 21.29 & 94.24 & 12.41 & 2.42 & 	26.89  & 28.77 & -3.53 & 28.99 & 21.02 & 94.88 & 125.34 & 16.67 \\
					\bottomrule\end{tabular}
				
				\begin{tablenotes}	
					
					\item \textit{Notes:} 
					Results are based on 10,000 Monte Carlo draws.  The left panel shows results for the original outcome and the right panel for  age as the dependent variable based on \cite{londono2020upstream} and a sample size of $n=5000$ (see Section~\ref{sec:Simulations} for details).						
					The bandwidth is chosen and the confidence sets are constructed based on robust bias correction. 
					The columns show the simulated mean standard error (Mean SE),   standard deviation (SD); simulated  bias (Bias);  root mean squared error (RMSE); average  bandwidth (Mean Bandwidth), 
					coverage of confidence intervals with 95\% nominal level (CI Cov); the average  confidence interval length (Mean CI Length); and the reduction in mean CI length relative to the no covariates CI length (Mean CI Length \% Red.). The estimators are described in Section~\ref{sec:Covariate_adj}. 
				\end{tablenotes}
		\end{threeparttable}}
	\end{table}
\end{landscape}}

Figures~\ref{fig:boxplots_n2000} and~\ref{fig:boxplots_n5000} illustrate the asymptotic equivalence result in Theorem~\ref{th:Equivalence} of age as the dependent variable and bias-aware inference.\footnote{The oracle and the feasible estimates are even more similar when using the original outcome as dependent variable, as the covariates do not have much explanatory power in this case.  The results are also very similar when conducting inference based on robust bias correction.} Specifically, they show the difference between the simulated RD estimates based on the
feasible adjustments and oracle adjustments for sample sizes of 2000 and 5000. As a reference point, we also displayed the full
distribution of the no covariates RD estimates. As predicted by our theory, RD estimates based on
feasible and oracle adjustments are very close to each other especially compared to the distribution
of no covariates RD estimates. They even become more similar when the sample size increases.

\begin{figure}
	\centering
	\resizebox*{0.85\textwidth}{!}{
	\includegraphics{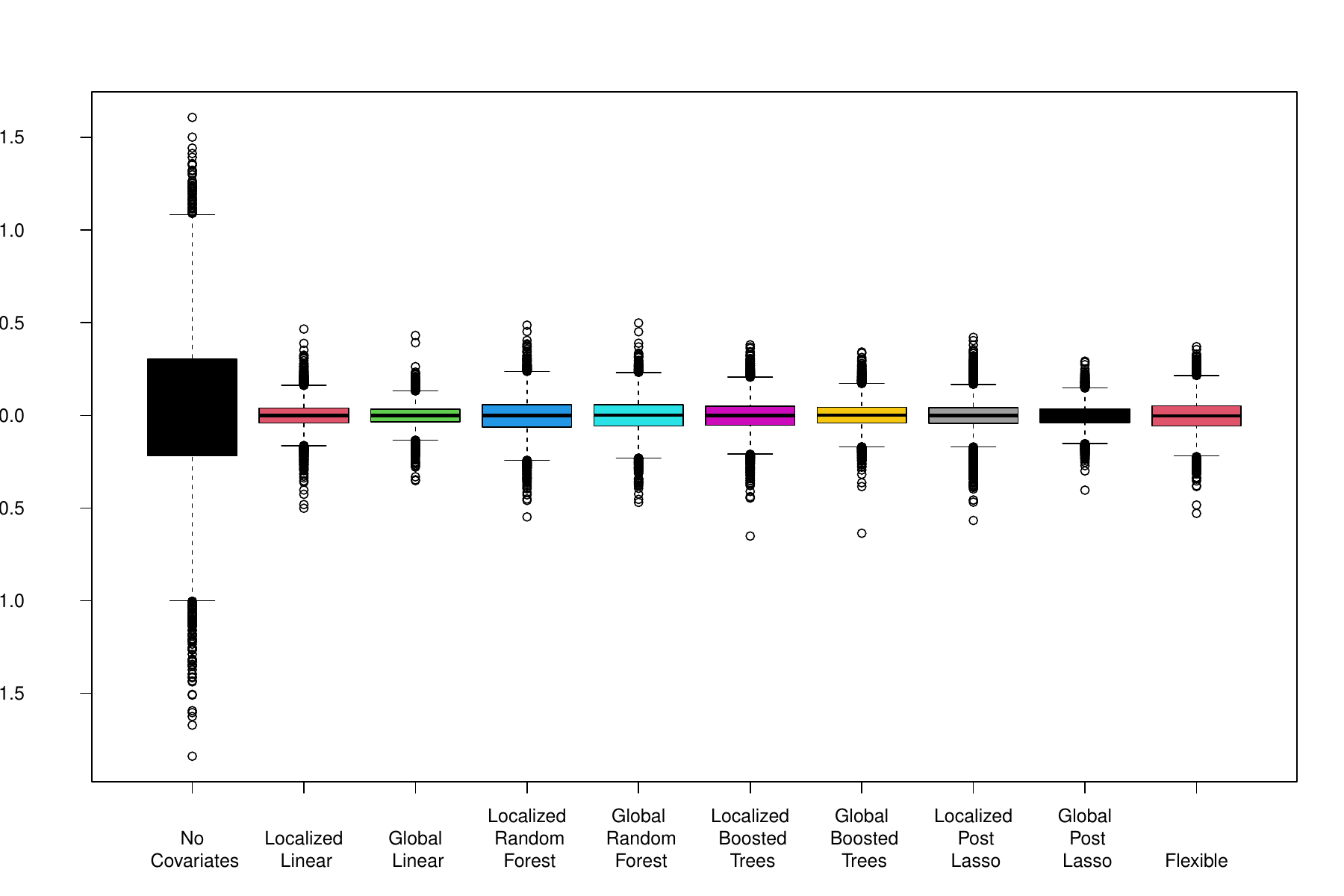}}
	\caption{Difference between cross-fitted feasible and oracle estimators for $n=2000$.}\label{fig:boxplots_n2000}
	\resizebox*{0.85\textwidth}{!}{	\includegraphics{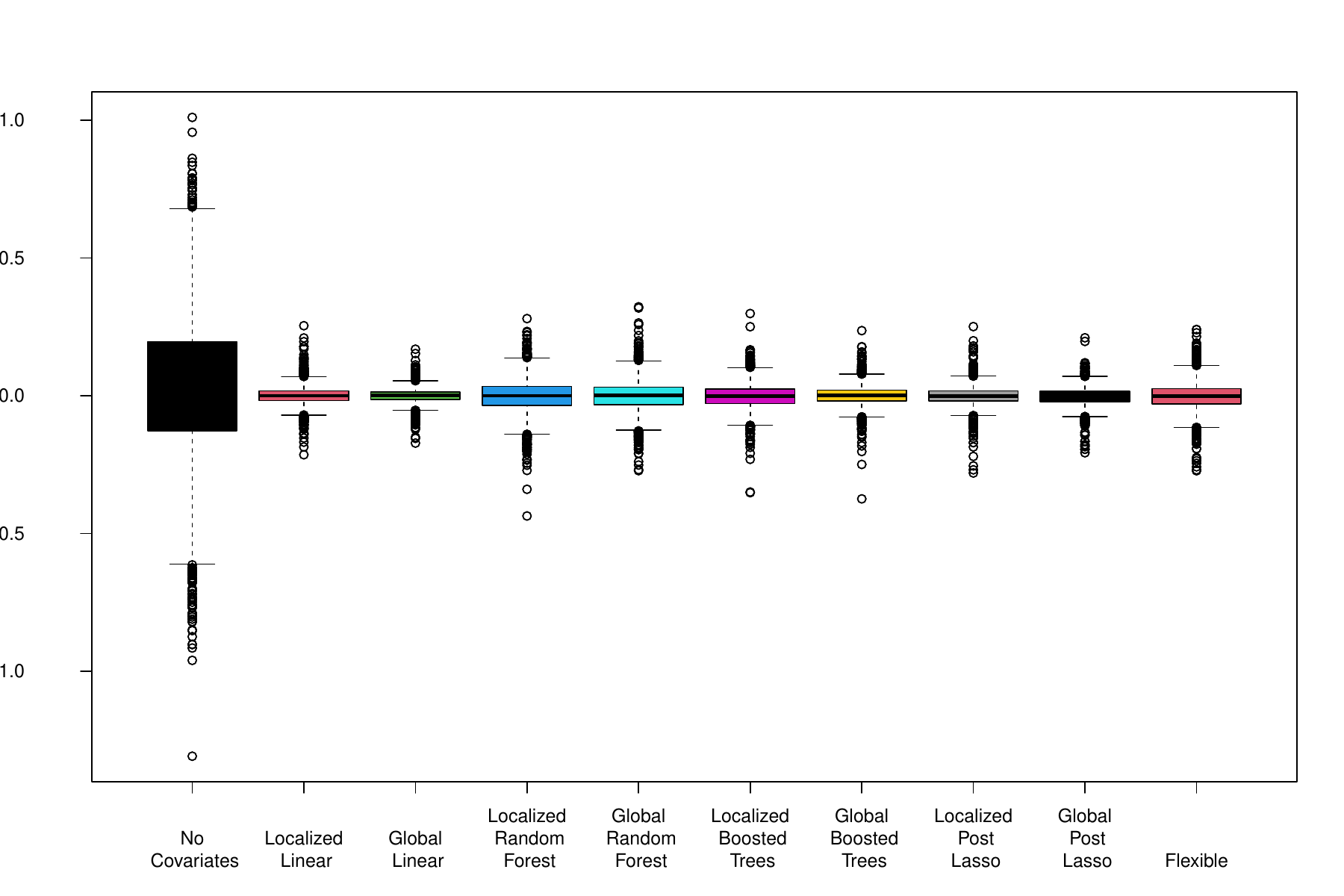}}
	\caption{Difference between cross-fitted feasible and oracle estimators for $n=5000$.}\label{fig:boxplots_n5000}
	
	\justifying
	\footnotesize{	\textit{Notes:} In each figure, the first box plot shows the distribution of the no covariates RD estimator and the other ones the difference of the cross-fitted feasible covariate-adjusted RD estimates and their respective oracle counterpart based on the respective adjustment methods. Simulations are based on \cite{londono2020upstream} and  age is the dependent variable.   See details for a description of the estimators in Section~\ref{sec:Covariate_adj} for details. Results are based on 10,000 Monte Carlo draws.
	}
\end{figure}

\subsection{Additional Results for Simulation II}
Figure~\ref{fig:appendix_many_n} extends the results presented in Figures~\ref{fig:sim_manyn_sd} and~\ref{fig:sim_manyn_inference}. It presents further simulation results for bias-aware inference for age as the dependent variable and for robust bias correction for both the original outcome and age as the dependent variables. The qualitative conclusions about the bias, the standard error, and the validity of confidence intervals are very similar to the ones discussed in Section~\ref{sec::simulations_many_cov}. In all cases, the simulated bias is insensitive to including many covariates, the inference based on the cross-fitted methods is valid, while the conventional linear adjustment leads to severely downward-biased standard errors and invalid confidence intervals as the number of covariates increases. The patterns in the standard deviations are different for the original outcome and age because of the different explanatory power of the covariates for these two dependent variables.

\begin{sidewaysfigure}
	\centering
	\includegraphics[scale=0.8]{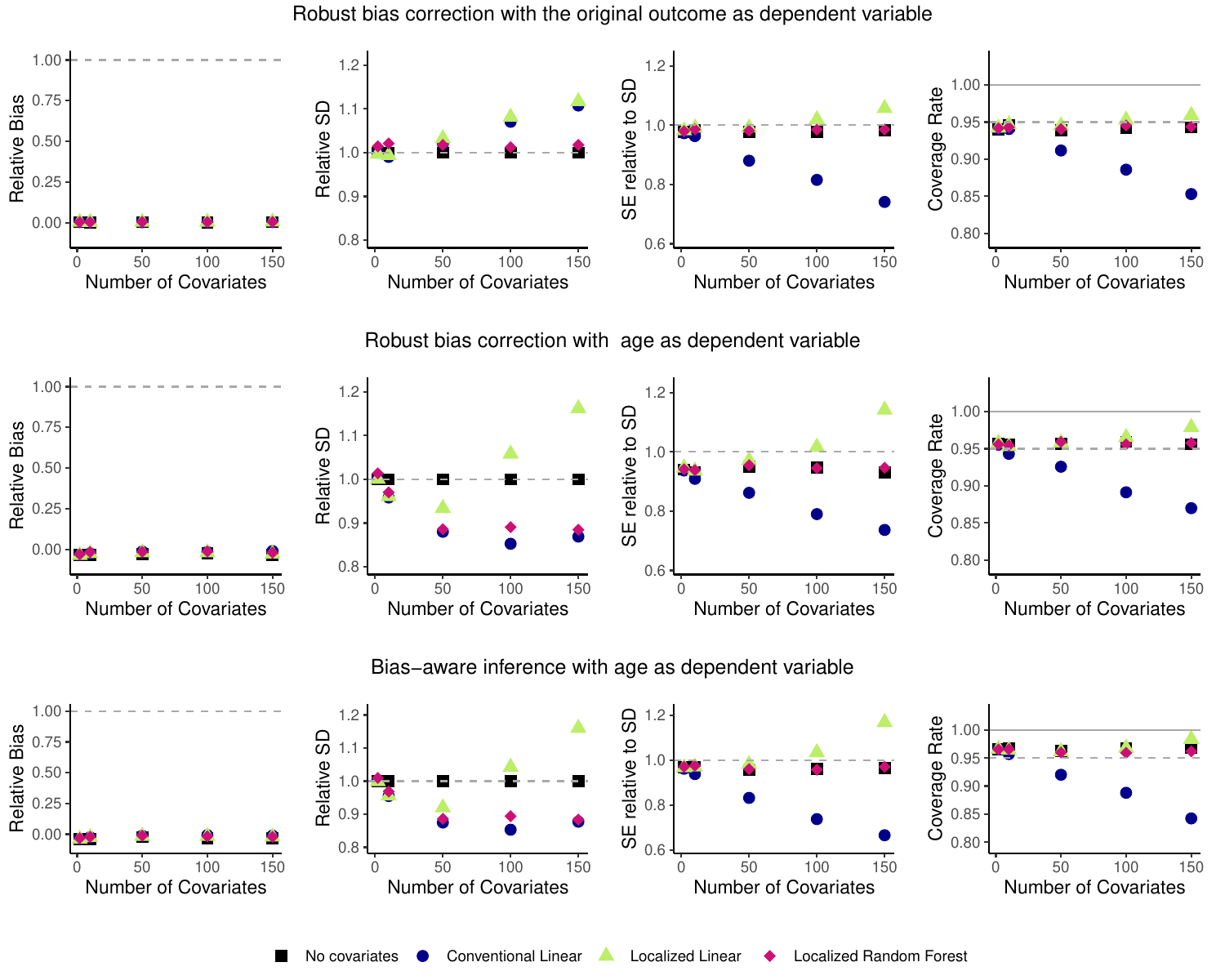}
	\caption{Additional results for Simulation II: Many covariates }\label{fig:appendix_many_n}
	
	\justifying
	\footnotesize{ \textit{Notes:} Simulation design is based  on \cite{londono2020upstream} and  sample sizes of $n=500$ for different numbers of covariates, see Section~\ref{sec::simulations_many_cov} for details. The three panels consider different combinations of the inference method and the dependent variable. In each row, the first and second graphs show  the bias and the standard deviation of the respective estimator relative to the standard deviation of the no covariates estimator.  The third graph shows the mean standard error of the respective estimator relative to its standard deviation. The last graph shows the simulated coverage of the confidence interval with 95\% nominal level. See Section~\ref{sec:Estimator}~and~\ref{sec:est_adj_function} for details  of the estimators. Results are based on 10,000 Monte Carlo draws.}
\end{sidewaysfigure}

\clearpage

\singlespacing
\bibliography{bibl}